\documentclass[aps,pra,11pt,letterpaper,onecolumn,tightenlines,superscriptaddress,notitlepage,longbibliography,nofootinbib]{revtex4-2}
\usepackage{qcircuit}
\usepackage[dvips]{graphicx}
\usepackage{amsmath,amssymb,amsthm,mathrsfs,amsfonts,dsfont}
\usepackage{mathtools}
\usepackage[caption=false]{subfig} %
\usepackage{ragged2e} %
\DeclareCaptionJustification{justified}{\justifying}
\usepackage{braket}
\usepackage{bm}
\usepackage{enumerate}
\usepackage{xcolor}
\usepackage{comment}
\usepackage{scalefnt}
\usepackage[colorlinks = true, 
citecolor = blue,
linkcolor = blue,
urlcolor = blue]{hyperref}
\usepackage{makecell,multirow}

\usepackage{pagecolor}
\usepackage[T1]{fontenc}
\usepackage{lmodern}
\usepackage[utf8]{inputenc}
\usepackage{microtype}

\usepackage[no-test-for-array]{nicematrix}
\usepackage{tikz}
\usepackage{tikzscale}
\usepackage{pgfplots}
\pgfplotsset{compat=newest}
\usetikzlibrary{plotmarks}
\usetikzlibrary{arrows.meta}
\usetikzlibrary{external}
\usepgfplotslibrary{patchplots}
\usepackage{grffile}
\usepackage{enumitem}
\usepackage{booktabs}

\DeclareMathOperator{\Tr}{Tr}

\newcommand{\od}[1]{^{(#1)}}

\newcommand{\coleq}{\mathrel{\mathop:}\nobreak\mkern-1.2mu=}
\newcommand{\eqcol}{\mkern-1.2mu=\mathrel{\mathop:}\nobreak}

\newcommand{\pt}{\mathrm{pt}}
\newcommand{\supp}{\mathrm{supp}}
\newcommand{\mf}{\mathfrak}
\newcommand{\mc}{\mathcal}

\newcommand{\mr}{\mathrm}
\newcommand{\msf}{\mathsf}
\newcommand{\mbb}{\mathbb}
\newcommand{\T}{{\sf T}}
\newcommand{\Pn}{{\sf P}^n}
\newcommand{\Cn}{{\sf C}^n}
\newcommand{\id}{\mathrm{id}}
\newcommand{\expval}[1]{{\langle #1 \rangle}}
\newcommand{\ketbra}[2]{{\vert #1 \rangle \langle #2 \vert}}
\newcommand{\lket}[1]{\vert #1 \rangle\!\rangle}
\newcommand{\lbra}[1]{\langle\!\langle #1 \vert}
\newcommand{\lbraket}[2]{\langle\!\langle #1 \vert #2 \rangle\!\rangle}
\newcommand{\lketbra}[2]{\vert #1 \rangle\!\rangle\langle\!\langle #2 \vert}

\newcommand{\imq}{{\mr{Im}{\mc Q}}}

\newtheorem{theorem}{Theorem}
\newtheorem{lemma}{Lemma}%
\numberwithin{lemma}{section}
\newtheorem{proposition}[lemma]{Proposition}%
\newtheorem{problem}{Problem}

\newtheorem{corollary}[lemma]{Corollary}%
\newtheorem{definition}[lemma]{Definition}

\renewcommand{\thelemma}{\arabic{section}.\arabic{lemma}}

\definecolor{applegreen}{rgb}{0.55, 0.71, 0.0}

\newcommand{\comments}[1]{}
\usepackage{algorithm}  
\usepackage[noend]{algpseudocode}  
\newcommand{\algorithmfootnote}[2][\footnotesize]{%
  \let\old@algocf@finish\@algocf@finish%
  \def\@algocf@finish{\old@algocf@finish%
    \leavevmode\rlap{\begin{minipage}{\linewidth}
    #1#2
    \end{minipage}}%
  }%
}
\usepackage{xparse}
\NewDocumentCommand{\LeftComment}{s m}{%
  \Statex \IfBooleanF{#1}{\hspace*{\ALG@thistlm}}\(\triangleright\) #2}
\algnewcommand{\LineComment}[1]{\Statex // #1}

\newcounter{algcounter}
\newcommand\algnum{\stepcounter{algcounter}\thealgcounter~}

\begin{document}

\title{Efficient self-consistent learning of gate set Pauli noise}
\author{Senrui Chen}
\affiliation{Pritzker School of Molecular Engineering, The University of Chicago, Chicago, Illinois 60637, USA}
\author{Zhihan Zhang}
\affiliation{Institute for Interdisciplinary Information Sciences, Tsinghua University, Beijing 100084, China}
\author{Liang Jiang} 
\affiliation{Pritzker School of Molecular Engineering, The University of Chicago, Chicago, Illinois 60637, USA}
\author{Steven T. Flammia}
\affiliation{Department of Computer Science, Virginia Tech, Alexandria, USA}
\affiliation{Phasecraft Inc., Washington DC, USA}
\date{\today}

\begin{abstract}

Understanding quantum noise is an essential step towards building practical quantum information processing systems.
Pauli noise is a useful model that has been widely applied in quantum benchmarking, error mitigation, and error correction.
Despite intensive study into Pauli noise learning, most existing works focus on learning Pauli channels associated with some specific gates rather than treating the gate set as a whole.
A learning algorithm that is self-consistent, complete, and efficient at the same time is yet to be established. 
In this work, we study the task of \emph{gate set Pauli noise learning}, where a set of quantum gates, state preparation, and measurements all suffer from unknown Pauli noise channels with a customized noise ansatz.
Using tools from algebraic graph theory, we analytically characterize the self-consistently learnable degrees of freedom for Pauli noise models with arbitrary linear ansatz, and design experiments to efficiently learn all the learnable information. %
Specifically, we show that all learnable information about the gate noise can be learned to relative precision, under mild assumptions on the noise ansatz.
We then demonstrate the flexibility of our theory by applying it to concrete physically motivated ansatzs (such as spatially local or quasi-local noise) and experimentally relevant gate sets (such as parallel CZ gates).
These results not only enhance the theoretical understanding of quantum noise learning, but also provide a feasible recipe for characterizing existing and near-future quantum computing devices.

\end{abstract}

\maketitle

\tableofcontents

\section{Introduction}

Quantum technology promises to bring significant advantages into various information processing tasks, including computation~\cite{nielsen2010quantum}, communication~\cite{gisin2007quantum,kimble2008quantum}, metrology~\cite{giovannetti2006quantum,giovannetti2011advances}, and more. 
Despite rapid experimental progress, noise remains as a major challenge in building a practical quantum information processing platform. 
In order to suppress, mitigate, and correct noise affecting a quantum device, it is essential to first gain knowledge about the noise. 
However, learning quantum noise is not easy. 
Noise is complicated, containing a large number of parameters that grows exponentially with the system size. 
Meanwhile, the available quantum control for noise learning is limited and can itself suffer from imperfection. 
Developing efficient, reliable, and informative quantum noise characterization protocols is thus an important problem and an active field of research~\cite{eisert2020quantum,proctor2024benchmarking,hashim2024practical}.

Among different noise characterization approaches, Pauli noise learning is a simple but powerful formalism that has gained increasing attention recently~\cite{Flammia2020,Erhard2019,Flammia2021,Harper2020,flammia2021averaged,wagner2023learning,van2023probabilistic,chen2023learnability,chen2022quantum,carignan2023error,van2023techniques,hockings2024scalable}. 
Here, the noise processes affecting different components (i.e., state preparation, measurements, gates) of a quantum device are modeled by Pauli channels, defined as a stochastic mixture of multi-qubit Pauli operations (see Eq.~\eqref{eq:Pauli_channel}). Pauli channels can naturally describe a wide spectrum of incoherent noises, including depolarizing and dephasing noises. 
It is also a widely used model in the study of quantum error correction and fault-tolerant quantum computing~\cite{chen2022calibrated}. 
Besides, there exist techniques known as randomized compiling~\cite{Knill2005,Wallman2016,hashim2020randomized} that, given mild assumptions about the noise and gate set, can tailor an arbitrary noise channel into a Pauli channel, justifying the validity of this noise model.

Despite intensive recent study, some core aspects of Pauli noise learning have not been well-understood. In particular, a learning algorithm that is simultaneously \emph{self-consistent}, \emph{comprehensive}, and \emph{efficient} is yet to be established. Here, self-consistency means using only operations that are noisy to learn about themselves, comprehensiveness means to learn all the information that is self-consistently learnable, and efficiency means the number of measurements and classical processing time scales favorable with the number of qubits. 
In this work, we build up a framework called \emph{gate set Pauli noise learning}, give analytic characterization of the self-consistently learnable degrees of freedom, and develop an algorithm with the above three desirable features. 

\begin{figure}[t]
    \centering
    \includegraphics[width=0.88\linewidth]{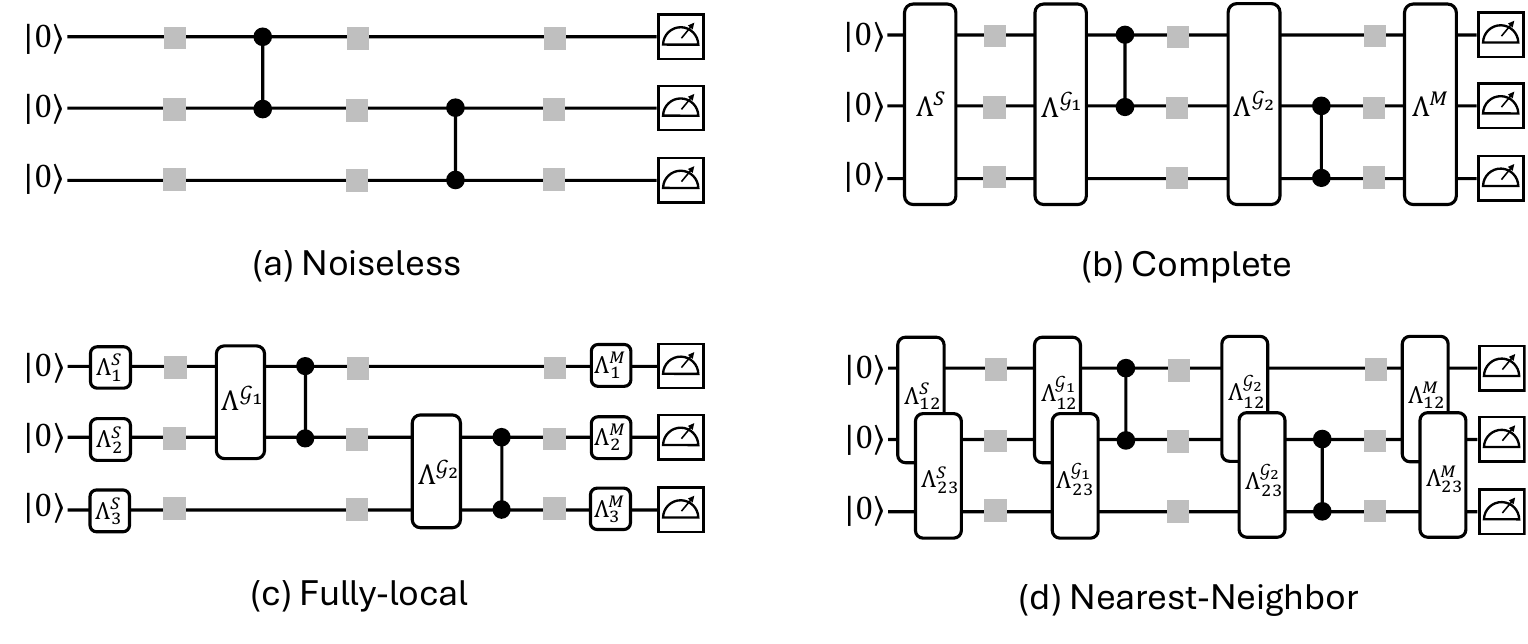}
    \caption{Examples of different noise ansatz on a 3-qubit system. The available set of multi-qubit gates are two CZ gates. The small square box represents single-qubit gates whose noise is ignored (or absorbed to multi-qubit gates and SPAM).
    (a) A noiseless circuit. (b) Circuit with complete Pauli noise, where the noise for SPAM and both gates are all generic $3$-qubit Pauli channels. (c) Circuit with fully local Pauli noise, where the SPAM noise channels are qubit-wise independent, while the gate noise channels has the same support as the gate. (d) Circuit with nearest-neighbor Pauli noise (a special case of quasi-local model), where each Pauli noise channel is nearest-neighbor $2$-local. See Sec.~\ref{sec:quasi} for a rigorous definition.}
    \label{fig:NoiseModel}
\end{figure}

\medskip
\noindent\textbf{Problem Setup.}
Consider an $n$-qubit quantum system. Suppose, by design, the following operations can be implemented on the platform: (1) Initialize the system to the ground state; (2) Apply a layer of $n$ single-qubit unitary gates; (3) Apply some multi-qubit Clifford gates gate from a finite set denoted by $ \mf G$; (4) Measure all qubits on the computational basis. We refer to the collection of all the above components as a \emph{gate set}, following the language of gate set tomography (GST)~\cite{nielsen2021gate}, although we are also including noisy SPAM (state preparation and measurement).
In practice, all the operations of a gate set can be noisy. Our goal is to characterize the noisy gate set \emph{self-consistently}. That is, using only operations from the gate set to learn about themselves.

Instead of dealing with the most general form of noise, we focus on what is known as the \emph{Pauli noise model}: We ignore the noise of any single-qubit gates, and assume the noise channels acting on the initial state, measurements, and any multi-qubit Clifford gate to be a (component-dependent) Pauli channel, which takes the following form,
\begin{equation}\label{eq:pauli_channel}
    \Lambda(\rho) = \sum_{a\in{\sf P}^n}p_aP_a\rho P_a = \sum_{b\in{\sf P}^n}\lambda_b P_b\Tr(P_b\rho)/2^n.
\end{equation}
Here, $\msf P^n \coleq \{I,X,Y,Z\}^{\otimes n}$ is the $n$-qubit Pauli group (modulo phase). $\{p_a\}$ are the \emph{Pauli error rates} while $\{\lambda_b\}$ are the \emph{Pauli fidelities} or \emph{Pauli eigenvalues}. We assume $\lambda_b>0$ throughout this work. Such a model can be ensured via randomized compiling under mild assumptions~\cite{Wallman2016,hashim2020randomized}.
The action of noise channel on the initial state $\rho_0$, POVM element $E_j$, and gate $\mc G\in\mf G$ is given by 
\begin{equation}
    \tilde{\rho}_0 = \Lambda^S(\rho_0),\quad \tilde E_j = \Lambda^M(E_j),\quad\tilde{\mc G}=\Lambda^{\mc G}\circ\mc G,
\end{equation}
respectively. A realization of the gate set Pauli noise model is specified by the involved Pauli channels $\{\Lambda^S,
\Lambda^M,\Lambda^{\mc G}:\mc G\in\mf G\}$. 
Since a generic Pauli channel contains exponentially many parameters, we often need to assume an efficient parametrization of the Pauli channels, which we call a \emph{noise ansatz}.
Examples of noise ansatzes include spatially local or quasi-local Pauli channels, as shown in  Fig.~\ref{fig:NoiseModel}(c) and (d), respectively.
We refer to a Pauli noise model with a noise ansatz as a \emph{reduced model}, and the one without a noise ansatz as a \emph{complete model}.
Throughout this work, we assume the true noise model is \emph{realizable}. 
That is, the true noise model can be represented exactly as a point in the model space. 

To characterize the noisy gate set means to learn all the Pauli noise channels with ansatz. However, it is generally impossible to fully determine those Pauli noise channels self-consistently~\cite{huang2022foundations,chen2023learnability}.
Indeed, consider the following \emph{gauge transformation}~\cite{nielsen2021gate} on the noisy operations, 
\begin{equation}\label{eq:gauge_transformation_intro}
    \tilde\rho_0\mapsto \mc D(\tilde\rho_0),\quad
    \tilde E_j\mapsto \mc D^{-1}(\tilde E_j),\quad
    \tilde{\mc G}\mapsto \mc D\tilde{\mc G}\mc D^{-1},\quad
    \bigotimes_{i=1}^n \mc U_i \mapsto \mc D\left(\bigotimes_{i=1}^n \mc U_i\right)\mc D^{-1},
\end{equation}
where $\mc D$ is some invertible linear map.
It is not hard to see that any experimental outcomes statistics remain unchanged under such transformation. Thus, two realizations of a noisy gate set related by a gauge transformation are \emph{indistinguishable}.
Meanwhile, one can carefully choose $\mc D$ (e.g., as a depolarizing channel on a subset of qubits) such that the transformed noisy gate set will remain a valid Pauli noise model within ansatz. That is, the single-qubit gates $\bigotimes_{i}\mc U_i$ remain noiseless, while $\Lambda^S,\Lambda^M,\Lambda^{\mc G}$ remain Pauli channels within the ansatz (possibly with different parameters). 

Now, consider a function of the gate set noise parameters. If the function value changes under certain valid gauge transformations, then the function is not self-consistently learnable. 
This raises the first question of gate set Pauli noise learning as follows,
\begin{problem}\label{prob:1}
How to classify the learnable and unlearnable functions of a gate set Pauli noise model? 
\end{problem}
\noindent This question is partially answered in~\cite{chen2023learnability} for the case that there is no constraint for the Pauli noise channels. When there are efficient ansatzes for the noise channels, the problem has remained open. After understanding what are the learnable degrees of freedom, the next natural problem is to find an efficient learning protocol, which can be summarized as,
\begin{problem}\label{prob:2}
    How to efficiently learn all the learnable information of a get set Pauli noise model?
\end{problem}
\noindent In this work, we will establish a generic theoretic framework to address the above two problems. To demonstrate the power and flexibility of our framework, we will apply it to Pauli noise models with concrete gate set and noise ansatzes that are considered in the literature (e.g.~\cite{flammia2021averaged,van2023probabilistic}) and resolve open problems therein. A detailed summary of our contributions is provided in the next section.

\subsection{Summary of Results}

\noindent\textbf{Learnability of gate set Pauli noise.}
Our approach to characterize the learnability (i.e. classification of learnable and unlearnable information) is to first encode all the noise parameters into a linear space. Let us start with the complete Pauli noise model (i.e. the one without a noise ansatz). Similar results have been obtained in~\cite{chen2023learnability} for this case, but here we will use a more formal way that can be extended to the reduced model later.

Define the logarithmic Pauli eigenvalues as $x_b\coleq -\log\lambda_b$. We collect all logarithmic Pauli eigenvalues (excluding the trivial ones corresponding to identity operator) from all Pauli noise channels into a ordered list denoted by $\bm x$, which can be viewed as a vector of a real linear space according to a standard basis. We denote the space $\bm x$ lives on as $X$, called the \emph{parameter space}. An \emph{experiment} is defined as a mapping $X\mapsto\mbb R^{2^n}$ specified by a sequence of gates $\mc C$. It maps $\bm x$ to a measurement outcome probability distribution given by the Born rule, $\Pr(j)=\Tr[\tilde E_j\tilde{\mc C}(\tilde{\rho}_0)]$, where the noisy realization of state, gates, and measurement is determined by $\bm x$. Now, we define a linear function $\bm f\in \mc L(X\mapsto\mbb R)\eqcol X'$ to be \emph{learnable} if its value can be completely specified by a set of experiments\footnote{One may wonder why we only consider \emph{linear} functions of the logarithmic Pauli eigenvalues. This is because any experiments can be determined by some learnable linear functions, so they suffices to characterize the learnability. See Theorem~\ref{th:linearOb} for more details.}; 
We also define a vector $\mf d\in X$ to be a \emph{gauge} vector if the two lists of noise parameters $\bm x$ and $\bm x+\mf d$ cannot be distinguished by any experiment, for all $\bm x\in X$. The linear subspace of all learnable linear function and all gauge vectors are denoted by $L$ and $T$, respectively. The problem then becomes characterizing $X$, $L$, and $T$.

For this purpose, we introduce a directed graph called the Pauli pattern transfer graph (PTG). This graph is first defined in~\cite{chen2023learnability} with some twist needed to fit into the current framework. The graph consists of $2^n$ nodes. One of which is a root node (also called a SPAM node), the others correspond to the $2^n-1$ non-zero $n$-bit strings, each denotes a Pauli pattern (the pattern of a $n$-qubit Pauli operator can be understood as its support); 
For each logarithmic Pauli eigenvalues, we assign a unique edge to the PTG: For gate noise parameters, each edge describe how a Clifford gate transform certain Pauli operators (where we only track the pattern); For each state preparation (or, measurement) noise parameter, we assign an edge from the root node to certain pattern node (or, from a pattern node to root). An example of PTG is given in Fig.~\ref{fig:cz_ptg}.

\begin{figure}[!thp]
    \centering
    \includegraphics[width=0.8\linewidth]{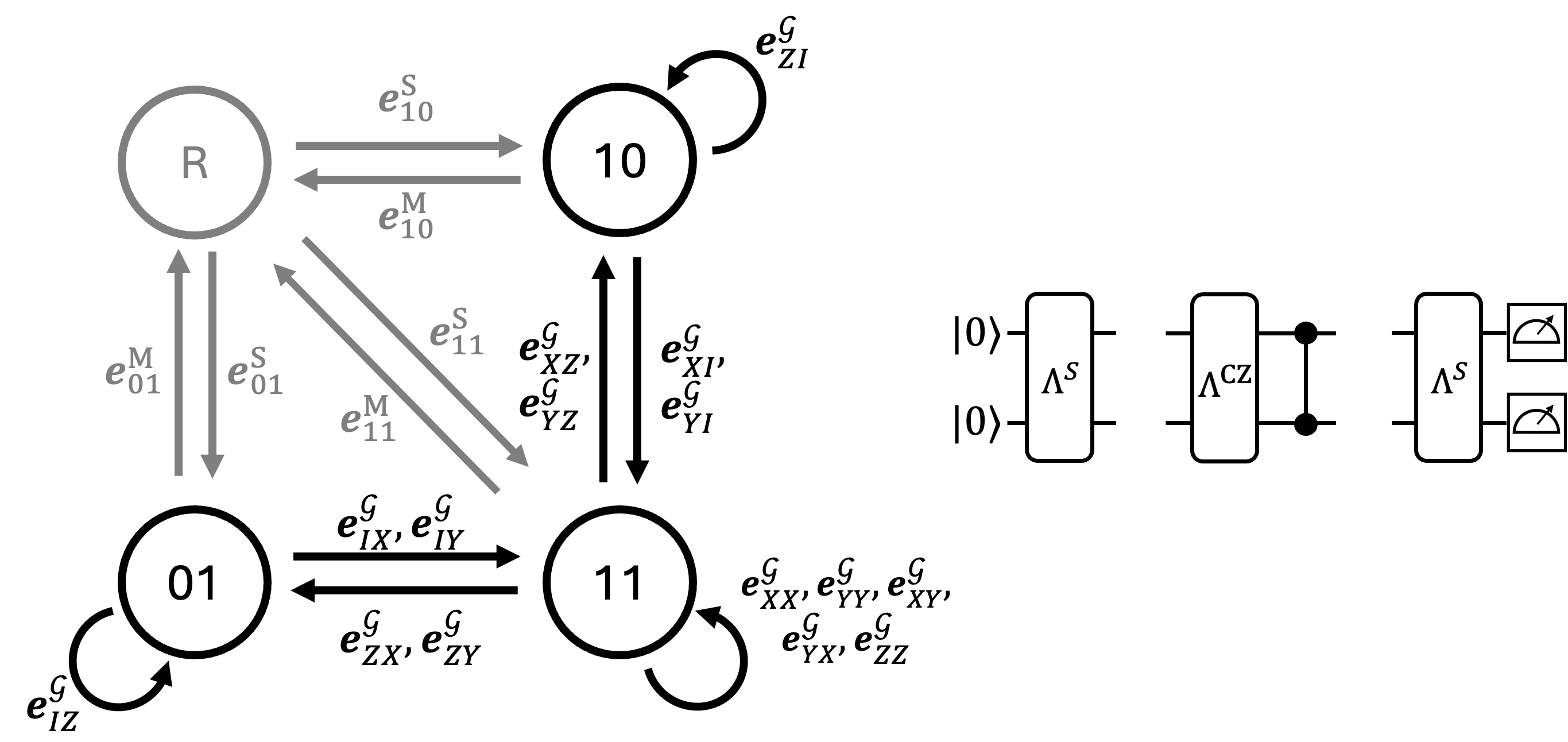}
    \caption{The pattern transfer graph (PTG) for a two-qubit gate set $\mf G=\{\mr{CZ}\}$. The gate set is depicted at the right of the graph. Here we define $\mc G=\mr{CZ}$ for notational simplicity. The black edges corresponds to gate noise parameters, while the gray edges corresponds to SPAM noise parameters. Multiple edges are shown as one edge with multiple labels.}
    \label{fig:cz_ptg}
\end{figure}

PTG contains all information about a complete Pauli noise model.
Intuitively, a path starting from and ending at the root node gives a combination of logarithmic Pauli eigenvalues that can be learned from an experiment.
This intuition can be made rigorous using tools from algebraic graph theory~\cite{bollobas1998modern,gleiss2003circuit}. Informally, let the edge space $E$ to be the real linear space spanned by all edges of PTG. For each cycle (or, cut) on PTG, we can assign it with a vector of $E$ in a proper way (see Sec.~\ref{sec:learnability_complete}). The subspace spanned by all cycle (or, cut) vectors is defined as the \emph{cycle space} $Z$ (or, the \emph{cut space} $U$). It is known that $Z$ and $U$ are orthogonal complements to each other within $E$.
Our main theorem is that, up to the natural identification between the parameter space $X$ and the edge space $E$, the learnable space $L$ is identical to the cycle space $Z$, while the gauge space $T$ is identical to the cut space $U$. This implies that $L$ and $T$ are also orthogonal complements to each other within $X$.\footnote{Rigorously speaking, $L$ is a subspace of the dual space $X'$. Here we identify $X$ with $X'$ using the standard basis for defining $X$.}

Next, we turn to study the learnability of a reduced model. A reduced model is specified by a tuple $(X_R,\mc Q)$. Here $X_R$ is the \emph{reduced parameter space}, a linear space whose components are viewed as lists of the reduced noise parameters (denoted by $\bm r$). $\mc Q:X_R\mapsto X$ is the \emph{embedding map} describing how the reduced parameters specifies the Pauli channels. Specifically, the logarithmic Pauli eigenvalues are given by $\bm x = \mc Q(\bm r)$.
We will assume $\mc Q$ is linear and injective. The notion of experiments on reduced model can be naturally induced from the definition for complete model. We can also similarly define the reduced learnable space $L_R$ (learnable functions of $X_R$) and the reduced gauge space $T_R$ (indistinguishable transformations within $X_R$). The problem is then to characterize $X_R$, $L_R$, $T_R$. The challenge here is that, it is not easy to construct a graph describing the reduced noise parameters as in the complete model. Our approach is to connect the reduced spaces back to the complete spaces, the use the characterization that we have already obtained.
We proved that $L_R = \mc Q^{\T}(L)$,~$T_R = \mc Q^{-1}(T)$. In words, the reduced learnable space is the complete learnable space projected by $\mc Q^\T$, while the reduced gauge space is the preimage of the complete gauge space via $\mc Q$.

The relation between all the linear spaces introduced above can be summarized as follows.

\begin{theorem}[Learnability, informal]\label{th:learnability_all}
For any reduced model $(X_R,\mc Q)$ such that $\mc Q$ is linear and injective,
    \vspace{1mm}

    \begin{center}
    \begin{tabular}{c c c c c c}
        \makecell{ Reduced model:\\{}}&\makecell{$\mathop{X_R}$\\{\scriptsize(reduced parameter)}} & \makecell{$=$\\{}} & \makecell{$\mathop{L_R}$\\{\scriptsize(reduced learnable)}} & \makecell{$\oplus^{\perp}$\\{}} & \makecell{$\mathop{T_R}$\\{\scriptsize(reduced gauge)}}\\
        
        &\makecell{$~~\downarrow_\mc Q$}&&\makecell{$\quad\uparrow_{\mc Q^\T}$}&&\makecell{$\quad~~\uparrow_{\mc Q^{-1}}$}\\
        
        \makecell{ Complete model:\\{}}&\makecell{$\mathop{X}$\\{\scriptsize(parameter)}} & \makecell{$=$\\{}} & \makecell{$\mathop{L}$\\{\scriptsize(learnable)}} & \makecell{$\oplus^{\perp}$\\{}} & \makecell{$\mathop{T}$\\{\scriptsize(gauge)}}\\

        &$\parallel$ && $\parallel$ && $\parallel$\\

        \makecell{ Pattern transfer graph:\\{}}&\makecell{$\mathop{E}$\\{\scriptsize(edge)}} & \makecell{$=$\\{}} & \makecell{$\mathop{Z}$\\{\scriptsize(cycle)}} & \makecell{$\oplus^{\perp}$\\{}} & \makecell{$\mathop{U}$\\{\scriptsize(cut)}}\\
    \end{tabular}.
    \end{center}
    The equality is up to the natural identification between $X$, $X'$, and $E$. $\bigoplus^\perp$ means orthogonal complements.
\end{theorem}

\noindent Theorem~\ref{th:learnability_all} resolved Problem~\ref{prob:1}, giving a complete characterization about the learnable degrees of freedom for a Pauli noise model with general gate sets and noise ansatzes.
They also serve as the basis for developing efficient self-consistent learning algorithms. Rigorous statements and detailed proofs for Theorem~\ref{th:learnability_all} are presented in Sec.~\ref{sec:framework}.

\medskip
\noindent\textbf{Efficient learning algorithms for gate set Pauli noise.}
    We now consider Problem~\ref{prob:2}, how to efficiently and self-consistently learn a Pauli noise model. 
    Thanks to Theorem~\ref{th:learnability_all}, to learn a reduced Pauli noise model means to learn every functions of the reduced learnable space $L_R$. It suffices to find a basis for $L_R$ and learn the value of each basis function, as all the other functions can be decided by linearity. This can be done by first finding certain cycle basis for $E=L$
    with each basis function learnable from some experiment, and then use the fact $L_R = \mc Q^\T(L)$ to pick up a subset of functions to form a basis for $L_R$.
    We give such a construction with the additional feature that each experiments uses at most one noisy gate, summarized as follows.
    \begin{theorem}[Simple experiment construction, informal]\label{th:exp_simple}
        There exists a set of $M=\dim(L_R)$ experiments that can determine the value of any functions of $L_R$. Moreover, any experiment from this set either apply one gate from $\mf G$ exactly once, or apply no gates from $\mf G$ at all.
    \end{theorem}
    \noindent 
    To show this, we first show that any cycle containing the root node exactly once (called a \emph{rooted cycle}) can by learned by one experiment. Then notice that, thanks to the connectivity of PTG, there exists a rooted cycle basis where every cycle is of length 2 or 3, corresponding to experiment with no gate or exactly one gate from $\mf G$.
    For an efficient noise ansatz, the number of reduced parameters should be at most polynomial in $n$, which means the number of experiments $M$ needed in Theorem~\ref{th:exp_simple} is also polynomial. The protocol can thus be efficiently executed. Details about this protocol are presented in Sec.~\ref{sec:learning_simple}.

    \medskip
    
    In the literature of quantum noise learning, an important consideration is whether the noise parameter can be learned to relative precision. That is, given a sufficiently small infidelity parameter $r$ of a noisy gate $\tilde{\mc G}$, can we obtain an estimator $\hat r$ such that $|\hat r - r|\le O(\varepsilon r)$ with high probability, by making a number of measurement that scales as $\varepsilon^{-2}\mr{polylog}(r^{-1})$? It is known that, if one can concatenate $m$ copies of $\tilde{\mc G}$ to amplify the infidelity to order $r^m$ for any $m$, this scaling can be achieved~\cite[Theorem 1]{harper2019statistical}.
    We should that such noise amplification is indeed possible for gate set Pauli noise model. We have that,
    \begin{theorem}[Amplifying gate noise parameters, informal]\label{th:amplify}
        For any function that is a cycle vector and only supported on gate noise parameters, let $\lambda \coleq \exp(-\bm f(\bm x))$.
        Then, there exists a family of experiments $\{F_m\}_{m\in\mbb N}$ and a Pauli observable, such that the expectation value of the $m$th experiment is given by $A\lambda^m$, where $A$ is an $m$-independent parameter.
    \end{theorem}
    \noindent Theorem~\ref{th:amplify} ensures one can achieve relative precision for a basis of learnable functions supported on the gate noise parameters, confirming that gate set Pauli noise learning has the same desirable feature as other quantum noise learning schemes such as randomized benchmarking~\cite{Gambetta2012}. Details about Theorem~\ref{th:amplify} are presented in Sec.~\ref{sec:learning_relative}.

    \medskip

    Theorem~\ref{th:exp_simple} and~\ref{th:amplify} address Problem~\ref{prob:2} from two different aspects, providing efficient experiment construction for gate set Pauli noise learning. 
    The key is to find an appropriate basis for the reduced learnable space. One caveat is that, since the numbers of edges and nodes in PTG grow exponentially with the number of qubits $n$, an algorithm that first finds a cycle basis in PTG and then select a subset to construct a basis for $L_R$ will inevitably suffer from exponential run time. 
    However, when the noise ansatzes have some structure, we can find such a basis analytically, thus circumventing the run time issue.
    In what follows, we demonstrate this point by applying our theory to concrete physically relevant noise ansatzes.

\medskip
\noindent\textbf{Applications to spatially-local noise ansatzes.}
We now apply the general theory of learnability and learning algorithms to concrete noise ansatzes. We will focus on spatially local or quasi-local noise model as an example. Such noise models have been widely studied in the literature of Pauli noise learning~\cite{Flammia2020,flammia2021averaged,hockings2024scalable,wagner2023learning,van2023probabilistic,carignan2023error,pelaez2024average}, but our results for the first time give a clear classification of the learnable and gauge degrees of freedom within such models.

    We first consider the fully-local model,
    where the SPAM noise channels are assumed to be qubit-wise product channels, while the gate noise channels act non-trivially only within the gate's support. 
    This model is also sometimes known a crosstalk-free model and have been considered in, e.g., ~\cite{flammia2021averaged,hockings2024scalable,pelaez2024average}. 
    To understand the learnable and gauge spaces for this model, we introduce a family of gauge vectors called the \emph{subsystem depolarizing gauges} (SDG), denoted by $\{\mf d_\nu\}_{\nu\in 2^{[n]}}$, where $\mf d_\nu$ represents a gauge transformation as in Eq.~\eqref{eq:gauge_transformation_intro} with $\mc D$ chosen as a partially depolarizing channel acting on a subsystem of qubits $\nu$.
    One can show the SDGs forms a basis for the complete cut space $T$. 
    Interestingly, we show that the subset of SDG's acting on one qubits fully characterize the reduced gauge space $T_R$ for a fully-local model.
    \begin{theorem}\label{th:fully_simplified}
    For any fully-local noise model $(X_R,\mc Q)$, its reduced gauge space satisfies,
    \begin{equation}
        \mc Q(T_R) = \mr{span}\{\mf d_s:|s|=1\}.
    \end{equation}
    \end{theorem}
    \noindent 
    Since $\mc Q$ is injective, this allows one to uniquely determine $T_R$.
    Intuitively, any multi-qubit SDG transformation will violate the locality assumption of qubit-wise independent SPAM noise (and sometimes also local gate noise), and is thus an invalid gauge transformation. 
    We also studied the more involved cases where either SPAM or gates has a fully-local noise assumption while the other does not. Even in that case, we can show that certain subset of SDG's still spanned $\mc Q(T_R)$. Detailed are presented in Sec.~\ref{sec:fully_learnability}.

    Regarding learning protocols, the reduced learnable space $L_R$ and a set of $\dim(L_R)$ experiments can be analytically written down. We further study the task of relative precision learning, which boils down to finding a cycle basis for the reduced gate noise parameters. Though we believe this is a hard problem in general, we give a computationally efficient procedure for the case where all Clifford gate from $\mf G$ acts non-trivially on two qubits, a practical setting for most existing quantum computing architectures (e.g., \cite{ibmquantum}). These issues are addressed in Sec.~\ref{sec:local_learning} and~\ref{sec:local_relative}.

    \medskip

    Going beyond the fully local model, we apply our theory to a quasi-local noise model that describes certain degrees of spatial correlation while remains parameter-efficient. 
    The quasi-local Pauli noise channels we consider are those whose Pauli eigenvalues factorize according to a factor graph~\cite{mao2005factor}.
    We first introduce some notations that are needed for defining the quasi-local noise model,
    defined as follows
    \begin{definition}\label{de:quasi_intro}
    An $n$-qubit Pauli channel $\Lambda$ is $\Omega$-local if its Pauli eigenvalues satisfies
    \begin{equation}
        \lambda_a = \prod_{b\triangleleft a,b\sim \Omega}\exp(-r_b),\quad \forall a\neq I_n,
    \end{equation}
    for some $\{r_b\in\mbb R:b\sim\Omega, b\neq I_n\}$ and a factor set $\Omega\subseteq 2^{[n]}$.
    \end{definition}
    \noindent Here, $\Omega\subseteq 2^{[n]}$ is a \emph{factor set} if it satisfies $\forall \nu\in\Omega$, $\forall \varnothing\subsetneq \mu\subseteq \nu$, $\mu\in\Omega$.
    We write $b\triangleleft a$ if $b_i\neq I\Rightarrow a_i = b_i$ for all $i\in[n]$. We write $b\sim\Omega$ if the support of $b$ belongs to $\Omega$.
    We show that, under certain ``covariance'' assumptions of the quasi-local model, the reduced gauge space is again characterized by a subset of SDGs that are consistent with the channel factorization,
    \begin{theorem}(Informal)\label{th:quasi_simplified}
    For any quasi-local noise model $(X_R,\mc Q)$, given that all Pauli channels are $\Omega$-local and is covariant for all gate in $\mc G$, then its reduced gauge space satisfies,
    \begin{equation}
        \mc Q(T_R) = \mr{span}\{\mf d_\nu:\nu\in\Omega\}.
    \end{equation}
    \end{theorem}
    \noindent Rigorous statements and detailed proofs are presented in Sec.~\ref{sec:quasi}, where we also review useful properties of such noise model and discuss its connection and distinction with existing models studied in the literature~\cite{Flammia2020,van2023probabilistic,wagner2023learning}.
    A cycle basis for the reduced learnable space $L_R$ is also analytically constructed for this case.

    \medskip

    Finally, in Sec.~\ref{sec:CZ}, we analyze a concrete gate set where the Control-Z (CZ) gates are the only multi-qubit entangling gates. 
    We characterize the learnability of this noise model with more than one noise ansatzes (which are directly related to those that have been studied in the literature~\cite{van2023probabilistic}) and give efficient learning algorithms in each case.
    These examples serve as an explicit recipe of gate set Pauli noise learning that is ready to be applied in current quantum devices.

\subsection{Prior works}\label{sec:discussion}

Learning the Pauli noise channel $\Lambda$ affecting a Clifford gate $\mc G$ (i.e., $\tilde{\mathcal G}=\Lambda\circ\mathcal G$) in the presence of unknown state-preparation-and-measurement (SPAM) noise is a standard task studied in multiple references~\cite{Flammia2020,Flammia2021,Harper2020,flammia2021averaged,van2023probabilistic,chen2023learnability,carignan2023error,van2023techniques,hockings2024scalable}.
When the ideal Clifford gate $\mc G$ equals to the identity, it has been shown that $\Lambda$ can be fully determined independent of the SPAM noise~\cite{Flammia2020};
When $\mc G$ is a non-trivial Clifford gate, it is known that certain properties of $\Lambda$ (e.g., the average fidelity) can be learned independent of SPAM~\cite{Erhard2019}, but no existing protocols can completely and SPAM-robustly learn $\Lambda$ in this case.
In fact, it was pointed out~\cite{chen2023learnability,huang2022foundations} that this is a fundamental limitation related to the model non-identifiability, also known as gauge ambiguity in the literature on gate set tomography~\cite{nielsen2021gate}. 
Specifically, \cite{chen2023learnability} have developed a theory that completely characterize the degrees of freedom of a Pauli noise channel associated to any set of Clifford gates and provides a “gauge-consistent” algorithm to learn all the learnable degrees of freedom. 
The main limitation in \cite{chen2023learnability} is that it considers the case that each $n$-qubit gate has a generic $n$-qubit Pauli noise channel, which has exponentially many parameters, and consequently the learning algorithm would take exponential time. 
In contrast, efficient parameterization of Pauli noise channels have been studied in the literature and applied in experiments~\cite{flammia2021averaged,van2023probabilistic,carignan2023error,wagner2023learning,hockings2024scalable,pelaez2024average}, but a good understanding about the gauge and learnable degrees of freedom in such models is yet to be obtained, which hinders the application of these protocols (specifically, how to learning such models self-consistently). 
Our work combines these two lines of research to achieve a self-consistent and efficient Pauli noise learning algorithm.

Our framework resembles ACES (i.e., averaged circuit eigenvalue sampling), a Pauli noise learning protocol proposed in \cite{flammia2021averaged} and further analyzed and implemented in \cite{cisneros2024average} and \cite{hockings2024scalable}. 
Similar to our setting, ACES models every different Clifford gate as having a gate-dependent Pauli noise channel (with an efficient noise ansatz such as fully local noise), and aims at learning every Pauli eigenvalue of every Pauli noise channel. 
ACES constructs linear equations of the form $\bm b = A \bm x$, where $\bm x$ are the logarithmic Pauli eigenvalues, $\bm b$ are the logarithmic circuit eigenvalues which can be experimentally measured, and $A$ is the design matrix determined by a set of Clifford circuits and Pauli measurements. 
The original ACES effectively fixes a gauge by ignoring state preparation noise, and is thus able to construct a full column-rank $A$ using random Clifford gates.
Within ACES's framework, our work basically says that if SPAM noise parameters are included in $\bm x$, there is a maximum possible rank of $A$ equal to the number of learnable degrees of freedom, and we provide an explicit recipe for constructing such an $A$. 
Moreover, the explicit construction can learn gate parameters to relative precision; see Theorem~\ref{th:relative}.
It is interesting to extend the statistical analysis for ACES~\cite{hockings2024scalable} to our setting. 
It is also interesting to study optimal learning and experimental design within our framework~\cite{van2023techniques,ostrove2023near}. 
Conversely, it is interesting to study the ultimate sample complexity limit for learning gate set Pauli noise models. 
Such bounds for learning Pauli channels in a black-box setting have been previously investigated~\cite{chen2022quantum,chen2024tight,fawzi2023lower,chen2023efficient,trinh2024adaptivity}.

\subsection{Future directions}

\noindent\textbf{Quantum error mitigation.} One promising practical application of our theory is in quantum error mitigation (QEM)~\cite{cai2023quantum}.
There have been QEM protocols based on the Pauli noise models~\cite{van2023probabilistic,ferracin2022efficiently,kim2023evidence,chen2020robust} which work by first learning the Pauli noise channels and then either canceling the noise via quasi-probability sampling, or amplifying the noise and then extrapolating to the noiseless limit. 
Due to the existence of gauge, these protocols generally make assumptions that are not physically justified (e.g., ``symmetry assumption''~\cite{van2023probabilistic,van2023techniques} or perfect state preparation~\cite{chen2020robust}) in order to determine the noisy Pauli channels.
However, it is known that the issue of gauge ambiguity should not hinder self-consistent QEM if assisted by gate set tomography~\cite{endo2018practical,suzuki2022quantum}, although the latter method is not easily scalable.
Using our framework, by efficiently learning all the learnable degrees of freedom and anchoring at an arbitrary point of the gauge, one will be able to mitigate the Pauli noise without any additional assumptions. 
Such a self-consistent and efficient QEM protocol is actively being developed by some of the current authors~\cite{chen2025disambiguating}.

\medskip
\noindent\textbf{Beyond Pauli noise model.} We believe that our framework and graph-theoretic analysis can be extended beyond the Pauli noise model.
We expect the statement that the cycle space equals the learnable space to be held more generally. 
Note that a Pauli noise model has a linear parameterization (i.e., the log of the Pauli eigenvalues), and one can study the learnable and gauge degrees of freedom as linear subspaces of the parameter spaces. 
For more general noise models, one might need to study the tangent space of the parameters, which has been considered in the theory of the first-order gauge-invariant (FOGI) model~\cite{nielsen2022first,blume2022taxonomy}. 
It is interesting to see if our method can be extended there to have a more quantitative understanding of the learnable and gauge degrees of freedom. 

Another important potential extension is to include mid-circuit measurement (MCM, also known as quantum non-demolition measurement), which is a crucial building block for fault-tolerant quantum computation and is being experimentally implemented~\cite{ryan2021realization,singh2023mid,google2023suppressing,bluvstein2024logical}. 
It has been shown that MCM can be incorporated into Pauli noise models via a generalized randomized compiling technique~\cite{beale2023randomized}. 
More recently, results on the learnability of MCM within the Pauli noise model have been reported~\cite{zhang2024generalized,hines2024pauli}. 
It is practically interesting to incorporate those results into the current framework, e.g., to have a self-consistent and efficient protocol for characterizing MCM with a scalable noise ansatz. 
More generally, it is interesting to study how our framework can be used for noise characterization of early fault-tolerant quantum computing devices~\cite{combes2017logical,wagner2023learning}.

\medskip
\noindent\textbf{Beyond realizable setting.} The learning algorithm we proposed relies on the fact that the noise model of interest is realizable, i.e., there is no out-of-model error. 
It would be interesting to analyze the effect of such errors. 
Our current analysis and algorithms require all the noise to be Pauli channels with certain constraints, specifically, the locality constraints as discussed in Sec.~\ref{sec:fully},~\ref{sec:quasi}, and~\ref{sec:CZ}. 
In practice, one can fit the same data into models with different locality assumption, and there will be a tradeoff between model expressivity and trainability.  
One can also utilize techniques from structure learning to find the optimal (quasi-)local Pauli noise model (see e.g. \cite{rouze2023efficient}). 
It would be interesting to explore learning algorithms with those features within our framework.

\section{Notations and Preliminaries}
    This section provides some basic definitions and preliminaries. For more symbols and notations defined later in the paper, see
    Table~\ref{tab:glossary}.
    \setlength{\tabcolsep}{4pt} %
    \begin{table}[t]
        \centering
        \begin{tabular}{l l l}
            \toprule
             Symbol & Definition or meaning & Reference \\
             \midrule
             $X$ & Parameter space & Def.~\ref{de:param_space}\\
             $\bm e_i$ & Standard basis for parameter space $X$ &  Def.~\ref{de:param_space}\\
             $X^{S/M/G}$ & Subspace of state prep./meas./gate parameters & Def.~\ref{de:param_space}\\
             $L$ & Learnable space & Lemma~\ref{le:learnable_space}\\
             $T$ & Gauge space & Lemma~\ref{le:gauge_space}\\
             $E$ & Edge space of the pattern transfer graph & Def.~\ref{de:edge_space}\\
             $Z$ & Cycle space & Def.~\ref{de:cycle_space}\\
             $U$ & Cut space & Def.~\ref{de:cut_space}\\
             $X_R$ & Reduced parameter space & Def.~\ref{de:reduced_param}\\
             $\bm\vartheta_i$ & Standard basis for reduced parameter space $X_R$ & Def.~\ref{de:reduced_param}\\
             $\mc Q$ & Embedding map from reduced to complete model& Def.~\ref{de:reduced_param}\\
             $L_R$ & Reduced learnable space & Def.~\ref{de:reduced_learnable}\\
             $T_R$ & Reduced gauge space & Def.~\ref{de:reduced_gauge}\\
             $L_R^G$ & Reduced learnable subspace for gate parameters & Def.~\ref{de:L_R^G}\\
             $\mf d_\nu$ & Subsystem depolarizing gauge (SDG) & Def.~\ref{de:sdg}\\
             $\pi^{S/M/G}$ & Projection onto subspace of state prep./meas./gate parameters & Proof of Theorem~\ref{th:fully_learnability}\\
             $\Omega$ & Factor set of a Markov Random Field & Sec.~\ref{sec:quasi_notation}\\
             $\Omega_*$ & Subset of all maximal factors from $\Omega$ &Sec.~\ref{sec:quasi_notation}\\
             $\Xi_{\mc G}$ & Extended support map & Lemma~\ref{le:covariant_sufficient}\\
             \bottomrule
        \end{tabular}
        \caption{Table of important symbols.}
        \label{tab:glossary}
    \end{table}

    \medskip
    \noindent\textbf{Bit string and index set.} For a positive integer $n$, define $[n]\coleq\{i\in\mbb Z:1\le i\le n\}$. We will often use the set of $n$-bit strings $\mbb Z_2^n$ and the collection of index sets $2^{[n]}\coleq\{\nu:\nu\subseteq [n]\}$.
    Define the following bijection between $\mbb Z_2^n$ and $2^{[n]}$: $s\in\mbb Z_2^n\leftrightarrow \{j\in[n]:s_j=1\}\in 2^{[n]}$. Throughout the paper, we often do not distinguish a bit string with an index set, as they can always be identified with each other via the aforementioned bijection. 
    We use $|s|$ to denote the Hamming weight of a bit string $s$ (or its cardinality if understood as an index set).
    We use $0_n$ to denote the all-zero bit string (which is $\varnothing$ if understood as an index set). We use $1_j$ to denote the bit string with its $j$th entry being $1$ and other entries being $0$ (which is $\{j\}$ if understood as an index set).
    
    We will often use an index set (or, equivalently, a bit string) as the subscript for a string of $n$ objects (e.g., another $n$-bit string or an $n$-qubit Pauli operator) to denote a subsequence from the selected indexes. For example, $XYZIZ_{\{1,3\}} = XYZIZ_{10100} = XZ$.

\medskip
\noindent\textbf{Pauli group.} We follow the notations of  \cite{Flammia2020,chen2022quantum}.
Let $\Pn=\{ I, X, Y, Z\}^{\otimes n}$ be the Pauli group modulo a global phase.
$\Pn$ is an Abelian group isomorphic to $\mbb Z^{2n}_2$.
Specifically, given a $2n$-bit string $a = (a_x,a_z) = a_{x,1}\cdots a_{x,n}a_{z,1}\cdots a_{z,n}$, the corresponding Pauli operator is given by
\begin{equation}    
	P_a = \bigotimes_{j=1}^n i^{a_{x,j}a_{z,j}}  X^{a_{x,j}} Z^{a_{z,j}},
\end{equation}
where the phase is chosen to ensure Hermiticity. We further define $P_0 \coleq I^{\otimes n} \equiv I_n$ as a shorthand for the $n$-qubit identity operator.

Besides the standard addition and inner product defined over $\mbb Z^{2n}_{2}$, we define the sympletic inner product over $\Pn$ as
\begin{equation}
	\expval{a,b} = \sum_{j=1}^n (a_{x,j}b_{z,j}+a_{z,j}b_{x,j}) \mod{2}.
\end{equation}
One can verify that $\expval{a,b}=0$ if $P_a,P_b$ commute and $\expval{a,b}=1$ otherwise.
We sometimes use the label $a$ and the Pauli operator $P_a$ interchangeably when there is no confusion.

Define $\pt:\Pn\mapsto\mbb Z^n_2$ that maps the $i$th single-qubit Pauli operator to $0$ if it is $I$, or $1$ if it is from $\{X,Y,Z\}$. We call $\pt(a)$ the \emph{pattern} of $P_a$. Alternatively, $\pt(a)$ represents the support of $P_a$ if understood as an index set. The Hamming weight of $\pt(a)$ is called the weight of $P_a$ and is denoted by $w(a)$ or $|a|$.

A $k$-qubit Pauli $P$ acting on a subset of qubits specified by an index set $\nu$ is denoted by $P_\nu$. Let $W\in\{I,X,Y,Z\}$ be a single-qubit Pauli operator and $\nu\in2^{[n]}$, we define 
$W^\nu=\bigotimes_{j\in\nu}W_j\bigotimes_{j\notin\nu} I_j$.

\paragraph{Pauli channel.} An $n$-qubit Pauli channel $\Lambda$ has the following two equivalent forms
\begin{equation}\label{eq:Pauli_channel}
    \Lambda(\rho) = \sum_{a\in{\sf P}^n}p_a P_a\rho P_a = \sum_{b\in{\sf P}^n}\frac{1}{2^n}\lambda_bP_b\Tr[P_b\rho].
\end{equation}
where $\{p_a\}_a$ is called the \emph{Pauli error rates}, while $\{\lambda_b\}_b$ is called the \emph{Pauli eigenvalues} or \emph{Pauli fidelities} (though the latter is potentially misleading since they can in general be negative).
These two sets of parameters are related by the Walsh-Hadamard transform~\cite{Flammia2020},
\begin{equation}
	\begin{aligned}
		\lambda_b = \sum_{a\in\Pn}p_a(-1)^{\expval{a,b}},\quad
		p_a = \frac{1}{4^n}\sum_{b\in\Pn}\lambda_b(-1)^{\expval{a,b}}.
	\end{aligned}
\end{equation}
Note that for $\Lambda$ to be a quantum channel, the trace-preserving condition is equivalent to $\sum_a p_a \equiv \lambda_0 = 1$, while the complete positivity is equivalent to $p_a\ge0,\forall a$. For $\Lambda$ that has the form of Eq.~\eqref{eq:Pauli_channel} but does not satisfy these two conditions, we refer to it as a \emph{Pauli diagonal map}.

We will also use the Pauli-Liouville representation.
A Hermitian operator $O$ acting on a $2^n$-dimensional Hilbert space can be vectorized as a $4^n$-dimensional real vector. We denote this vectorization of $O$ as $\lket{O}$ and the corresponding Hermitian conjugate as $\lbra{O}$, and define the inner product as $\lbraket{A}{B}\coleq \Tr(A^\dagger B)$, which is the Hilbert-Schmidt inner product in the original space. We also introduce the \emph{normalized Pauli operators} $\{\sigma_a\coleq P_a/\sqrt{2^{n}},~a\in\mbb Z_2^{2n}\}$, whose vectorization $\{\lket{\sigma_a}\}_a$ forms an orthonormal basis for the vectorized space.
The action of a quantum channel $\mc E$ becomes matrix multiplication in this representation, i.e, $\lket{\mc E(\rho)} = \mc E^{\mr{PL}}\lket{\rho}$, where $\mc E^{\mr{PL}}$ is the corresponding Pauli-Liouville operator and is usually simply denoted by $\mc E$ when there is no confusion.
In the Pauli-Liouville representation, a Pauli channel is given by
$$
\Lambda = \sum_{a\in\mbb Z_2^{2n}} \lambda_a\lketbra{\sigma_a}{\sigma_a}.
$$
which justifies the name of Pauli eigenvalues for $\{\lambda_a\}$.

\medskip
\noindent\textbf{Clifford group.} The $n$-qubit Clifford group ${\sf C}^n$ is the group of $n$-qubit unitary operators that preserve the Pauli group upon conjugation. Mathematically, $\forall C\in\Cn$, $\forall P_a\in\Pn$,  $CP_aC^\dagger \in \Pn\times\{\pm\}$.
We often denote the unitary channel of a Clifford gate using the same calligraphic letter, e.g., $\mc C(\rho) \coleq C\rho C^\dagger$.
We also define, with slight abuse of notation, the mapping $\mc C:\mbb Z_2^{2n}\mapsto \mbb Z_2^{2n}$ via $\mc C(P_a) \equiv CP_aC^\dagger= \pm P_{\mc C(a)}$. 
The Pauli-Liouville operator of a Clifford $\mc C$ has the following form,
\begin{equation}
    \mc C = \sum_{a} (\pm )\lketbra{\sigma_{\mc C(a)}}{\sigma_a}.
\end{equation}

\medskip
\noindent\textbf{Miscellaneous.}
For a real linear space $X$, the dual space of it is denoted by $X'$. For a linear map $\mc Q$, its dual map is denoted by $\mc Q^{\T}$. The image and kernel of a map $\mc Q$ are denoted by $\mr{Im}\mc Q$ and $\mr{Ker}\mc Q$, respectively. $\mathds 1[X]$ is the indicator function that takes $1$ if the statement $X$ is true and $0$ elsewise. All $\log$ represents natural logarithm.

\section{Learnability of gate set Pauli noise}\label{sec:framework}

\subsection{Basic definitions}

In this work, we focus on gate sets consisting of arbitrary single-qubit gates and a set of multi-qubit Clifford gates. We assume time-stationary and Markovian noise that has been twirled into Pauli channels. More formally, we have the following,
\begin{definition}[Complete Pauli noise model]
    Given a set of $n$-qubit Clifford layers $\mf G$, consider the following gate set,
\begin{enumerate}
    \item Initialization: $\rho_0=\ketbra{0}{0}^{\otimes n} \mapsto\tilde{\rho}_0=\Lambda^S(\rho_0)$, where $\Lambda^S = \sum_{a\in{\sf P}^n} \lambda_{\pt(a)}^{S}\lketbra{\sigma_a}{\sigma_a}$.
    \item Measurement: $\{E_j=\ketbra{j}{j}\}\mapsto \{\widetilde E_j = \Lambda^M(\ketbra{j}{j})\}$, where $\Lambda^M = \sum_{a\in{\sf P}^n} \lambda_{\pt(a)}^{M}\lketbra{\sigma_a}{\sigma_a}$.
    \item Single-qubit layer: $\bigotimes_{i=1}^n U_i\mapsto\bigotimes_{i=1}^n U_i,~\forall U_i\in \mr{SU}(2)$.
    \item Multi-qubit Clifford layer: $\mc G\mapsto\tilde{\mc G}=\mc G\circ\Lambda^{\mc G}$,~$\forall\mc G\in\mf G$, where $\Lambda^{\mc G} = \sum_{a\in{\sf P}^n} \lambda_a^{\mc G}\lketbra{\sigma_a}{\sigma_a}$.
\end{enumerate}
Furthermore, we assume all $\lambda_{\pt(a)}^S,\lambda_{\pt(a)}^M,\lambda_a^{\mc G}$ are strictly positive.
\end{definition}

In reality, such a model can always be imposed by applying appropriate Pauli twirling sequence, a technique known as randomized compiling~\cite{Wallman2016,hashim2020randomized}.
The assumption that single-qubit layer being noiseless can be relaxed to that they have gate-independent noise, in which case such noise can be thought as absorbed into the multi-qubit layer.
Note that assuming SPAM noise channel to be a generalized depolarizing channels (i.e., $\lambda_a$ only depends on $\pt(a)$) is without loss of generality as they only act on $Z$ basis eigenstates.

\begin{definition}[Parameter space]\label{de:param_space}
    Let $x_a^{S/M/\mc G}\coleq -\log\lambda_a^{S/M/\mc G}$ for all $a\neq \bm0$ for SPAM parameters and all $a\neq I_n$ for all gates $\mc G$'s parameters. Let $\bm x$ be a formal vector defined as
    \begin{equation}
        \bm x \coleq \sum_{u\neq\bm0}x_u^S\bm e_u^S + \sum_{u\neq\bm0}x_u^M\bm e_u^M +  \sum_{\mc G\in\mf G}\sum_{a\neq I_n}x_a^{\mc G}\bm e_a^{\mc G}
    \end{equation}
    where $\{\bm e_u^S,\bm e_u^M,\bm e_a^{\mc G}:u\neq 0_n,a\neq I_n,\mc G\in\mf G \}$ forms an orthonormal basis for a real linear space $X$, called the parameter space.    
    Furthermore, define $X^{S/M}\coleq \mr{span}\{\bm e_u^{S/M}:u\neq0_n\}$, $X^{G}\coleq \mr{span}\{\bm e_a^{\mc G}:a\neq I_n, \mc G\in\mf G\}$. 
\end{definition}
Note that, we do not include the Pauli fidelity corresponding to identity, as they are always $1$ for the channels to be trace-preserving. Clearly, any realization of the noise parameters uniquely corresponds to a vector $\bm x$ in $X$, but not all $\bm x\in X$ corresponds to a physical noise model as the channels might not be completely positive.

\begin{definition}[Experiment] \label{de:experiment}
    An experiment is a map $F:X\mapsto\mbb R^{2^n}$ specified by a quantum circuit $\mc C$ (which is a sequence of finitely many single-qubit layers and multi-qubit Clifford layers within $\mf G$),
    defined as $F(\bm x)_j = \Tr[\widetilde E_j\widetilde{\mc C}(\widetilde{\rho}_0
    )]$.
    Here the noisy initial state $\widetilde\rho_0$, noisy circuit $\widetilde{\mc C}$, and noisy measurement, $\widetilde{E_j}$ depends on $\bm x$.
\end{definition}

When $\bm x$ represents a physical noise model, any experiment $F(\bm x)$ yields a Born's rule probability distribution. However, for mathematical convenience, we naturally extend the domain of $F$ to all over $X$.

\begin{definition}[Learnable function]\label{de:learnable}
    A real-valued function $f$ on $X$ is \textbf{learnable} if there exists a set of experiments $\{F_j\}_{j=1}^m$ such that $f$ is a function of $\{F_j\}_{j=1}^m$. That is, there exists some $\hat f$ such that $f(\bm x) = \hat f(F_1(\bm x),\cdots,F_m(\bm x))$ for all $\bm x\in X$.
\end{definition}

We comment that the learnable function we define here is also known as identifiable function in the literature of identifiability analysis~\cite{hsiao1983identification}. Here, for consistency, we will stick to the term ``learnable'' which has been used in \cite{chen2023learnability,huang2022foundations}. 

\medskip

We will be specifically interested in studying the learnability of linear functions (and we will later explain why that suffices). Note that, there exists an identification between $X$ and its dual space $X'\equiv\mc L(X,\mbb R)$ using the standard inner product, given by $\bm x \leftrightarrow f_{\bm x}(\bm y)\coleq\bm x\cdot\bm y$.
One can thus think of a linear function on $X$ as a co-vector and denote it with a boldface letter.

\begin{lemma}\label{le:learnable_space}
    The learnable linear functions form a subspace of $X'$, called the \textbf{learnable space}, denoted by $L$.
\end{lemma}

\begin{definition}[Gauge vector]\label{de:gauge}
    A vector $\bm y\in X$ is a \textbf{gauge} vector if for any $\bm x\in X$, $F(\bm x) = F(\bm x + \eta\bm y)$ for any experiment $F$ and $\eta\in\mbb R$.
\end{definition}

\begin{lemma}\label{le:gauge_space}
    The gauge vectors form a linear subspace of $X$ called the \textbf{gauge space}, denoted by $T$.
\end{lemma}

\begin{proof}
    By definition, $\bm y\in X$ being a gauge vector implies $\eta\bm y$ being a gauge vector for all $\eta\in\mbb R$.
    Let $\bm y_1,\bm y_2\in X$ be any gauge vectors. By definition, for any $\eta\in\mbb R$ and any observation $F$,
    \begin{equation}
        F(\bm x) = F(\bm x + \eta\bm y_1) = F(\bm x + \eta\bm y_1 + \eta\bm y_2).
    \end{equation}
    Therefore, $\bm y_1 + \bm y_2$ is also a gauge vector. This proves the gauge vectors forms a linear subspace.
\end{proof}

\begin{lemma}\label{le:L_T_perp}
    $L\perp T$. That is, for any $\bm f\in L$ and $\bm y\in T$ we have $\bm f(\bm y) = \bm f\cdot \bm y =0$.
\end{lemma}
\begin{proof}
    $\bm y\in T$ implies $F(\bm x)=F(\bm x + \bm y)$ for any {experiment} $F$ and any $\bm x\in X$. $\bm f\in L$ implies there is a set of experiments $\{F_1,\cdots,F_m\}$ and a function $\hat f$ such that $\bm f(\bm x) = \hat f(F_1(\bm x),\cdots,F_m(\bm x))$. Therefore,
    \begin{equation}
        \bm f(\bm x + \bm y) = \hat f(\{F_j(\bm x+\bm y)\}_{j=1}^m) = \hat f(\{F_j(\bm x)\}_{j=1}^m) = \bm f(\bm x) 
    \end{equation}
    for any $\bm x$. Thus $\bm f(\bm y) = \bm f(\bm x + \bm y) - \bm f(\bm x)= 0$ by linearity. This completes the proof.
\end{proof}

\subsection{Learnability of complete noise model}\label{sec:learnability_complete}

In this section, we study the learnability of the complete Pauli noise model. The results presented here are an extension to those in \cite{chen2023learnability}. However, we will use a quite different proof technique that allows us to generalize the results to noise model with arbitrary constraints in the next section.
The following graph is originally defined in~\cite{chen2023learnability}, but we twist the definition a bit to include a ``SPAM node''.

\begin{definition}[pattern transfer graph with SPAM node (PTG)]\label{de:ptg}
  The pattern transfer graph associated with an $n$-qubit Pauli noise model is a directed graph ${\mathsf G}=({\sf V},{\sf E})$ constructed as follows:
\begin{itemize}
    \item $\sf V$: There are $2^n-1$ nodes $v_s$ for all $s\in\{0,1\}^n\backslash \{0^n\}$ and $1$ additional root node $v_{R}$.
    \item $\sf E$: %
    \begin{enumerate}
        \item For each $a\in{\sf P}^n\backslash\{I\}$ and $\mc G\in\mf G$, construct an edge $v_{\pt(P_a)}\to v_{\pt(\mc G(P_a))}$ labeled by $\bm e_{a}^{\mc G}$.
        \item For each $u\in\{0,1\}^n\backslash\{0^n\}$, construct an edge $v_R\to v_{u}$ labeled by $\bm e_{u}^{S}$.
        \item For each $u\in\{0,1\}^n\backslash\{0^n\}$, construct an edge $v_u\to v_{R}$ labeled by $\bm e_{u}^{M}$.
    \end{enumerate}
\end{itemize}
\end{definition}

The following definition of algebraic graph theoretic terminology follows \cite{chen2023learnability}.
\begin{definition}\label{de:edge_space}
    The \emph{edge space} $E$ is a real linear space spanned by a standard basis $\{\bm e\}$ where each $\bm e$ is a formal vector corresponding to an edge $e\in\mathsf E$.
\end{definition}
Let the edge space be the real linear space that takes all the edges as a basis, denoted by $E$. 
Clearly $E\cong X \cong X'$ up to the natural isomorphism. There are two subspaces of $E$ of special interest, which are defined as follows:

A \emph{path} is an alternating sequence of vertices and edges $ z=(v_{i_0},e_{i_1},v_{i_1},e_{i_2},v_{i_2},...,v_{i_{q-1}},e_{i_q},v_{i_q})$ such that each edge satisfies $e_{i_k}=(v_{i_{k-1}},v_{i_k})$.
A \emph{cycle} is a closed path, which means $v_{i_0} = v_{i_q}$.\footnote{
In the literature of algebraic graph theory, people sometimes allow reverse edges (i.e., $e_{i_k}=(v_{i_{k}},v_{i_{k-1}})$) in the definition of a cycle and refer to our definition as an \emph{oriented cycle} or a \emph{circuit}. Since all graphs considered in this work are either strongly-connected or a union of strongly-connected components, and it has been shown that oriented cycles can linearly represent any cycle for such graphs~\cite{gleiss2003circuit}, working with our current definition would not cause any inconsistency.
}
A cycle is \emph{simple} if it passes through each unique edge at most once.
For a simple cycle $z$ in the PTG, we assign a vector $\bm x_z\in E$ as follows
\begin{equation}
    \bm x_z[e] = \left\{
    \begin{aligned}
      +1,\quad& e\in z.\\
      0,\quad& e\notin z.
    \end{aligned}
    \right.
\end{equation}
\begin{definition}\label{de:cycle_space}
    The \textbf{cycle space} $Z$ is a linear subspace of $E$ spanned by all simple cycles in $\sf G$.
\end{definition}

\noindent A graph is \emph{strongly-connected} if for any two vertices, there exists a cycle containing both of them. The PTG we considered is strongly-connected by construction, as $v_R$ is directed connected to any other nodes.

\medskip

Given a partition of vertices $\mathsf V=\mathsf V_0\cup {\mathsf V_0^c}$, a \emph{cut} is the set of all edges $e = (u,v)$ such that one of $u,v$ belongs to ${{\mathsf{V}_0}}$ and the other belongs to ${{\mathsf{V}_0^c}}$. For each cut $p$ we assign a vector $\bm y_p\in E$ as follows
\begin{equation}\label{eq:cut}
    \bm y_p(e) = \left\{
    \begin{aligned}
      +1,\quad& e\in p,~\text{$e$ goes from ${{\mathsf{V}_0}}$ to ${{\mathsf{V}_0^c}}$.}\\
      -1,\quad& e\in p,~\text{$e$ goes from ${{\mathsf{V}_0^c}}$ to ${{\mathsf{V}_0}}$.}\\
      0,\quad& e\notin p.
    \end{aligned}
    \right.
\end{equation}
\begin{definition}\label{de:cut_space}
    The \textbf{cut space} $U$ is a linear subspace of $E$ spanned by all cuts in $\sf G$. 
\end{definition}

\medskip
\noindent Below is an important fact about the cycle space $Z$ and cut space $U$.
\begin{lemma}[\cite{gleiss2003circuit}]
    $E = Z\oplus^{\perp} U$. Here $\oplus^\perp$ means an orthogonal direct sum.
\end{lemma}

\bigskip

The following theorem is a crucial ingredient to understand the learnability of the complete Pauli noise model. 
\begin{theorem}
    [Linearity of experiments]\label{th:linearOb}
    Define a cycle to be \textbf{rooted} if it contains $v_R$ exactly once.
    \begin{enumerate}[label=(\alph*)]
        \item For any experiment $F$, there exists a finite set of linear functions $\{\bm f_j\}_{j=1}^M$ where each $\bm f_j$ corresponds to a rooted cycle, such that $F$ is a function of $\{\bm f_j\}$. That is, there exists a function $\widehat F$ such that $F(\bm x) = \widehat F(\bm f_1(\bm x),\cdots,\bm f_M(\bm x))$ for all $\bm x\in X$.
        \item For any linear function $\bm f_j$ corresponding to a rooted cycle, there exists an experiment $F$ such that $\bm f_j$ is a function of $F$. That is, there exists a function $\widehat f_j$ such that $\bm f_j(\bm x) = \widehat f_j(F(\bm x))$ for all $\bm x\in X$.
    \end{enumerate}
\end{theorem}
\noindent This is detailed in \cite{chen2023learnability}. {For completeness, we give a proof in SM Sec.~\ref{sec:proof_linearOb}.}  
Intuitively, Theorem~\ref{th:linearOb}(b) says that any rooted cycle can be learned from at least one experiments. In fact, the experiment involves preparing certain noisy Pauli eigenstate, applying a sequence of perfect single-qubit and noisy multi-qubit Clifford gates, and measuring certain noisy Pauli observable. The expectation value will have the form of $\lambda_{b_0}^S\lambda_{P_1}^{\mc G_1}\cdots\lambda_{P_m}^{\mc G_m}\lambda^M_{b_m}$, taking a log of which yields the desired rooted cycle; Theorem~\ref{th:linearOb}(a) says that rooted cycles are complete, in the sense that any experiment (possibly involving non-stabilizer states, measurements, and non-Clifford single-qubit gates) can be predicted using linear combinations of rooted cycles. Therefore, it is sufficient to characterize the learnability of linear functions and not to worry about the non-linear ones.

\medskip

To understand the full picture of learnable and gauge degrees of freedom, we introduce the following definition
\begin{definition}[Observation map]
    Fix an arbitrary basis $\{\bm f_i\}_{i=1}^{|Z|}$ for $Z$ such that each $\bm f_i$ represents a rooted cycle. Define the \textbf{observation} map $\mc P:X\mapsto\mbb R^{|Z|}$ via
    \begin{equation}
        \mc P(\bm x) = [\bm f_1(\bm x),\bm f_2(\bm x),\cdots,\bm f_{|Z|}(\bm x)]^{\mr T}.
    \end{equation}
\end{definition}
Note that, thanks to the strong connectivity of the PTG, 
a cycle basis consisting of only rooted cycles always exist.
One construction is given below. 
\begin{lemma}[A rooted cycle basis]\label{le:rooted_cycle_basis}
    The following set of vectors forms a basis for $Z$,
    \begin{equation}\label{eq:rooted_cycle_basis}
        {\mathsf B} \coleq \{\bm e^S_u+\bm e^M_u:u\neq {0_n}\} \cup \{ \bm e^S_{\pt(a)}+\bm e^{\mc G}_a+\bm e^M_{\pt(\mc G(a))}:\mc G\in\mf G,a\neq I_n\}.
    \end{equation}
\end{lemma}
\begin{proof}
    To begin with, every vector in ${\mathsf B}$ is obviously a rooted cycle.
    Since the PTG is strongly connected by construction, the dimension of $T$ is $2^n-1$~\cite{gleiss2003circuit}, thus
    \begin{equation}
        |Z| = |E| - |T| = 2^n-1 + |\mf G|(4^n-1),
    \end{equation}
    which is the same as the cardinality of ${\mathsf B}$. Therefore, it suffices to show the linear independence for ${\mathsf B}$. Consider the equation $\sum_{i:\bm e_i\in {\mathsf B}}\alpha_i \bm e_i = \bm 0$. Note that each $\bm e_a^{\mc G}$ appears in exactly one vector from the second part of ${\mathsf B}$, for which the corresponding $\alpha_i$ must be all zero, but the first part is also obviously linearly independent, which means all $\alpha_i$ must be zero, proving the linear independence.
\end{proof}

\begin{theorem}\label{th:complete} %
Up to identification between $X$ and $X'$ using the standard basis,
\begin{equation}
\begin{aligned}
    &\mr{Im}\mc P^{\T} ~~ \mr{Ker}\mc P\\
    &\quad~\|\quad\quad\| \\
    X &= L \oplus^\perp T\\
    \|&\quad~\|\quad\quad\| \\
    E &= Z \oplus^\perp U
\end{aligned}~.
\end{equation}
In words, the cycle space $Z$ is equal to the learnable space $L$, while the cut space $U$ is equal to the gauge space $T$. Besides, the two spaces are orthogonal complement to each other. Finally, the learnable space $L$ is also equal to the image for the dual of the observation map $\mc P$, while the gauge space $T$ is the kernel of $\mc P$.
\end{theorem}

\noindent \textit{Remark:} while $E=Z\oplus^\perp U$ is a graph-theoretic fact, we don't know a priori whether $X=L \oplus T$ holds under our definition. 

\begin{proof}
    We first show that $T=\mr{Ker}(\mc P)$, \textit{i.e.,} $\bm y\in T \Leftrightarrow \mc P(\bm y)=\bm 0$. 
    Indeed, Lemma~\ref{th:linearOb}(a) implies that $F(\bm x) = \widehat F(\mc P(\bm x))$ for some function $\widehat F$. Thus, for any $\bm y\in\mr{Ker}(\mc P)$ and $\eta\in\mbb R$, 
\begin{equation}
\begin{aligned}
    F(\bm x + \eta\bm y)&= \widehat{F}(\mc P(\bm x+\eta\bm y)) = \widehat{F}(\mc P(\bm x))=F(\bm x).
\end{aligned}
\end{equation}
The second equality uses the linearity of $\mc P$. Therefore, $\mr{Ker}(\mc P)\subseteq T$.
On the other hand, for any $\bm y\notin \mr{Ker}(\mc P)$, there is at least one $\bm f_i$ such that $\bm f_i(\bm y)\ne 0$. Now, Lemma~\ref{th:linearOb}(b) implies that there exists an experiment $F_i$ and a function $\hat f_i$ such that $\bm f_i(\bm x)=\hat f_i(F_i(\bm x))$ for all $\bm x$. Thus,
\begin{equation}
\begin{aligned}
    \bm f_i(\bm y)\neq 0 &\Rightarrow \bm f_i(\bm x)
    \neq \bm f_i(\bm x+\bm y)\\ &\Rightarrow
    \hat f_i(F(\bm x))\neq \hat f_i(F(\bm x+\bm y))\\
    &\Rightarrow
    F(\bm x)\neq F(\bm x+\bm y)
    \\&\Rightarrow \bm y\notin T.
\end{aligned}
\end{equation}
Therefore, $T\subseteq \mr{Ker}{\mc P}$. This proves the claim that $T = \mr{Ker}{P}$.

We then show that $L=\mr{Im}(\mc P^{\T})$. For any $\bm f\in\mr{Im}(\mc P^{\T})$, there is some $\bm h\in \mc L(\mbb R^{|Z|},\mbb R)$ such that $\bm f(\bm x) = \bm h(\mc P(\bm x)),~\forall\bm x\in X$. Lemma~\ref{th:linearOb}(b) implies that each entry of $\mc P(\bm x)$ is learnable, thus $\bm f$ is learnable. This means $\mr{Im}(\mc P^{\T})\subseteq L$; On the other hand, $L\perp T=\mr{Ker}(\mc P)\Rightarrow L\subseteq \mr{Im}(\mc P^{\T})$ as the kernel is the orthogonal complement of the adjoint image. This proves the claim that $L=\mr{Im}(\mc P^{\sf T})$.

Finally, it is not hard to see that $\mr{Ker}(\mc P) = U$ thanks to the duality between cycle space and cut space~\cite{gleiss2003circuit}, and hence $\mr{Im}(\mc P^{\T})= Z$. This completes the proof of Theorem~\ref{th:complete}.
\end{proof}

\subsection{Learnability of reduced noise model}~\label{sec:learnability_reduced}

Since the complete noise model contains exponentially many parameters, we often need to make additional assumptions in practice to reduce its complexity.
Here, we want to understand the learnability of such reduced noise model. 

\begin{definition}%
    \label{de:reduced_param}
    A {reduced Pauli noise model} is defined by a tuple $(X_R,\mc Q)$, where
    \begin{itemize}
        \item $X_R$ is a real linear space called the \textbf{reduced parameter space}. A realization of the reduced noise parameters $\{r_i\}_{i=1}^{|X_R|}$ corresponds to a vector $\bm r\in X_R$ via $\bm r = \sum_{i}r_i\bm\vartheta_i$ where $\{\bm\vartheta_i\}$ is a standard basis for $X_R$.
        \item $\mc Q:X_R\mapsto X$ is an injective linear map called the \textbf{embedding map} such that for any $\bm r$, the corresponding realization of the complete noise parameters is given by $\bm x\coleq\mc Q(\bm r)$.
    \end{itemize}
\end{definition}

Note that we require $\mc Q$ to be linear and injective. We will see in the following sections that many physically-motivated reduced noise models (e.g., spatially local noise) can indeed be described by a linear embedding map.
The requirement of injectivity, equivalent to $\mr{Ker}(\mc Q) = \{\bm 0\}$, is for mathematical convenience. In practice, if the embedding map is not injective, one can always construct another reduced model with fewer parameters and admits an injective embedding map.

\medskip

Any experiment $F$ as in Definition~\ref{de:experiment} induces a function $F\circ\mc Q$ acting on $X_R$.
The definitions and properties of learnable functions and gauge vectors then carry over to the reduced model, by replacing $X$ by $X_R$ in Definitions~\ref{de:learnable}, \ref{de:gauge}. Formally, we have
\begin{definition}[Reduced learnable space]\label{de:reduced_learnable}
    A real-valued function $f$ on $X_R$ is \textbf{learnable} if there exists a finite set of experiments $\{F_j\}_{j=1}^m$ such that $f$ is a function of $\{F_j\circ\mc Q\}_{j=1}^m$. That is, there exists some $\hat f$ such that $f(\bm r) = \hat f(F_1(\mc Q(\bm r)),\cdots,F_m(\mc Q(\bm r)))$ for all $\bm r\in X_R$.
    The linear subspace of $X_R'$ spanned by all linear learnable functions is called the \textbf{reduced learnable space}, denoted by $L_R$.
\end{definition}

\begin{definition}[Reduced gauge space]\label{de:reduced_gauge}
    A vector $\bm y\in X_R$ is a \textbf{gauge} vector if for any $\bm r\in X_R$, $F(\mc Q(\bm r)) = F(\mc Q(\bm r + \eta\bm y))$ for any experiment $F$ and $\eta\in\mbb R$. The linear subspace of $X_R$ spanned by all the gauge vectors is called the \textbf{reduced gauge space}, denoted by $T_R$.
\end{definition}

The characterization of learnability naturally extends to the reduced model, as stated in the next theorem. 

\begin{theorem}\label{th:reduced}
    $L_R = \mr{Im}(\mc Q^\T\circ\mc P^{\T})$, $T_R = \mr{Ker}(\mc P\circ\mc Q)$.
\end{theorem}

\noindent We provide a proof for Theorem~\ref{th:reduced} in Sec.~\ref{sec:proof_reduced}. Basically, all steps in the proof of Theorem~\ref{th:complete} naturally carry over.
The following corollary characterizes the relation of the learnable/gauge spaces between the complete and the reduced noise model.
\begin{corollary}\label{cor:main}
    $L_R = \mc Q^{\T}(L)$,~
    $T_R = \mc Q^{-1}(T)$. 
\end{corollary}
\noindent Here $\mc Q^{-1}(T)\coleq\{\bm x\in X_R~|~\mc Q(x)\in T \}$ is the preimage of $T$ under $\mc Q$.
Alternatively, $T_R = \mc Q^{-1}(T\cap \mr{Im}(\mc Q))$ where $\mc Q^{-1}$ is viewed as the inverse of $Q$ restricted to its image (recall that $\mc Q$ is assumed to be injective). 
    
\begin{proof}
    For any $\bm f\in X_R'$, $\bm f\in L_R\Leftrightarrow \exists \bm g\in \mc L(\mbb R^{|Z|},\mbb R)~s.t.~\bm f = \mc Q^{\T}(\mc P^{\T}(g))\Leftrightarrow  \bm f\in\mc Q^{\T}(L)$. Thus, $L_R=\mc Q^{\T}(L)$. 

    Let us consider $\mc Q(T_R)$. Theorems~\ref{th:complete},~\ref{th:reduced} imply that $\mc Q(T_R)\subseteq T$. Also, $\mc Q(T_R)\subseteq\mr{Im}(Q)$ trivially holds. On the other hand, for any $\bm y\in T\cap\mr{Im}(\mc Q)$, $\mc P(\bm y)=0$, and there exists some $\bm x\in X_R$ such that $\bm y=\mc Q(\bm x)$, but this implies $\mc P(\mc Q(\bm x))=0$ and thus $\bm x\in T_R$, $\bm y\in\mc Q(T_R)$. We conclude that $\mc Q(T_R) = T\cap\mr{Im}(\mc Q)$.    Finally, note that $\mc Q$ is injective by assumption and is thus invertible when restricted to its image. We conclude that $T_R = \mc Q^{-1}(T\cap\mr{Im}(\mc Q))$.
\end{proof}

\section{Algorithms for learning gate set Pauli noise}\label{sec:learning}

We have established a theory for characterizing the learnable/gauge degrees of freedom for any complete or reduced Pauli noise models. The next question to be asked is how to design experiments to learn the learnable degrees of freedom, and how efficiently one can learn them.
In this section, we first introduce a general algorithm for learning Pauli noise models, which leads to efficient experimental designs. 
We then discuss the issue of learning noise parameters to additive/multiplicative precision.
The algorithms discussed here will be applied in the next section to address many practically interesting examples of Pauli noise models.

From this section, we focus on Clifford circuits as they are sufficient for learning Pauli noise. Also, instead of describing an experiment with its measurement probability distribution, we will talk about expectation values of measuring certain Pauli observable. Since we assume noiseless single-qubit layers, measuring Pauli $P_a$ effectively picks up $\lambda^M_{\pt(a)}$ from the measurement noise channel. We thus define the effective noisy observable as $\tilde{P}_a\coleq\lambda^M_{\pt(a)}P_a$. Similarly, we use $\rho_a$ to denote the $+1$ eigenstate of $P_a$ obtained by applying an appropriate single-qubit Clifford layer to $\rho_{0}\equiv \ketbra{0}{0}^{\otimes n}$.\footnote{
Although $\rho_{a}$ is not uniquely defined when $P_a$ contains identity, choosing any $\rho_a$ that satisfies this condition will make no observable difference.} The noisy realization of $\rho_a$ is denoted as $\tilde{\rho}_a$ and satisfies $\Tr[P_a\tilde\rho_a]=\lambda^S_{\pt(a)}$. Finally, we use the tuple $({\rho_a, \mc C, P_b})$ to denote an experiment that (ideally) starts with $\rho_a$, applies the sequence of gates $\mc C$, and measures the observable $P_b$.

\subsection{Basic definitions and vanilla algorithm}~\label{sec:learning_simple}

Let us first define the task of learning a Pauli noise model. 
Given a reduced noise model described by an embedding map $\mc Q$. Recall that $L_R$ describes the linear space for all learnable functions about the noise model (and that Theorem~\ref{th:linearOb} says \textit{linear} learnable functions suffice). Thus, we just need to do the following: 
\begin{enumerate}
    \item  Specify a basis for $L_R$, denoted as $\{\bm f_j\}_{j=1}^{|L_R|}$,
    \item Specify the experiments to learn each $\bm f_j$.
\end{enumerate}
By conducting those experiments, we will have effectively learned every function in $L_R$, which are all we can learn about the noise model.

However, for an arbitrary basis vector $\bm f_j\in L_R$, it is not straightforward to find the experiments that learn it, despite their existence. 
We would like to find a reduced basis that is more closely related to experiments. Recall that for the complete model, the learnable space $L$ corresponds to the cycle space $Z$ of PTG, and each rooted cycle can be extracted from one experiment (Theorem~\ref{th:linearOb}(b)).
This motivates the following definition.
\begin{definition}[Reduced cycle]\label{de:reduced_cycle}
    For any $\bm f\in L_R$, if there exists a $\bm g\in L$ that is a cycle and that $\bm f = \mc Q^\T(\bm g)$, we call $\bm f$ a \textbf{reduced cycle} corresponding to $\bm g$. If $\bm g$ is also a rooted cycle, we call $\bm f$ a \textbf{rooted reduced cycle}.
\end{definition}
\noindent We remind the reader that not all $\bm g\in L$ represents a cycle, as some might only be linear combinations of cycles.
Since Corollary~\ref{cor:main} states that $L_R = \mc Q^\T(L)$, given a rooted cycle basis $\{\bm g_i\}_{i=1}^{|L|}$ for $L$, we know that $\{\mc Q^\T(\bm g_i)\}_{i=1}^{|L|}$ spans $L_R$, thus we can pick a subset $\{\bm f_j\}_{j=1}^{|L_R|}\subseteq\{\mc Q^\T(\bm g_i)\}_i$ as a basis for $L_R$ consisting solely of rooted reduced cycles. Theorem~\ref{th:linearOb}(b) thus immediately gives an experiment to learn $\bm f_j$.

As a concrete example, consider the rooted cycle basis introduced in Lemma~\ref{le:rooted_cycle_basis},
\begin{equation}
    {\mathsf B} \coleq \{\bm e^S_u+\bm e^M_u:u\neq {\bm0}\} \cup \{ \bm e^S_{\pt(a)}+\bm e^{\mc G}_a+\bm e^M_{\pt(\mc G(a))}:\mc G\in\mf G,a\neq I\}.
\end{equation}
We have the following procedure for noise learning.

\medskip

\noindent\textbf{Algorithm \algnum} [A simple algorithm for learning gate set Pauli noise].
\begin{enumerate}
    \item Select a subset ${\mathsf B}_R\subseteq {\mathsf B}$ such that $\mc Q^\T({\mathsf B}_R)$ forms a basis for $L_R$.
    \item Learn all $\bm e^S_u+\bm e^M_u\in {\mathsf B}_R$ by preparing $\tilde{\rho}_0$ and measuring $\tilde{Z}^u$.
    \item Learn all $\bm e^S_{\pt(a)}+\bm e^{\mc G}_a+\bm e^M_{\pt(\mc G(a))}\in {\mathsf B}_R$ by preparing $\tilde \rho_a$, applying $\tilde{\mc G}$, and measuring $\tilde P_{\mc G(a)}$.
\end{enumerate}

As an obvious challenge to this approach, the cardinality of ${\mathsf B}$ is $2^n-1 + |\mf G|(4^n-1)$, scaling exponentially with $n$. Even if the dimension of $L_R$ is polynomial in $n$, Step 1 of picking a basis for $L_R$ still requires exponential computational complexity in general.
However, for physically-motivated noise ansatzs such as spatially local noise, it is often easy to efficiently construct such a basis. This will be discussed in details when we turn to concrete noise models. Conditioned on that ${\mathsf B}_R$ can be efficiently constructed, step 2 and 3 involves at most $|L_R|$ experiments to run, which is polynomial in $n$ for any feasible reduced noise model.
One might notice that this algorithm resembles standard quantum process tomography~\cite{mohseni2008quantum}. Specifically, each gate from the gate set is learned one by one. No concatenation or composition of different gates is used.

\subsection{On learning gate noise to relative precision}\label{sec:learning_relative}

In practice, there are often considerations beyond merely constructing a complete set of experiments. 
One important factor to consider is the precision-efficiency trade-off in estimating the noise parameters.
Many existing RB-type noise characterization protocols can learn gate noise parameter to relative precision (with respect to the infidelity) using a small number of initialization and measurements~\cite{Flammia2020,harper2019statistical,Erhard2019}. More precisely, suppose we are interested in learning certain gate fidelity parameter $\lambda$, RB-type protocol can generally construct a family of experiments $\{(\rho_0,\mc C^d,P)\}_d$ such that the circuit fidelities satisfy
\begin{equation}\label{eq:exp_fit}
    {f}_d\coleq\Tr[\tilde P\tilde{\mc C}^d(\tilde\rho_0)]=\alpha\lambda^d,    
\end{equation}
where $\alpha$ represents some SPAM fidelity, $d$ is the number of concatenation of the relevant gates. We assume both $\alpha$ and $\lambda$ are sufficiently close to $1$.
Now choose $d\in\{0,m\}$, estimate each $\widehat{{f}}_d\coleq {f}_d+\hat\varepsilon_m$, and construct the following ratio estimator,
\begin{equation}
\begin{aligned}
    \hat\lambda\coleq \frac{\widehat{{f}}_m}{\widehat{{f}}_0} &= \left(\frac{\alpha\lambda^m+\hat\varepsilon_m}{\alpha+\hat\varepsilon_0}\right)^{\frac1m}
    \approx \lambda\left(1+\frac{\hat\varepsilon_m\lambda^{-m}-\hat\varepsilon_0}{m\alpha}\right).
\end{aligned}
\end{equation}
The last approximation holds given that $|\hat\varepsilon_0/\alpha|,|{(\hat\varepsilon_m\lambda^{-m}-\hat\varepsilon_0)}/{m\alpha}|\ll1$. If we choose $1/m\approx \log\lambda^{-1} = 1 - \lambda + o(1-\lambda)$, we would obtain that
\begin{equation}
    \hat\lambda\approx \lambda\left(1 + (1-\lambda)\frac{\hat\varepsilon_m e - \hat\varepsilon_0}{\alpha}\right).
\end{equation}
This means as long as we estimate ${f}_0$, ${f}_m$ to $\pm\varepsilon$ additive precision, we would obtain $O(\varepsilon)$ relative precision estimation for $\lambda$ with respect to the infidelity $1-\lambda$. As a side note, one can efficiently find the required $m$ using a search subroutine described in~\cite[Sec. V]{harper2019statistical}. The basic idea is to estimate ${f}_d$ for $d$'s chosen from an exponentially increasing list, and output $m=d$ when $\widehat{{f}}_d/\widehat{{f}}_0$ is below certain threshold.

\medskip

The above analysis shows that, as long as one can construct a family of experiments as in Eq.~\eqref{eq:exp_fit}, one can amplify $\lambda$ by concatenation and learn it to relative precision. We now explain how to fit this into the framework of gate set Pauli noise learning. 
\begin{definition}\label{de:L_R^G}
    For any reduced Pauli noise model $(X_R,\mc Q)$, the reduced learnable space for gate parameters is defined as $L^G_{R} \coleq \mc Q^\T(L\cap X^G)$. 
\end{definition}

\begin{theorem}[Noise parameter amplification]\label{th:relative}
    Given a reduced model $(X_R,\mc Q)$, for any reduced cycle $\bm f\in L_R^G$, there exists a family of experiments $\{(P,\mc C^m,\rho_P)\}_m$ and a rooted cycle $\bm\alpha\in L$ such that
    \begin{equation}
        \Tr[\tilde P\tilde{\mc C}^m(\tilde\rho_P)]=\exp(-\bm\alpha(\mc Q(\bm r)))\exp(-m\bm f(\bm r)),\quad\forall m\in\mbb N.
    \end{equation}
\end{theorem}
\noindent In other words, any functions on the reduced model that corresponds to learnable functions consisting solely of gate noise parameters in the complete model can be amplified via concatenation, and can thus be learned to relative precision. Here, the gate sequence $\mc C$ can be a composition of multiple gates from the gate set interleaved by noiseless single-qubit gates, which is similar to the concept of ``germs'' in gate set tomography~\cite{nielsen2021gate}, and also appears in other Pauli noise learning protocols~\cite{chen2023learnability,van2023probabilistic,carignan2023error}. 

\begin{proof}
    By definition, there exists a cycle $\bm g\in L\cap X^G$ such that $\bm f =\mc Q^\T(\bm g)$. Let $u$ be any vertex of the PTG that is contained in $\bm g$. Then $\bm\alpha\coleq \bm e^S_u + \bm e^M_u$ is a rooted cycle also containing $u$, and it is not hard to see that $\bm\alpha + m \bm g$ for any $m\in\mbb N$ forms a rooted cycle. Theorem~\ref{th:linearOb}(b) then implies that there exists an experiment $(P,\mc C_m, \rho_P)$ such that 
    \begin{equation}
        (\bm\alpha+m\bm g)\cdot\mc Q(\bm r) = -\log(\Tr[\tilde P\tilde{\mc C}_m(\tilde\rho_P)]).
    \end{equation}
    Finally, it is clear from the proof of Theorem~\ref{th:linearOb}(b) that each $\mc C_m$ can actually be chosen as $\mc C^m$ (i.e., concatenating $\mc C$ $m$ times) for an appropriate gate sequence $\mc C$. This completes the proof.
\end{proof}

\medskip

In light of Theorem~\ref{th:relative}, one can design an algorithm that not only learns everything in $L_R$, but also learn a basis of $L_R^G$ to relative precision.

\medskip
\noindent\textbf{Algorithm \algnum} [Relative precision learning of gate set Pauli noise].
\begin{enumerate}
    \item Select a reduced cycle basis of $L_R^G$, denoted as ${\mathsf B}_R^G$. 
    \item Augment ${\mathsf B}_R^G$ into a basis for $L_R$ using vectors from $Q^\T({\mathsf B})$. %
    \item Learn all vectors from ${\mathsf B}_R^G$ by running the corresponding set of experiments $\{(P,\mc C^m,\rho_P)\}_m$ described by Theorem~\ref{th:relative}.
    \item Learn all the other basis vectors using Step 2 and 3 from Algorithm 1.
\end{enumerate}
\noindent Again, the major challenge is to efficiently construct the reduced cycle basis, which would generally require doing linear algebra on an exponential-dimensional linear space. We will soon see that for some physically-motivated noise ansatz and realistic gate set, such basis can be constructed efficiently.

\medskip

We close this section with one final remark: Although specifying a reduced cycle basis of $L_R^G$ is sufficient to specify any functions therein, learning each basis function to relative precision does not guarantee relative precision estimation for arbitrary functions.
For example, relative-precision estimators for both $\bm f_1$ and $\bm f_2$ does not guarantee a relative-precision estimator for $\bm f_1-\bm f_2$, especially when the true value of $\bm f_1$ and $\bm f_2$ are close to each other. 
In such cases, one might need to conduct more experiments than merely learning a basis in order to retrieve relative precision estimation for those functions. We will leave this issue for future study.

\section{Applications to fully local noise model}\label{sec:fully}

\subsection{Basic definitions}

We have established general results on the learnability of any reduced Pauli noise models. Now we apply our theory to concrete physically-motivated noise ansatz. In this section, we will focus on the fully local noise model, or the so-called crosstalk-free noise model. As we will soon see, the reduced gauge space in this case is generated by certain subsystem depolarizing channel. We will also give efficient construction of the learning algorithms.

Before defining into the spatially local Pauli noise model, let us first introduce the \emph{layer-uncorrelated noise model}, which is mostly for mathematical convenience but also practically relevant:
\begin{definition}\label{de:layer_indep}
    For a reduced noise model $(X_R, \mc Q)$, if $X_R$ and $\mc Q$ can be decomposed as 
    \begin{equation}
        \begin{aligned}
            X_R &= X_R^S\oplus X_R^M\bigoplus_{\mc G\in\mf G}X_R^{\mc G},\\
            \mc Q&= \mc Q^S\oplus \mc Q^M\bigoplus_{\mc G\in\mf G}\mc Q^{\mc G},
        \end{aligned}
    \end{equation}
    where $\mc Q^\square:X_R^\square\mapsto X^\square$ satisfies $\mc Q^\square(\bm r^\square)=\bm x^\square$ for all $\square\in\{S,M\}\cup\mf G$, then we call it a \textbf{layer-uncorrelated noise model}.
\end{definition}
\noindent In words, a reduced noise model being layer-uncorrelated means that, there is an independent set of reduced noise parameters for state preparation, measurement, and each layer of gates, respectively. 
Many noise models studied in the literature of quantum benchmarking satisfies this assumption (e.g.,~\cite{flammia2021averaged,carignan2023error}). One prominent exception is the so-called ``gate-independent noise model'' studied in e.g. standard randomized benchmarking~\cite{Emerson2005,Knill2008,Dankert2009,magesan2011scalable,Magesan2012b} where the noise channel associated with any Clifford gate is assumed to be the same. Though we will not consider such a noise model in this section, we remark that the general framework established in the previous section is still applicable.

\begin{lemma}\label{le:layerwise_indep}
    For a layer-uncorrelated model $(X_R, \mc Q)$, the reduced gauge space $T_R$ satisfies
    \begin{equation}
        \mc Q(T_R) = \bigcap_{\square\in\{S,M\}\cup\mf G}\left\{\bm z\in T:\bm z^\square\in\imq^\square\right\}.
    \end{equation}
\end{lemma}
\begin{proof}
    Corollary~\ref{cor:main} says $\mc Q(T_R) = T\cap\imq$, which satisfies
    \begin{equation}
    \begin{aligned}
        T\cap\imq &= \{\bm z\in T:\bm z\in \imq^S\oplus\imq^M\bigoplus_{\mc G\in\mf G}\imq^{\mc G}\}\\
        &= \{\bm z\in T: \bm z^{S}\in\imq^S,~\bm z^{M}\in\imq^M,~\bm z^{\mc G}\in\imq^\mc G,~\forall\mc G\in\mf G\}\\
        &= \bigcap_{\square\in\{S,M\}\cup\mf G}\left\{\bm z\in T:\bm z^\square\in\imq^\square\right\}.
    \end{aligned}
\end{equation}
\end{proof}
As a result, we can separately analyze the gauge vectors allowed by state preparation, measurements, and each layer of gates. The intersection of all of them will give the correct reduced gauge space. In what follows, we will use this approach to analyze the gauge space for the fully local noise model, and for the quasi-local noise model in the next section.

\subsection{Learnability of fully local noise model}\label{sec:fully_learnability}

Let us first define the fully local noise model.
For any $\mc G\in\mf G$, we use $\mr{supp}(\mc G)$ to denote the support of $\mc G$, i.e., the subsystem of qubits that $\mc G$ non-trivially acts on. 
The restriction of $\mc G$ to its support is denoted by $\hat{\mc G}$.
The number of qubits in the support is denoted by $|\mc G|$.
Recall that we use $\{\bm\vartheta_i\}$ and $\{\bm e_i\}$ to denote the standard basis for $X_R$ and $X$, respectively. That is, the reduced and complete noise parameters are encoded into vectors according to $\bm r=\sum_i r_i\bm\vartheta_i$ and $\bm x=\sum_i x_i\bm e_i$, respectively.

We will first give the formal definition of the fully local noise model, and explain the intuition in a moment.

\begin{definition}[fully local noise]\label{de:fully_local}
    Consider a layer-uncorrelated noise model $(X_R,\mc Q)$:
    \begin{itemize}
        \item If the state preparation noise parameters are given by $\{r_j^S:j\in[n]\}$ and that
        \begin{equation}
            \mc Q(\bm\vartheta_j^S) = \sum_{u\neq 0_n}\mathds 1 [u_j=1] \bm e^S_u,\quad\forall j\in[n],
        \end{equation}
        we say the model has \textbf{fully local state preparation} noise. Similar for the measurement.
        \item If the noise parameters for some $\mc G\in\mf G$ are given by $\{r_a^{\mc G}:a\in{\sf P}^{|\mc G|},a\neq I_{|\mc G|}\}$ and that
        \begin{equation}
            \mc Q(\bm\vartheta_a^{\mc G}) = \sum_{b\neq I_n}\mathds 1 [b_{\mr{supp}(\mc G)}=a] \bm e^{\mc G}_b,\quad\forall a\in{\sf P}^{|\mc G|}, a\neq I_{|\mc G|},
        \end{equation}
        we say the model satisfies fully local noise for $\mc G$. If this holds for all $\mc G\in\mf G$ we say the model has \textbf{fully local gate noise}.
    \end{itemize}
    The model is \textbf{fully local} if it has fully local state preparation, measurement, and gate noise.
\end{definition}

The following lemma describes the transformation between the reduced parameters and the complete parameters within a fully local model, as an immediate consequence of the above definitions.
\begin{lemma}\label{le:fully_local_trans}
    If $(X_R,\mc Q)$ satisfies fully local state preparation noise, then
        \begin{equation}
        \left\{\begin{aligned}
            x_u^S &= \sum_{j\in[n]}\mathds 1[u_j=1]r_j^S,\quad\forall u\in\{0,1\}^n,u\neq0_n.\\
            r_j^S &= x^S_{1_j},\quad\forall j\in[n].
        \end{aligned} \right.   
        \end{equation}
        Similar for the measurement side; If $(X_R,\mc Q)$ satisfies fully local noise for some gate $\mc G$, then
        \begin{equation}
        \left\{
        \begin{aligned}
            x_b^{\mc G} &= r^{\mc G}_{b_{\mr{supp}(\mc G)}},\quad\forall b\in\Pn,b\neq I_n.\\
            r^{\mc G}_a &= x^{\mc G}_b~s.t.~b_{\mr{supp}{(\mc G})}=a,\quad \forall a\in\mathsf{P}^{|\mc G|},a\neq I_{|\mc G|}.
        \end{aligned}
        \right.
        \end{equation}
\end{lemma}
\begin{proof}
    By definition of the fully local noise model, one can easily verify that
    \begin{align}
         \bm x^S &= \mc Q^S(\bm r^S) = \sum_{u\neq0_n}\sum_{j\in[n]}\mathds 1[u_j=1]r_j^S\bm e_u^S.\\
         \bm x^{\mc G} &= \mc Q^{\mc G}(\bm r^{\mc G}) = \sum_{b\neq I_n}r^{\mc G}_{b_{\mr{supp}(\mc G)}}\bm e_b^{\mc G}.
    \end{align}
    Comparing these with the definitions of $\bm x$ and $\bm r$ immediately gives us the transformation from $\{r_i\}$ to $\{x_i\}$. The correctness of the inverse transformation can be verified by substitution.
\end{proof}

\medskip

\begin{figure}
    \centering
    \includegraphics[width=\linewidth]{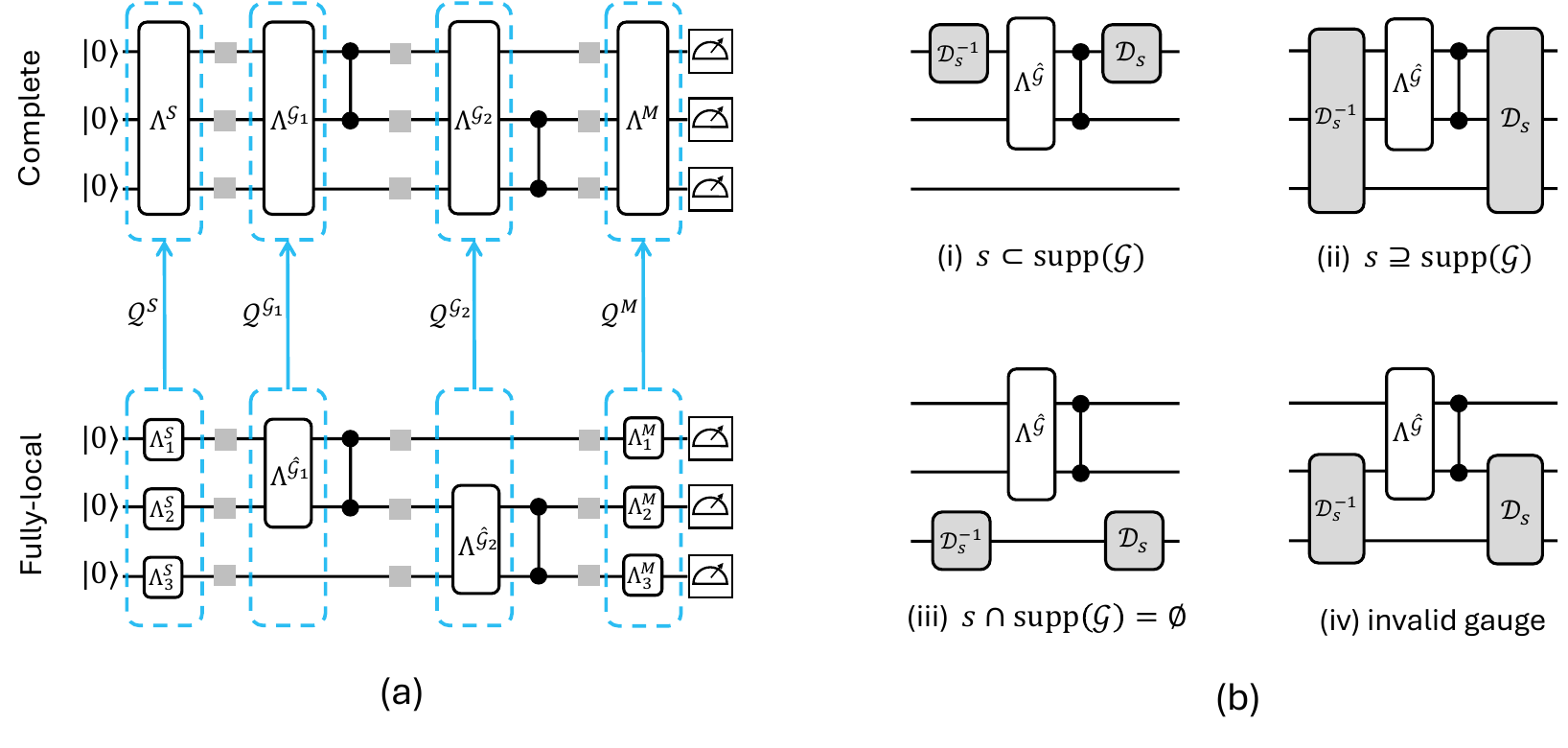}
    \caption{(a) Illustration of a fully local noise model (bottom) and how it is embedded into a complete noise model (above). Note that the state prep., meas., and each layer of gates have an independent embedding map, which means this is a layer-uncorrelated noise model.
    (b) Illustration of the subsystem depolarizing gauges (denoted by $\mc D_s$) acting on a gate with fully local noise. Here, (i)-(iii) shows the three types of valid SDG's as given by Theorem~\ref{th:fully_learnability}(b), out of which only case (i) can non-trivially changes the noise channel. (iv) shows a SDG that is not an allowed gauge transformation (when the gate has a fully-connected pattern transfer subgraph). Intuitively, it propagates the $2$-qubit noise channel into a $3$-qubit channel, violating the fully local assumptions.}
    \label{fig:fully_additional}
\end{figure}

In words, the fully local state preparation or measurement noise is a qubit-wise independent channel, and the fully local gate noise is a channel acting only on the gate's support and inducing no crosstalk on idling qubits. See Fig.~\ref{fig:fully_additional}(a) for an illustration. 
We study these cases separately, as there could be situations where the SPAM noise is local while gate noise is not, or the other way around.

We also note that whether fully local gate noise describes contextual noise depends on the definition of $\mf G$: If one define $\mf G$ in a way that each $\mc G\in\mf G$ acts non-trivially at, say, at most $2$-qubit. Let $\mc G_1$, $\mc G_2$ be two Clifford layers from $\mf G$ that act non-trivially on disjoint sets of qubits. The noise channel for applying both gates in parallel can then be viewed as the tensor product of two associated noise channels. However, if one believe the parallel application should suffer from a different noise channel, one can add it as an independent layer into $\mf G$. Although one needs to be careful not to makes the size $\mf G$ exponentially large.

The advantages of a fully local noise model is a small number of model parameters (as small as $O(n)$ for a system with a linear connectivity), which makes working with such model potentially computationally efficient.
Previous work, including ACES~\cite{flammia2021averaged}, has focused on characterizing such a noise model. It is thus of great interest to understand the learnable/gauge degrees of freedom within this model.
For this purpose, let us introduce the following set of gauge vectors.
\begin{definition}[Subsystem depolarizing gauge, SDG]\label{de:sdg}
    For any $s\in\{0,1\}^n\backslash0_n$, let $\mf d_s$ be the cut vector specified by the vertex partition $$
    {{\mathsf{V}_0}} = \{u:u_s=0_{|s|}\}\cup\{v_R\},\quad {{\mathsf{V}_0^c}}=\{u : u_s \neq 0_{|s|}\}.$$ 
    We call such $\mf d_s$ a \textbf{subsystem depolarizing gauge} (SDG).
\end{definition}

\noindent To see why $\mf d_s$ is called a subsystem depolarizing gauge, let us first write down its expansion on the standard basis: 
\begin{equation}\label{eq:sdg_entry}
\begin{aligned}
    \bm e_u^S\cdot \mf d_s &= \mathds 1[u_s\neq0_{|s|}],&&\quad\forall u\in\{0,1\}^n\backslash0_n.\\
    \bm e_u^M\cdot \mf d_s &= -\mathds 1[u_s\neq0_{|s|}],&&\quad\forall u\in\{0,1\}^n\backslash0_n.\\
    \bm e_a^{\mc G}\cdot \mf d_s &= \mathds 1[\pt(\mc G(a))_s\neq 0_{|s|}] - \mathds 1[\pt(a)_{s}\neq 0_{|s|}], &&\quad\forall \mc G\in\mf G,~\forall a\in {\sf P}^{n}\backslash I_n.
\end{aligned}
\end{equation}
which can be obtained using the definition of cut vectors and the PTG. Now, consider a transformation of noise parameters given by $\bm x\mapsto \bm x+\eta\mf d_s$. It is not hard to verify this corresponds to the following gauge transformation,
\begin{equation}
    \tilde\rho_0\mapsto \mc D_{s,\eta}(\tilde\rho_0),\quad
    \tilde E_j\mapsto \mc D^{-1}_{s,\eta}(\tilde E_j),\quad
    \tilde{\mc G}_s\mapsto \mc D_{s,\eta}\circ\tilde{\mc G}\circ\mc D_{s,\eta}^{-1},
\end{equation}
where $\mc D_{s,\eta}\coleq \sum_{a}e^{-\eta\mathds 1[a_s\neq I_{|s|}]}\lketbra{\sigma_a}{\sigma_a}$ is a partially depolarizing channel (with strength parameter $\eta$) acting on the subsystem specified by the index set $s$. For notational convenience, we define
\begin{equation}\label{eq:sd_channel}
    \mc D_s \coleq \mc D_{s,1} = \sum_{a}e^{-\mathds 1[a_s\neq I_{|s|}]}\lketbra{\sigma_a}{\sigma_a},
\end{equation}
and call $\mc D_s$ simply as the subsystem depolarizing channel on $s$.

\begin{lemma}\label{le:sdg_basis}
    The set of SDGs $\{\mf d_s\}_{s\neq0_n}$ forms a basis for the gauge space $T$.
\end{lemma}
\noindent The proof is given in Sec.~\ref{sec:proof_sdg}. 
In fact, the gauge space of the fully local noise model is spanned by a subset of the SDG's. We have the following results:
\begin{theorem}\label{th:fully_learnability}
Given a layer-uncorrelated noise model $(X_R,\mc Q)$
\begin{enumerate}[label=(\alph*)]
    \item If the model satisfies fully local state preparation noise, then 
    \begin{equation}\label{eq:fully_local_spam}
        \{\bm z\in T:\bm z^S\in\imq^S\} = \mr{span}\{\mf d_s:|s|=1\}.
    \end{equation}
    Similarly for the measurement noise. 
    \item If the model satisfies fully local noise for some gate $\mc G\in\mf G$, then
    \begin{equation}\label{eq:fully_local_gate}
        \{\bm z\in T:\bm z^{\mc G}\in\imq^{\mc G}\}\supseteq \mr{span}\{\mf d_s:s\subset\mr{supp}(\mc G)~\textrm{or}~s\supseteq\mr{supp}(\mc G)~\textrm{or}~s\cap\mr{supp}(\mc G)=\varnothing\}.
    \end{equation}
    Furthermore, the two sides are equal if $\hat{\mc G}$ induces a connected pattern transfer subgraph.
    \item If the model is fully local, its embedded gauge space is given by
    \begin{equation}
        \mc Q(T_R) = \mr{span}\{\mf d_s:|s|=1\}.
    \end{equation}
    And thus $\dim T_R = \dim \mc Q(T_R) = n$.
\end{enumerate}
\end{theorem}
\noindent Here, the pattern transfer subgraph of $\hat{\mc G}$ is a $|\mc G|$-qubit PTG with the only gate being $\hat{\mc G}$, and with the SPAM (root) node removed. Examples of $2$-qubit gates that have connected pattern transfer subgraph include CZ, iSWAP, but not SWAP~\cite{chen2023learnability}, see Fig.~\ref{fig:PTMs}. As another example, the pattern transfer subgraph for $\mr{CZ}\otimes\mr{CZ}$ is not connected. 
The allowed/forbidden gauge transformation for the fully local gate noise are depicted in Fig.~\ref{fig:fully_additional}(b).

\medskip

Theorem~\ref{th:fully_learnability}(c) characterizes the learnability of fully local Pauli noise models. In words, the embedded gauge space $\mc Q(T_R)$ is spanned exactly by those $n$ single-qubit depolarizing gauges. 
One can also write down the reduced gauge space as $T_R = \mr{span}\{\mc Q^{-1}(\mf d_s):|s|=1\}$, where each $\mc Q^{-1}(\mf d_s)$ is well-defined\footnote{
Though $\mc Q$ is not invertible, it is injective and we know that $\forall|s|=1,~\mf d_s\in\imq$, thus $\mc Q^{-1}(\mf d_s)$ is uniquely defined.
} and can be efficiently calculated using the inverse transformation given in Lemma~\ref{le:fully_local_trans}.
The learnable space $L_R$ can be obtained by, e.g., computing the orthogonal complement of $T_R$ within $X_R$ using standard linear algebra procedures. We will come back to the construction of $L_R$ later.

\medskip

Now we present the proof for Theorem~\ref{th:fully_learnability}(a), the direct part of \ref{th:fully_learnability}(b), and \ref{th:fully_learnability}(c). The converse part of \ref{th:fully_learnability}(b) is slightly more involved and is postponed to App.~\ref{sec:proof_fully_converse} for clarity.

\begin{proof}[Proof of Theorem~\ref{th:fully_learnability}(a)]
   Let ${\pi}^S: X\mapsto X^S$ be the projection map ${\pi}^S(\bm x) = \bm x^S$, i.e., restriction to the state preparation parameters. We first show $L.H.S.\supseteq R.H.S.$, for which we only need to prove that $\pi^S(\mf d_{1_j})\in\imq^S$ for all $j\in[n]$. Indeed, we have that,
   \begin{equation}
       \pi^S(\mf d_{1_j}) = \sum_{u} \mathds 1[u_j\neq 0]\bm e^S_u = \mc Q^S(\bm\vartheta_j^S).
   \end{equation}
    The two equations are by Definition~\ref{de:sdg} and Definition~\ref{de:fully_local}, respectively. Thus we have $L.H.S.\supseteq R.H.S.$.

    \medskip
    
    Conversely, notice that ${\pi}^S$ is an invertible map when restricted to $T$. To see this, consider the canonical cut basis $\{\mf y_u:u\neq0_n\}$ for $T$ such that each $\mf y_u$ is induced by the vertex partition ${\mathsf V}={{\mathsf{V}_0}}\cup {\mathsf{V}_0^c}$ where ${{\mathsf{V}_0}} = \{u\}$. We have ${\pi}^S(\mf y_u) = -\bm e_u^S$. This means ${\pi}^S$ is a basis transform between $T$ and $X^S$, thus invertible. Therefore,
    \begin{equation}
        \dim(\{\bm z\in T:\pi^S(\bm z)\in\imq^S\}) = \dim(\imq^S) \le \dim(X_R^S) = n,
    \end{equation}
    where the first equation uses the invertibility of $\pi^S$ restricted to $T$. This means the dimension of L.H.S. is no larger than R.H.S. Consequently, they must be equal. The proof works similarly on the measurement side. 
\end{proof}

\medskip

\begin{proof}[Proof of Theorem~\ref{th:fully_learnability}(b)]
Let ${\pi}^{\mc G}:X\mapsto X^{\mc G}$ be the projection map ${\pi}^{\mc G}(\bm x) = \bm x^{\mc G}$.
Recall that $$\bm e_a^{\mc G}\cdot \mf d_s = \mathds 1[\pt(\mc G(a))_s\neq 0_{|s|}] - \mathds 1[\pt(a)_{s}\neq 0_{|s|}],\quad \forall a\neq I_n.$$ 
\begin{itemize}
    \item If $s\cap\mr{supp}(\mc G)=\varnothing$: we have $\pt(\mc G(a))_s = \pt(a)_s$ and thus $\bm e_a^{\mc G}\cdot\mf d_s = 0$ for all $a$. Therefore, we trivially have ${\pi}^{\mc G}(\mf d_s)\in\imq^{\mc G}$.
    \item  If $s\supseteq\mr{supp}(\mc G)$: we have $\pt(\mc G(a))_s=0_{|s|}\Leftrightarrow \pt(a)_s=0_{|s|}$ because $\hat{\mc G}(a_{\mr{supp}(\mc G)})=I_{|\mc G|}$ if and only if $a_{\mr{supp}(\mc G)} = I_{|\mc G|}$. Thus $\bm e_a^{\mc G}\cdot\mf d_s = 0$ and ${\pi}^{\mc G}(\mf d_s)\in\imq^{\mc G}$ trivially holds.
    \item If $s \subset \mr{supp}(\mc G)$: choose an $\mf r^{\mc G}\in X^{\mc G}_R$ such that 
    \begin{equation}
        \bm\vartheta_b^{\mc G}\cdot\mf r^{\mc G} = \mathds 1[\pt(\hat{\mc G}(b))_s\neq 0_{|s|}] - \mathds 1[\pt(b)_s\neq 0_{|s|}].
    \end{equation}
    Then we can see that
    \begin{equation}
    \begin{aligned}
        \bm e_a^{\mc G}\cdot\mc Q^{\mc G}(\mf r^{\mc G}) = \mathds 1[\pt({\hat{\mc G}}(a_{\mr{supp}(\mc G)}))_s\neq 0_{|s|}] - \mathds 1[\pt(a_{\mr{supp}(\mc G)})_s\neq 0_{|s|}] = \bm e_a^{\mc G}\cdot\mf d_s.
    \end{aligned}
    \end{equation}
    This means ${\pi}^{\mc G}(\mf d_s) = \mc Q^{\mc G}(\mf r^{\mc G})$. Thus ${\pi}^{\mc G}(\mf d_s)\in\imq^{\mc G}$.
\end{itemize}
Combining all the three cases and using linearity, we conclude that
\begin{equation}\label{eq:th_fully_again}
    \{\bm z\in T:\bm z^{\mc G}\in\imq^{\mc G}\}\supseteq \mr{span}\{\mf d_s:s\subset\mr{supp}(\mc G)~\textrm{or}~s\supseteq\mr{supp}(\mc G)~\textrm{or}~s\cap\mr{supp}(\mc G)=\varnothing\}.
\end{equation}
Conversely, the two sides are equal if $\hat{\mc G}$ induces a connected pattern transfer subgraph, which is proven in App.~\ref{sec:proof_fully_converse}. The proof works by establishing an upper bound on the dimension of the L.H.S., and then show by counting that this bound equals the dimension of the R.H.S. 
\end{proof}

We make a few remarks regarding the proof of  Theorem~\ref{th:fully_learnability}(b). First, it is clear from the proof that $\mr{span}\{\mf d_s:s\subset\mr{supp}(\mc G)\}$ are gauges that non-trivially change the noise parameters of $\mc G$, while $\mr{span}\{\mf d_s:s\supseteq\mr{supp}(\mc G)~\text{or}~s\cap\mr{supp}(\mc G)=\varnothing\}$ are gauges that ``pass through'' $\mc G$ leaving gate noise parameters unchanged;
Second, The proof for Eq.~\eqref{eq:th_fully_again} taking equality crucially relies on the connectivity of the pattern transfer subgraph of $\widehat{\mc G}$. If this is not satisfied, there could be more gauges that can trivially pass through $\mc G$, which needs to be analyzed \textit{ad hoc} using Theorem~\ref{th:reduced}.

\medskip

\begin{proof}[Proof of Theorem~\ref{th:fully_learnability}(c)]
    Thanks to Lemma~\ref{le:layerwise_indep}, $\mc Q(T_R)$ is given by the intersection of Eq.~\eqref{eq:fully_local_spam} for SPAM and Eq.~\eqref{eq:fully_local_gate} for all $\mc G\in\mf G$. It's easy to see that Eq.~\eqref{eq:fully_local_spam} is always a subset of Eq.~\eqref{eq:fully_local_gate} for any $\mc G\in\mf G$, thus their intersection is simply given by Eq.~\eqref{eq:fully_local_spam}, the single-qubit SDG's.
\end{proof}

\subsection{Efficient learning of fully local noise model}\label{sec:local_learning}

Having characterized the reduced gauge space of the fully local noise model, we turn to the task of designing efficient learning experiments. Following the discussion in Sec.~\ref{sec:learning}, the key is to find an appropriate reduced cycle basis for $L_R$. 
Let us start with the first type of construction which gives the simplest experimental design (but does not care about relative precision estimation). 
Recall the rooted cycle basis defined in Eq.~\eqref{eq:rooted_cycle_basis},
\begin{equation}
    {\mathsf B} \coleq \{\bm e^S_u+\bm e^M_u:u\neq {0_n}\} \cup \{ \bm e^S_{\pt(a)}+\bm e^{\mc G}_a+\bm e^M_{\pt(\mc G(a))}:\mc G\in\mf G,a\neq I_n\}.
\end{equation}
Consider the following subset of $\mathsf B$, denoted as $\mathsf B'$,
\begin{equation}
    {\mathsf B'} \coleq \{\bm e^S_{1_j}+\bm e^M_{1_j}:j\in[n]\} \cup \{ \bm e^S_{\pt(a)}+\bm e^{\mc G}_a+\bm e^M_{\pt(\mc G(a))}:\mc G\in\mf G,\mr{supp}(a)\subseteq\mr{supp}(\mc G),a\neq I\}.
\end{equation}
\begin{proposition}
    Let $\mathsf B_R\coleq\mc Q^\T(\mathsf B')$. Then, $\mathsf B_R$ forms a reduced cycle basis for $L_R$.
\end{proposition}
\begin{proof}
    We have
    \begin{equation}
    \begin{aligned}
        {\mathsf B}_R \coleq& \mc Q^\T\left(\{\bm e^S_u+\bm e^M_u:|u|=1\} \cup \{ \bm e^S_{\pt(a)}+\bm e^{\mc G}_a+\bm e^M_{\pt(\mc G(a))}:\mc G\in\mf G,\supp{(a)}\subseteq \supp{(\mc G)},a\neq I_n \}\right)\\
            =& \left\{\bm \vartheta^S_j+\bm \vartheta^M_j:j\in[n]\right\} \cup
            \left\{\sum_{j\in\mr{supp}(\mc G)}\mathds (\mathds 1[{\pt(b)_j\neq0}]\bm\vartheta^S_j+\mathds 1[{\pt(\widehat{\mc G}(b))_j\neq0}]\bm\vartheta^M_j)+\bm\vartheta^{\mc G}_{b}:\mc G\in\mf G,b\neq I_{|\mc G|}\right\}.
    \end{aligned}
    \end{equation}
    Recall that $\widehat{\mc G}$ is $\mc G$ restricted to its support.
    It is not hard to see that ${\mathsf B}_R$ forms a linearly-independent set. Meanwhile, $|\mathsf B_R| = n + \sum_{\mc G\in\mf G}(4^{|\mc G|}-1) = \dim(X_R) - \dim(T_R) = \dim(L_R)$. This means ${\mathsf B}_R$ is indeed a basis for $L_R$.
\end{proof}
\bigskip

\noindent Following Algorithm 1, the following set of experiments is sufficient to learn the noise model:
\begin{enumerate}
    \item For $j\in[n]$, prepare $\tilde\rho_0$ and measure $\tilde Z$ on the $j$th qubit.
    \item For $\mc G\in\mf G$, $b\in{\sf P}^{|\mc G|}\backslash I_{|\mc G|}$, prepare $\tilde{\rho}_b$, apply $\tilde{\mc G}$, and measure $\tilde P_{\widehat{\mc G}(b)}$ on the support of $\mc G$.
\end{enumerate}
The total number of experiments is $n+\sum_{\mc G\in\mf G}(4^{|\mc G|}-1)$. If each $\mc G$ acts non-trivially only on a constant number of qubits, the number of experiments needed is $O(n)$, and is efficient to conduct. 
We also remark that, here we have not taken advantages of the fact that multiple observables can be extracted from the same experiments. In practice, the sample complexity can thus be further reduced. 

To learn the gate noise parameters to relative precision, one needs to construct a reduced cycle basis for the gate noise parameter space $L_R^{G} = \mc Q^\T(L\cap X^G)$. We find this problem hard in general. However, if every $\mc G\in\mf G$ is supported non-trivially on at most 2 qubits, we do have an efficient algorithms to construct such a basis. See App~\ref{sec:local_relative}.

\section{Applications to quasi-local noise model}\label{sec:quasi}

We have presented a comprehensive treatment on learning the fully local Pauli noise model. 
In practice, the effect of spatially correlated noise can sometimes be non-negligible, both for SPAM~\cite{bravyi2021mitigating} and for gate~\cite{fowler2014quantifying,nickerson2019analysing}. 
It is thus desirable to study a quasi-local noise model that takes a finite spatial correlation (e.g., nearest-neighbor) into consideration while keeping the total number of parameters manageable. Such models have been studied in the literature of Pauli noise learning~\cite{Flammia2020,van2023probabilistic,Erhard2019,wagner2023learning}.
In this section, we will introduce the quasi-local model in the framework of gate set Pauli noise learning, provide results on its learnability, and discuss efficient learning algorithms. As a specific example, we will apply our theory to a noise model that has been intensively studied in recent experiments~\cite{van2023probabilistic,kim2023evidence} and obtain a self-consistent algorithm for learning the noise model, which has not been previously achieved.

\subsection{Basic definitions}\label{sec:quasi_notation}

We first introduce some notations that are needed for defining the quasi-local noise model.
For $a,b\in{\sf P}^n$, we write $b\triangleleft a$ if $b_i\neq I\Rightarrow a_i = b_i$ for all $i\in[n]$, calling $b$ majorized by $a$. 
For any $S\subseteq[n]$, we write $a\sim S$ if $a_i\neq I\Rightarrow i\in S$. For any $\mbb S\subseteq 2^{[n]}$, we write $a\sim\mbb S$ if $a\sim S$ for at least one $S\in\mbb S$.

Let $\Omega$ be a subset of $2^{[n]}$ that satisfies $\forall \nu\in\Omega$, $\forall \varnothing\subsetneq \mu\subseteq \nu$, $\mu\in\Omega$.
We call such $\Omega$ a \emph{factor set}.
We also use $\Omega_*$ to denote the set of maximal factors within $\Omega$, i.e., those that do not belong to any other factor. As an example, $\Omega = \{\{1\},\{2\},\{3\},\{1,2\},\{2,3\}\}$ is a valid factor set, with the corresponding $\Omega_* = \{\{1,2\},\{2,3\}\}$. Obviously, $\Omega$, $\Omega_*$ uniquely determines each other.

\begin{definition}\label{de:omega_local_channel}
    An $n$-qubit Pauli channel $\Lambda$ is $\Omega$-local if its Pauli eigenvalues satisfies
    \begin{equation}
        \lambda_a = \prod_{b\triangleleft a,b\sim \Omega}\exp(-r_b),\quad \forall a\neq I_n,
    \end{equation}
    for some $\{r_b\in\mbb R:b\sim\Omega, b\neq I_n\}$ and a factor set $\Omega$.
\end{definition}
\noindent For notational simplicity, we sometimes omit the requirement that $b\neq I_n$ and set $r_{I_n}=0$.

\medskip

A few remarks regarding this definition: First, all $n$-qubit Pauli channels (with all Pauli eigenvalues being positive) are $2^{[n]}\backslash\varnothing$-local. 
Note that we do not require $r_b$ to be non-negative.
Second, this definition basically says that the eigenvalues $\{\lambda_a\}_{a}$ factorize according to a Markov Random Field (MRF)~\cite{clifford1990markov} specified by the factor decomposition $\Omega$. This is opposed to \cite{Flammia2020} where Pauli channels with error rates $\{p_a\}_a$ factorized according to an MRF are considered. It is unclear whether these two models are equivalent, even approximately. 
Alternatively, our definition of $\Omega$-local is equivalent to requiring that $\Lambda$ can be expressed as a concatenation of Pauli diagonal maps\footnote{
Note that, we do not require those Pauli diagonal maps to be completely positive. Similarly, when we say our quasi-local model is equivalent to the sparse Pauli-Lindblad model, we have not required the positivity of the model parameters.
} acting on each subsystem from $\Omega$, as proven in~\cite{wagner2023learning} (In fact, this is the picture we will use in most figures shown in this section). One can also show that the sparse Pauli-Lindblad model used in~\cite{van2023probabilistic,kim2023evidence,hu2024demonstration} is equivalent to our definition.
Finally, throughout this paper we assume $\Omega$ is known a priori, from, e.g., knowledge about the device topology, leaving the topic of structure learning (see e.g. \cite{rouze2023efficient}) for future study. We provide more details about the quasi-local Pauli noise models in Appendix~\ref{app:factorize}.

\subsection{Learnability of quasi-local noise model}

\begin{figure}[!t]
    \centering
    \includegraphics[width=0.55\linewidth]{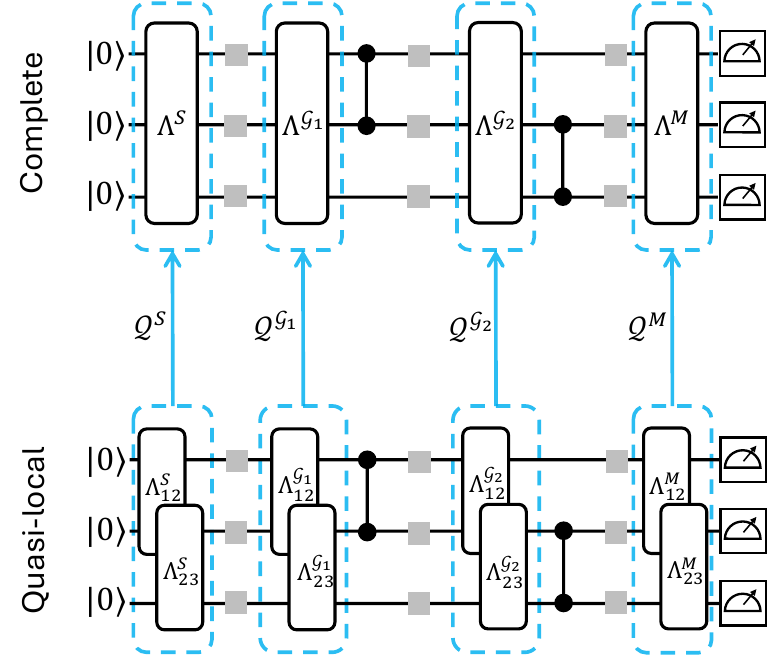}
    \caption{Example of a $3$-qubit quasi-local noise model and the embedding into a complete model. Here, the noise channels for state prep., meas., and all gates are assumed to be $\Omega$-local where $\Omega^*=\{\{1,2\},\{2,3\}\}$.
    In other words, $\Lambda^S = \Lambda^S_{12}\circ\Lambda^S_{23}$ where $\Lambda^S_{12},~\Lambda^S_{23}$ are Pauli diagonal maps supporting on $\{1,2\},~\{2,3\}$, respectively. Similar for other Pauli noise channels.}
    \label{fig:quasi_local_embedding}
\end{figure}

Note that, in terms of $x_a$, Definition~\ref{de:omega_local_channel} gives that $x_a = \sum_{b\triangleleft a,b\sim\Omega}r_b$.
This motivates the following definition of a quasi-local Pauli noise model. Note that, we have taken into account that the SPAM noise channels are symmetric Pauli channels (i.e., eigenvalues only depend on the Pauli pattern). An example of quasi-local noise model is given in Fig.~\ref{fig:quasi_local_embedding}
\begin{definition}[Quasi-local noise]\label{de:quasi_local_model}%
    Consider a layer-uncorrelated noise model $(X_R,\mc Q)$:
    \begin{itemize}
        \item If the state preparation noise parameters are given by $\{r_{\nu}^S:\nu\in\Omega^S\}$, where $\Omega^S$ is a factor set, such that
        \begin{equation}
            \mc Q(\bm\vartheta_\nu^S) = \sum_{\mu\neq \varnothing}\mathds 1 [\nu\subseteq \mu] \bm e^S_\mu,\quad\forall \nu\in\Omega^S, \nu\neq \varnothing,
        \end{equation}
        we say the model has \textbf{$\Omega^S$-local state preparation} noise. Similar for the measurement.
        \item If the noise parameters for some $\mc G\in\mf G$ are given by $\{r_{b}^{\mc G}:b\sim\Omega^{\mc G},b\neq I_n\}$, where $\Omega^{\mc G}$ is a factor set, such that
        \begin{equation}
            \mc Q(\bm\vartheta_b^{\mc G}) = \sum_{a\neq I_n}\mathds 1 [b\triangleleft a] \bm e^{\mc G}_a,\quad\forall b\sim\Omega^{\mc G}, b\neq I_n,
        \end{equation}
        We say the model has \textbf{$\Omega^{\mc G}$-local noise for $\mc G$}. 
        \item If the state preparation, measurements, and all layers have quasi-local noise with locality parameters $\{\Omega^S,\Omega^M,\Omega^{\mc G}:\mc G\in\mf G\}$, we say the model is quasi-local with the corresponding parameters.
    \end{itemize}
\end{definition}

The following lemma describes the transformation between the reduced and the complete parameters within a quasi-local model.
\begin{lemma}\label{le:quasi_local_trans}
    If $(X_R,\mc Q)$ satisfies $\Omega^S$-local state preparation noise, then
    \begin{equation}
    \left\{
    \begin{aligned}
        x_\mu^S &= \sum_{\substack{\nu\subseteq\mu,~\nu\in\Omega^S}}r_\nu^S,\quad\forall\mu\in2^{[n]}\backslash\varnothing.\\
        r_\nu^S &= \sum_{\substack{\mu\subseteq\nu,~\mu\in\Omega^S}}(-1)^{|\nu|-|\mu|}x_\mu^S,\quad \forall\nu\in\Omega^S.
    \end{aligned}
    \right.
    \end{equation}
    Similar for the measurement side;  If $(X_R,\mc Q)$ satisfies $\Omega^{\mc G}$-local noise for a gate $\mc G$, then
    \begin{equation}
    \left\{
    \begin{aligned}
        x_a^{\mc G} &= \sum_{\substack{b\triangleleft a,~b\sim\Omega^{\mc G}}}r_b^{\mc G},\quad\forall a\in{\Pn}\backslash I_n.\\
        r_b^{\mc G} &= \sum_{\substack{a\triangleleft b,~a\sim\Omega^{\mc G}}}(-1)^{w(b)-w(a)}x_a^{\mc G},\quad\forall b\sim\Omega^{\mc G},~b\neq I_n.
    \end{aligned}
    \right.
    \end{equation}
\end{lemma}
\begin{proof}
The transformation from $r$'s to $x$'s can be obtained by expanding $\bm x = \mc Q(\bm r)$ in the standard basis and compare coefficients. The inverse transformation can be verified by substitution. As a side comment, these pairs of transformation are known as the M\"obius transformation. They are also used in~\cite{wagner2023learning} to study Pauli noise learning in the logical level.
\end{proof}

\medskip

As a sanity check, one can retrieve the fully local model in Definition~\ref{de:fully_local} by setting
\begin{equation}
    \Omega^S = \Omega^M = \{\{j\}:j\in[n]\},\quad \Omega^{\mc G} = \{\nu: \varnothing\neq\nu\subseteq\mr{supp}(\mc G)\},~\forall\mc G\in\mf G.
\end{equation}

Now, we make two important remarks about the gate noise models. First, we will take a more general approach to deal with parallel gates than the fully local case. Instead of letting each $\mc G$ to have a constantly small support, we allow it to be a tensor product of many few-qubit gates, i.e., $\mc G = \hat{\mc G}_1\otimes\cdots\otimes\hat{\mc G}_M\otimes\id$, where the identity acts on $[n]\backslash\bigcup_i \supp(\mc G_i)$. 
Second, note that a Clifford gate affected by Pauli noise can be modeled in two equivalent ways: whether the noise channel happens before or after the gate, i.e.,
\begin{equation}
    \tilde{\mc G} = \mc G\circ\Lambda = \Lambda'\circ\mc G,\quad\text{where}~\Lambda' = \mc G\circ\Lambda\circ\mc G^\dagger.
\end{equation}
Throughout this work, we have been adopting the noise-before-gate convention, but now it is worth looking at the other convention.
Specifically, when $\Lambda$ is an $\Omega$-local Pauli channel, $\Lambda'$ might not be $\Omega$-local in general. This raised the following question: is there a physical reason to believe the gate noise has certain locality constraint in the noise-before-gate convention, but not in the noise-after-gate convention? To the best of our knowledge, no physical justification of this kind has been established in the literature. To circumvent the dilemma, we propose to study a specific type of locality assumption:
\begin{definition}\label{de:covariant}
    An $\Omega$-local noise model is called \textbf{$\mc G$-covariant}, if for any $\Omega$-local Pauli channel $\Lambda$,
    $\mc G\circ\Lambda\circ\mc G^\dagger$ and $\mc G^\dagger\circ\Lambda\circ\mc G$ are also $\Omega$-local Pauli channels.
\end{definition}

\begin{figure}[!t]
    \centering
    \includegraphics[width=0.9\linewidth]{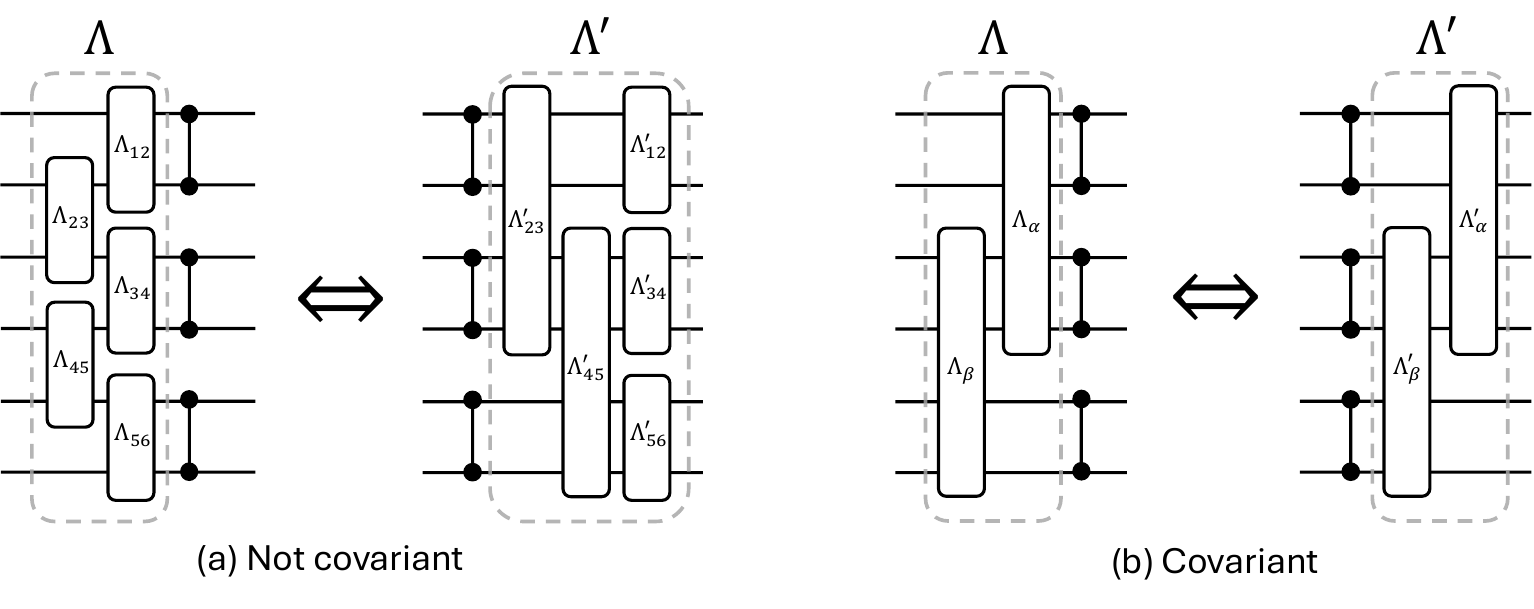}
    \caption{Examples of quasi-local gate noise models that are incovariant and covariant, respectively. Here we consider a $6$-qubit system and let $\mc G=\mr{CZ}^{\otimes 3}$. (a) The nearest-neighbor $2$-local noise model ($\Omega^*=\{\{1,2\},\{2,3\},\cdots,\{5,6\}\}$)
    \label{fig:covar_incovar} which is not $\mc G$-covariant. Indeed, if we commute the noise channel from before-$\mc G$ (left) to after-$\mc G$ (right),  certain $2$-body noise terms (e.g., $\Lambda_{23}$) will generically be propagated to be $4$-body, making the whole noise channel no longer $\Omega$-local; (b) A $4$-local noise model ($\Omega^*=\{\alpha=\{1,2,3,4\},\beta=\{3,4,5,6\}\}$) which is $\mc G$-covariant. Indeed, commuting the noise channel from before-$\mc G$ to after-$\mc G$ preserves the $\Omega$-local structure. 
    Note that all Pauli diagonal maps commute with each other.
    }
\end{figure}

\noindent Since a $\mc G$-covariant noise model has the same locality constraints in either convention, it is much easier to justify their physical relevance. An example of covariant vs. incovariant noise model is given in Fig.~\ref{fig:covar_incovar}.
When $\mc G$ represents a parallel application of many Clifford gates with disjoint supports, we have the following sufficient condition for a noise model to be $\mc G$-covariant:
\begin{lemma}\label{le:covariant_sufficient}
    Given $\mc G=\left(\bigotimes_{i=1}^M\hat{\mc G}_i\right)\otimes\id$, define the extended support map $\Xi_{\mc G}:[n]\mapsto[n]$ as
    \begin{equation}
        \Xi_{\mc G}(s) \coleq \left(\bigcup_{\supp(\hat{\mc G}_i)\cap s\neq\varnothing}\supp(\hat{\mc G}_i)\right)\cup s.
    \end{equation}
    For an $\Omega$-local noise model, if $\forall s\in\Omega$, $\Xi_{\mc G}(s)\in\Omega$, then $\Omega$ is $\mc G$-covariant.
\end{lemma}
\noindent Clearly, for a Pauli $P$ with support $s$, the support of $\mc G(P)$ cannot exceed $\Xi_{\mc G}(s)$. Lemma~\ref{le:covariant_sufficient} roughly says that if for all Pauli $P\sim\Omega$, one has $\mc G(P)\sim\Omega$, then $\Omega$ must be $\mc G$-covariant. The proof is given in Appendix.~\ref{app:factorize}.

\medskip

We are ready to present our main results on the learnability of quasi-local Pauli noise model.
Perhaps not surprisingly, the gauge degrees of freedom can still be described using the SDG's in Definition~\ref{de:sdg}, as follows:
\begin{theorem}\label{th:quasi_learnability}
Given a layer-uncorrelated noise model $(X_R,\mc Q)$,
\begin{enumerate}[label=(\alph*)]
    \item If the model has $\Omega^S$-local state preparation noise, then 
    \begin{equation}\label{eq:quasi_local_spam}
        \{\bm z\in T:\bm z^S\in\imq^S\} = \mr{span}\{\mf d_\nu:\nu\in \Omega^S\}.
    \end{equation}
    Similarly for the measurement noise. 
    \item If the model has $\Omega^{\mc G}$-local noise for some $\mc G\in\mf G$ and is $\mc G$-covariant, then
    \begin{equation}\label{eq:quasi_local_gate}
        \{\bm z\in T:\bm z^{\mc G}\in\imq^{\mc G}\}\supseteq \mr{span}\{\mf d_\nu:\nu\in\Omega^{\mc G}~\textrm{or}~\Xi_{\mc G}(\nu) = \nu\}.
    \end{equation}
    \item If $\Omega^S\subseteq\Omega^M$, $\Omega^{\mc G}$ is $\mc G$ covariant for all $\mc G\in\mf G$, and that $\forall \nu\in\Omega^S,\forall\mc G\in\mf G$, either $\nu\in\Omega^\mc G$ or $\Xi_{\mc G}(\nu)=\nu$, then the embedded gauge space for the whole model is given by
    \begin{equation}
        \mc Q(T_R) = \mr{span}\{\mf d_\nu:\nu\in\Omega^S\}.
    \end{equation}
\end{enumerate}
\end{theorem}

Theorem~\ref{th:quasi_learnability}(c) confirms that, given certain ``goodness'' assumptions of the quasi-local noise model, the embedded gauge space is indeed spanned by a subset of SDG's consistent with the factor sets. Similar to the fully local case, one can efficiently compute the reduced gauge space $T_R = \mr{span}\{\mc Q^{-1}(\mf d_\nu):\nu\in\Omega^S\}$ by applying the M\"obius inverse transformation given in Lemma~\ref{le:quasi_local_trans}.

\medskip

\begin{proof}[Proof of Theorem~\ref{th:quasi_learnability}(a)]
    Let ${\pi}^S: X\mapsto X^S$ be the projection map ${\pi}^S(\bm x) = \bm x^S$. 
    Recall from the proof of Theorem~\ref{th:fully_learnability}(a) that ${\pi}^S$ is an isomorphism between $T$ and $X^S$. 
    We claim the following two linear spaces are equal.
    \begin{equation}\label{eq:quasi_local_thm1_step1}
        \mr{span}\{\mc Q^S(\bm \vartheta_\nu^S):\nu\in\Omega^S\} = \mr{span}\{{\pi}^S(\mf d_\nu):\nu\in\Omega^S\}.
    \end{equation}
    Now we prove this claim. By definition, for any $\nu\in\Omega^S$ we have
    \begin{align}
        \mc Q^S(\bm\vartheta_\nu^S) &= \sum_{\mu\neq\varnothing}\mathds 1[\nu\subseteq\mu]\bm e_{\mu}^S \eqcol\bm\alpha_\nu,\\
        {\pi}^S(\mf d_\nu) &= \sum_{\mu\neq\varnothing}\mathds 1[\nu\cap\mu\neq\varnothing]\bm e_\mu^S \eqcol\bm\beta_\nu.
    \end{align}
    Since both $\{\bm\alpha_\nu\},\{\bm\beta_\nu\}$ have the same cardinality, and we know $\{\bm\beta_\nu\}$ are linearly-independent thanks to the linear independence of the SDG's and ${\pi}^S$ being an isomorphism from $T$ to $X^S$, we just need to show  $\{\bm\alpha_\nu\}$ can linearly represent $\{\bm\beta_\nu\}$. We claim the following holds,
    \begin{equation}\label{eq:quasi_linear_indep}
        \bm\beta_\nu = -\sum_{\nu'\in\Omega^{S}}\mathds 1[\nu'\subseteq\nu](-1)^{|\nu'|}\bm\alpha_{\nu'},\quad\forall \nu\in\Omega^S.
    \end{equation}
    Indeed, expanding the RHS yields,
    \begin{equation}
        \begin{aligned}
            R.H.S.&= -\sum_{\mu\neq\varnothing}\bm e_\mu^S\sum_{\nu'\in\Omega^S}\mathds 1[\nu'\subseteq\nu]\mathds 1[\nu'\subseteq\mu](-1)^{|\nu'|}\\
            &= -\sum_{\mu\neq\varnothing}\bm e_\mu^S\sum_{\nu'\in\Omega^S}\mathds 1[\nu'\subseteq\nu\cap\mu](-1)^{|\nu'|}\\
            &= -\sum_{\mu\neq\varnothing}\bm e_\mu^S\sum_{\substack{\nu'\subseteq\nu\cap\mu,\\\nu'\neq\varnothing}}(-1)^{|\nu'|}\\
            &= \sum_{\mu\neq\varnothing}\bm e_\mu^S\mathds 1[\nu\cap\mu\neq\varnothing] = \bm\beta_\nu,
        \end{aligned}
    \end{equation}
    where the third line uses the defining property of $\Omega^S$ that for any $\nu\in\Omega^S$, any nonempty $\nu'\subseteq\nu$ also satisfies $\nu'\in\Omega^S$; The last line uses the binomial theorem. This completes the proof of Eq.~\eqref{eq:quasi_local_thm1_step1}.
    Consequently, we have
    \begin{equation}
        \begin{aligned}
            \{\bm z\in T:\bm z^S\in\imq^S\} &= \{\bm z\in T: {\pi}^S(\bm z)\in\mr{span}\{\mc Q^S(\bm\vartheta_\nu^S):\nu\in\Omega^S\}\}\\
            &= \{\bm z\in T: {\pi}^S(\bm z)\in\mr{span}\{{\pi}^S(\mf d_\nu):\nu\in\Omega^S\}\}\\
            &=\mr{span}\{\mf d_{\nu}:\nu\in\Omega^S\}.
        \end{aligned}
    \end{equation}
    The second line uses Eq.~\eqref{eq:quasi_local_thm1_step1}. The last line uses ${\pi}^S$ being an isomorphism between $T$ and $X^S$.
    The proof works similarly for $\Omega^M$-local measurement noise.
\end{proof}

\begin{proof}[Proof of Theorem~\ref{th:quasi_learnability}(b)]
Let ${\pi}^{\mc G}:X\mapsto X^{\mc G}$ be the projection map ${\pi}^{\mc G}(\bm x) = \bm x^{\mc G}$.
Recall that for all $a\neq I_n$, $\nu\neq\varnothing$,
\begin{equation}
\begin{aligned}
    \bm e_a^{\mc G}\cdot \mf d_\nu &= \mathds 1[\pt(\mc G(a))_\nu\neq 0_\nu] - \mathds 1[\pt(a)_\nu\neq 0_\nu],\\
    &= -\log\left(\lbra{\sigma_a}\mc G^\dagger \mc D_\nu\mc G \mc D_\nu^{-1}\lket{\sigma_a}\right).
\end{aligned}
\end{equation}
Here 
$\mc D_\nu$ is the subsystem depolarizing channel on $\nu$ defined in Eq.~\eqref{eq:sd_channel}.
\begin{itemize}
    \item If $\Xi_{\mc G}(\nu) = \nu$, then $\mc G$ is a tensor product of gates on $\nu$ and on $[n]\backslash\nu$. Consequently, $\mc G$ commutes with $\mc D_\nu$. Thus, $\bm e_a^{\mc G}\cdot\mf d_\nu = 0$ for all $a$, which means ${\pi}^{\mc G}(\mf d_\nu)\in\imq^{\mc G}$ trivially holds. See Fig.~\ref{fig:covariant_gauge}(ii) for an example.
    \item If $\nu\in\Omega^{\mc G}$, one can see that both $\mc D_\nu$ and $\mc D_\nu^{-1}$ are $\Omega^{\mc G}$-local Pauli channels. Since the noise model is $\mc G$-covariant, $\mc G^\dagger\mc D_\nu\mc G$ is $\Omega^{\mc G}$-local. 
    Since the concatenation of two $\Omega^{\mc G}$-local Pauli channels is clearly also $\Omega^{\mc G}$-local, $\mc G^\dagger\mc D_\nu\mc G\mc D_\nu^{-1}$ is $\Omega$-local. This means,
    \begin{equation}
        \bm e_a^{\mc G}\cdot\mf d_\nu = \sum_{\substack{b\triangleleft a,b\sim\Omega^{\mc G}}}r'_b,
    \end{equation}
    for some coefficients $\{r'_b\in\mbb R:b\sim\Omega^{\mc G},b\neq I_n\}$. Comparing this to the definition of quasi-local gate noise, we see that ${\pi}^{\mc G}(\mf d_\nu)\in\imq^{\mc G}$. See Fig.~\ref{fig:covariant_gauge}(i) for an example.
\end{itemize}
Combining the above two cases and using linearity, we complete the proof for Theorem~\ref{th:quasi_learnability}(b).
\end{proof}

We make a few remarks about Theorem~\ref{th:quasi_learnability}(b). First, it is clear from the proof that $\mr{span}\{\mf d_\nu:\Xi_{\mc G}(\nu)=\nu\}$ are gauges that pass trivially through $\mc G$ leaving its noise parameters unchanged, while $\mr{span}\{\mf d_\nu:\nu\in\Omega^{\mc G}\}$ are gauges that non-trivially transform the noise parameters of $\mc G$;
Second, here we do not have a converse side as in Theorem~\ref{th:fully_learnability}(b), meaning there could be more allowed gauges than what we have characterized. See Fig.~\ref{fig:covariant_gauge} for an example of different types of gauges. We do not pursue such a tight characterization in the quasi-local case as it looks considerably more complicated. For our purpose of proving Theorem~\ref{th:quasi_learnability}(c), the current results are sufficient; Finally, the requirement of noise being $\mc G$-covariant is indispensable for the theorem to hold. In Fig.~\ref{fig:incovariant_gauge} we show a counterexample where $\mf d_\nu$ is not an allowed gauge even if $\nu\in\Omega^{\mc G}$, when the gate noise model is not $\mc G$-covariant. This will be further discussed in Sec.~\ref{sec:CZ} when we encounter some realistic examples.

\begin{figure}[!tp]
    \centering
    \includegraphics[width=0.85\linewidth]{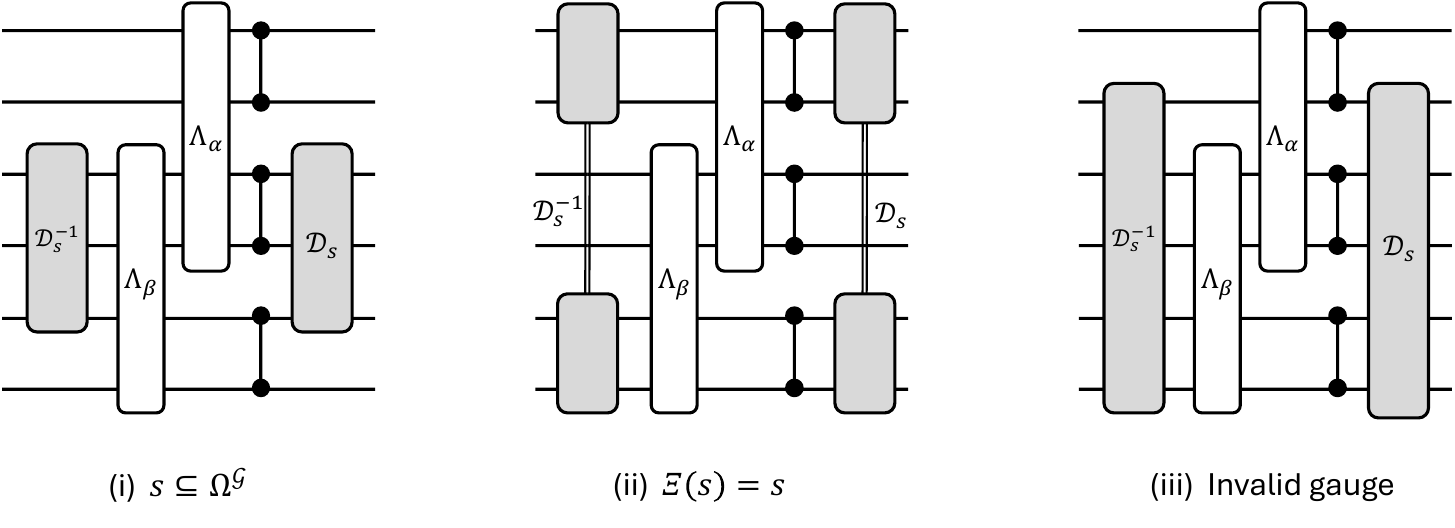}
    \caption{Examples of three different cases of SDG's ($\mc D_s$) acting on $\mc G=\mr{CZ}^{\otimes 3}$ with covariant $\Omega$-local noise model, $\Omega^* = \{\{1,2,3,4\},\{3,4,5,6\}\}$. Case (i) and (ii) are allowed gauges as shown in Theorem~\ref{th:quasi_learnability}(b). Specifically, In Case (i) we have $s=\{3,4,5\}\in\Omega$, thus $\mc D_s$ is a valid gauge that non-trivially transform the noise channel; In Case (ii) we have $s=\{1,2,5,6\} = \Xi_{\mc G}(s)$, or in words, any Pauli supports within $s$ still supports within $s$ after the action of $\mc G$ (see Lemma~\ref{le:covariant_sufficient}).
    Here, even though $s\notin\Omega$, $\mc D_s$ commutes through $\mc G$ leaving the noise channel unchanged, and is thus a trivially valid gauge; Case (iii) shows an invalid gauge ($s=\{2,3,4,5,6\}$), which would generate a $6$-body noise term, violating the locality assumption.}
    \label{fig:covariant_gauge}
\end{figure}
\begin{figure}[!tp]
    \centering
    \includegraphics[width=0.5\linewidth]{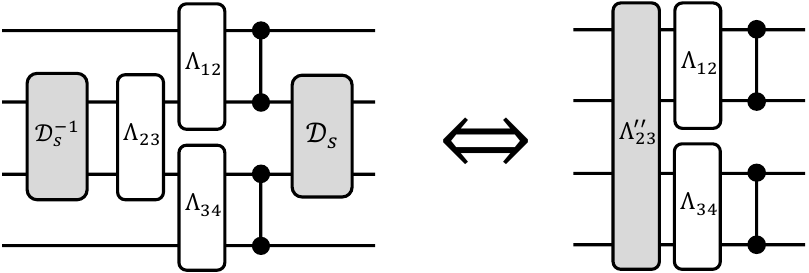}
    \caption{Example of SDG acting on $\mc G=\mr{CZ}^{\otimes 2}$ with nearest-neighbor $2$-local noise model (i.e., $\Omega^*=\{\{1,2\},\{2,3\},\{3,4\}\}$) that is \emph{not} $\mc G$-covariant. In this case, though $s=\{2,3\}\in\Omega$, The SDG $\mc D_s$ would create a $4$-body noise term (shown in the right) violating the $\Omega$-local assumption, and is thus not an allowed gauge. This highlights the necessity of $\mc G$-covariance in obtaining Theorem~\ref{th:quasi_learnability}(b).}
    \label{fig:incovariant_gauge}
\end{figure}

\begin{proof}[Proof of Theorem~\ref{th:quasi_learnability}(c)]
    Thanks to Lemma~\ref{le:layerwise_indep}, $\mc Q(T_R)$ is given by the intersection of Eq.~\eqref{eq:quasi_local_spam} for SPAM and Eq.~\eqref{eq:quasi_local_gate} for all $\mc G\in\mf G$. Under the given conditions, it is easy to see that Eq.~\eqref{eq:quasi_local_spam} is always a subset of Eq.~\eqref{eq:quasi_local_gate} for any $\mc G\in\mf G$, thus their intersection is simply given by Eq.~\eqref{eq:quasi_local_spam}.
\end{proof}

\subsection{Efficient learning of quasi-local noise model}

We now discuss algorithms for learning a quasi-local noise model. In this section, we will assume all conditions of Theorem~\ref{th:quasi_learnability}(c) are satisfied. That is, $\Omega^S\subseteq\Omega^M$; $\Omega^{\mc G}$ is $\mc G$-covariant for all $\mc G\in\mf G$; and that $\forall \nu\in\Omega^S,\forall\mc G\in\mf G$ either $\nu\in\Omega^{\mc G}$ or $\Xi_{\mc G}(\nu)=\nu$. For such a well-conditioned quasi-local noise model, we can easily find a reduced cycle basis for $L_R$ as follows:

Consider the following subset of the rooted cycle basis ${\mathsf B}$ defined in Eq.~\eqref{eq:rooted_cycle_basis}, denoted by $\mathsf B''$,
\begin{equation}
    {\mathsf B}''\coleq\{\bm e_\mu^S+\bm e_\mu^M:\mu\in\Omega^M\}\cup\{\bm e_{\supp(a)}^S+\bm e_a^{\mc G}+\bm e_{\supp(\mc G(a))}^M:\mc G\in\mf G,~a\sim\Omega^{\mc G}\}.
\end{equation}
\begin{lemma}
    Let $\mathsf B_R\coleq\mc Q^\T(\mathsf B'')$. Then $\mathsf B_R$ forms a reduced cycle basis for $L_R$.
\end{lemma}
\begin{proof}
First, acting $\mc Q^\T$ on the rooted cycle basis ${\mathsf B}$ defined in Lemma~\ref{le:rooted_cycle_basis} yields,
\begin{equation}
\begin{aligned}
    \mc Q^\T(\bm e_\mu^S+\bm e_\mu^M) &= \sum_{\substack{\nu\subseteq\mu,\\\nu\sim \Omega^S}}\bm\vartheta_\nu^S + \sum_{\substack{\nu\subseteq\mu,\\\nu\sim \Omega^M}}\bm\vartheta_\nu^M,\quad&& \forall\mu\neq\varnothing,\\
    \mc Q^\T(\bm e_{\supp(a)}^S+\bm e_a^{\mc G}+\bm e_{\supp(\mc G(a))}^M)&= \sum_{\substack{\nu\subseteq\supp(a),\\\nu\sim \Omega^S}}\bm\vartheta_\nu^S + 
    \sum_{\substack{b\triangleleft a,\\b\sim\Omega^{\mc G}}}\bm\vartheta_b^{\mc G}
    +\sum_{\substack{\nu\subseteq\supp(\mc G(a)),\\\nu\sim \Omega^M}}\bm\vartheta_\nu^M,\quad&& \forall a\neq I_n,~\forall\mc G\in\mf G.
\end{aligned}
\end{equation}
Thus, ${\mathsf B}_R\coleq \mc Q^\T({\mathsf B}'')$ takes the following form,
\begin{equation}
\begin{aligned}
    {\mathsf B}_R &=\left\{\sum_{\substack{\nu\subseteq\mu,\\\nu\sim \Omega^S}}\bm\vartheta_\nu^S + \sum_{\substack{\nu\subseteq\mu}}\bm\vartheta_\nu^M: \mu\in\Omega^M\right\}
    \cup\\
    &\left\{
    \sum_{\substack{\nu\subseteq\supp(a),\\\nu\sim \Omega^S}}\bm\vartheta_\nu^S + 
    \sum_{\substack{b\triangleleft a}}\bm\vartheta_b^{\mc G}
    +\sum_{\substack{\nu\subseteq\supp(\mc G(a)),\\\nu\sim \Omega^M}}\bm\vartheta_\nu^M: \mc G\in\mf G,~a\sim \Omega^{\mc G}
    \right\},
\end{aligned}
\end{equation}
where we have used the defining property of $\Omega^{M}$ and $\Omega^{\mc G}$. 
let us first prove the linear independence of $\mathsf B_R$.
Consider a linear combination of all vectors from ${\mathsf B}_R$ that yields $\bm 0$. Since only the second part of ${\mathsf B}_R$ contains gate noise parameters, and one can show $\{\sum_{b\triangleleft a}\bm\vartheta_b^{\mc G}:\mc G\in\mf G,~a\sim\Omega^{\mc G}\}$ are linearly-independent (as they can represent  $\{\bm\vartheta_a^{\mc G}:\mc G\in\mf G,a\sim\Omega^{\mc G}\}$ using the inclusion-exclusion principle, see Appendix~\ref{app:factorize} for similar expressions), the coefficients for all vectors from the second part must be zero; On the other hand, because $\{\sum_{\nu\subseteq\mu}\bm\vartheta_\nu^M:\nu\in\Omega^M\}$ are also linearly-independent (also thanks to the inclusion-exclusion principle), the coefficients for vectors from the first part must also be zero. This means ${\mathsf B}_R$ must be linearly independent.
On the other hand, Theorem~\ref{th:quasi_learnability}(c) implies $\dim(T_R)=|\Omega^S|$. We can see from the expression of ${\mathsf B}_R$ that $|{\mathsf B}_R| = \dim(X_R) - |\Omega^S| = \dim(X_R)-\dim(T_R) =\dim(L_R)$. We can thus conclude that ${\mathsf B}_R$ forms a basis for $L_R$.
\end{proof}

The construction of ${\mathsf B}_R$ implies the following set of experiments can fully learn the quasi-local reduced noise model,
\begin{enumerate}
    \item For all $\mu\in\Omega^{M}$, prepare $\tilde\rho_0$ and measure $\tilde {Z}^\mu$.
    \item For all $\mc G\in\mf G$, $b\sim\Omega^{\mc G}$, prepare $\tilde{\rho}_b$, apply $\tilde{\mc G}$, and measure $\tilde P_{{\mc G}(b)}$.
\end{enumerate}
If $\Omega^M$ and $\Omega^{\mc G}$ have polynomial size, the number of experiments needed is also polynomial, and is thus efficient. Besides, if each $\mc G$ is a parallel application of many Clifford gates each acting on a constant number of qubits, many of the above experiments can be conducted in parallel. We leave such an optimization of experiment design as future research.

\medskip

Of course, one might also want to learn gate noise parameter to relative precision in the quasi-local noise model. 
We will not pursue a general efficient algorithm for finding reduced cycle basis for $L_R^{\mc G}$ in the quasi-local case due to its complexity. Instead, we will study one practically-interesting examples in the following section.

\section{Case study: CZ gates with 1D topology}\label{sec:CZ}

In this section, we will apply our theory to a concrete gate set where the Control-Z (CZ) gates are the only multi-qubit entangling gates. Such a gate set, augmented with single-qubit gates, are able to conduct universal quantum computation, and is adopted by state-of-the-art experimental platforms (e.g., IBM Quantum~\cite{ibmquantum}\footnote{
Note that the Control-NOT (CNOT) and the Echoed Cross-Resonance (ECR) gates are all equivalent to CZ up to single-qubit rotations. Our theory can similarly be applied to them.
}).
We will discuss the noise learnability and learning algorithms with several different noise ansatz. We hope the results presented here can not only illustrate our theoretical results, but also serve as a new noise characterization protocol for 
practical use. A quick summary of results is given in Table~\ref{tab:CZ}. 

\begin{table}[!htp]
    \centering
    \begin{tabular}{|c|c|c|c|c|c|c|c|c|}
         \hline
         Section & $n$ & $\mf G$ & Noise model & $\dim(X_R)$  & $\dim(L_R)$ & $\dim(X_R^{G})$ & 
         $\dim(L_R^{G})$\\
         \hline
         \ref{sec:cz1} & $2$ & $\{\mr{CZ}\}$ & complete &$21$  & $18$ & $15$ & $13$\\
         \hline
         \ref{sec:cz2} & $n\ge 4$ & $\{\mr{CZ}_{i,i+1}\}$ & fully local & $17n$  & $16n$ & $15n$ & $14n$\\
         \hline
         \ref{sec:NN_CZ} & $n\ge 6$ & $\{\mr{CZ}_e,\mr{CZ}_o\}$ & nearest-neighbor &$28n$  & $27n$ & $24n$ & $23n$\\
         \hline
         \ref{sec:cz_covar} & $n\ge 6$ & $\{\mr{CZ}_e,\mr{CZ}_o\}$ & covariant $4$-local &$244n$ & $242n$ & $240n$ & $-$\\
         \hline
    \end{tabular}
    \caption{Summary of results of this section. Each row corresponds to a gate set and noise model studied in one subsection.}
    \label{tab:CZ}
\end{table}

\subsection{Single CZ, complete noise}\label{sec:cz1}
We begin with a minimal example: $n=2$, $\mf G = \{\mr{CZ}\}$. The goal is to learn the complete Pauli noise model. Recall that the associated PTG for this model is given in Fig.~\ref{fig:cz_ptg}. Thanks to Theorem~\ref{th:complete}, the parameter space $X$, gauge space $T$, and learnable space $L$ are given by (let $\mc G=\mr{CZ}$ for notational simplicity),
\begin{itemize}
    \item $X$: $\dim(X)=21$. Basis: $\{\bm e_u^S,\bm e_u^M,\bm e_a^{\mc G}:u\neq 00, a\neq II\}$. 
    \item $T$: $\dim(T)=3$. Basis: $\{\mf d_{10},\mf d_{01},\mf d_{11}\}$. 
    \item $L$: $\dim(L)=18$. Cycle basis: $\left\{\bm e^S_u+\bm e^M_u:u\in\{01,10,11\}\right\}\cup\left\{\bm e_{\pt(a)}^S+\bm e_a^{\mc G}+\bm e^M_{\pt(\mc{G}(a))}:a\in{\sf P}^2\backslash{II}\right\}$. 
\end{itemize}
The learnable space for gate noise $L^{G}\coleq L\cap X^{G}$ (which is the cycle space for the subgraph of PTG with edges only for gate noise parameters) is spanned by the following cycle basis
\begin{equation}\label{eq:1CZ_gate}
\begin{gathered}
    \left\{\bm e_{IZ}^{\mc{G}},\bm e_{ZI}^{\mc{G}},\bm e_{ZZ}^{\mc{G}},\bm e_{XX}^{\mc{G}},\bm e_{YY}^{\mc{G}},\bm e_{XY}^{\mc{G}},\bm e_{YX}^{\mc{G}},\right.\\\left.
    \bm e_{XI}^{\mc{G}}+\bm e_{XZ}^{\mc{G}},\bm e_{YI}^{\mc{G}}+\bm e_{YZ}^{\mc{G}},\bm e_{XI}^{\mc{G}}+\bm e_{YZ}^{\mc{G}},\bm e_{IX}^{\mc{G}}+\bm e_{ZX}^{\mc{G}},\bm e_{IY}^{\mc{G}}+\bm e_{ZY}^{\mc{G}},\bm e_{IX}^{\mc{G}}+\bm e_{ZY}^{\mc{G}}\right\}.
\end{gathered}
\end{equation}
We have $\dim(L^G)=13$. Note that $\dim(X^G)-\dim(L^G)=2$. This is because ${\pi}^{\mc G}(\mf d_{11})=\bm 0$, which means the global depolarizing gauge does not change gate noise parameters. Only the two single-qubit depolarizing gauges $\{\mf d_{01},\mf d_{10}\}$ changes gate parameters nontrivially. The above basis for $L^G$ can be augmented to a basis of $L$ by adding the following $5$ rooted cycles,
\begin{equation}
    \left\{\bm e^S_u+\bm e^M_u:u\in\{01,10,11\}\right\}\cup\left\{\bm e^S_{10}+\bm e^{\mc G}_{XI}+\bm e^M_{11},\bm e^S_{01}+\bm e^{\mc G}_{IX}+\bm e^M_{11}\right\}
\end{equation}

\medskip

We remark that Eq.~\eqref{eq:1CZ_gate} was previously given in~\cite{chen2023learnability}. The new insight from our theory is to include SPAM parameters as the target of learning, and to explicitly parameterize the gauge degrees of freedom. 
We also note that \cite{chen2023learnability} has presented experimental designs to learn all of these gate parameters to relative precision. 
Concretely, any of the above parameters can be learned by concatenating $\mr{CZ}$ or $\mr{CZ}\circ(\sqrt{Z}\otimes\sqrt{Z})$ an even number of time, input the correct Pauli eigenstate, and measurement in the correct Pauli eigenbasis. See Fig.~\ref{fig:all_exp} (a).

\begin{figure}[t]
    \centering
    \includegraphics[width=0.95\linewidth]{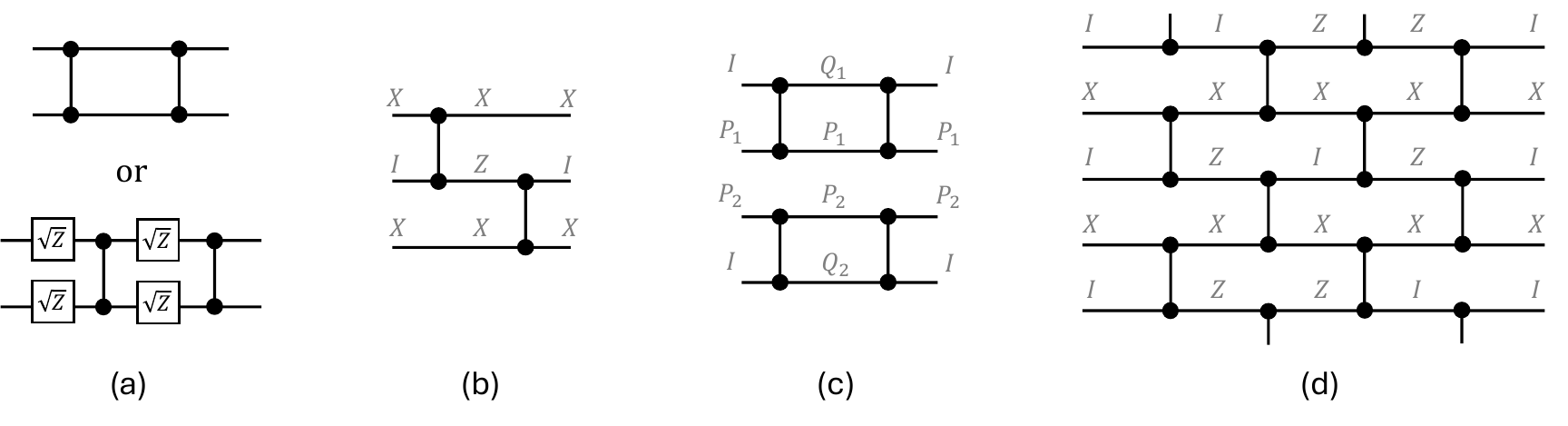}
    \caption{Circuits for learning different types of cycles.
    (a) Circuits for learning all cycles of a single CZ (and for certain cycles of the other models). Any cycles from Eq.~\eqref{eq:1CZ_gate} can be learned by inputting an appropriate Pauli eigenstate to the upper or lower circuit and measure the same Pauli observable.
    (b) Circuits for learning the 2-gate cycle from Eq.~\eqref{eq:fl_cz_2gate}. The evolution of the corresponding Pauli operators is shown in gray.
     (c) Circuits for learning cycles from Eq.~\eqref{eq:NN_CZ_basis2}. 
     (d) Circuits for learning for cycles from Eq.~\eqref{eq:NN_CZ_basis3}. 
     }
    \label{fig:all_exp}
\end{figure}

\subsection{CZ's on 1D ring, fully local noise}\label{sec:cz2}

Let $n>2$, consider the case where the $n$ qubits are placed in an 1D ring, and one can apply an CZ between any neighboring pair of qubits, i.e., $\mf G=\{\mr{CZ}_{i,i+1}:i\in[n]\}$, where the ($i+1$)-th qubit is the same as the first qubit. Suppose the noise model is fully local (c.f.~Definition~\ref{de:fully_local}), using Theorem~\ref{th:fully_learnability}, we obtain the following characterization for the learnability of this model (let $\mc G_i=\mr{CZ}_{i,i+1}$),
\begin{itemize}
    \item $X_R$: $\dim(X_R)=17n$. Basis: $\{\bm\vartheta_j^S,\bm\vartheta_j^M,\bm\vartheta_a^{\mc G_i}:j\in[n],i\in[n],a\in{\sf P}^2\backslash II\}$. 
    \item $T_R$: $\dim(T_R)=n$. Basis: $\{\mc Q^{-1}(\mf d_{\{i\}}):i\in[n]\}$. 
    \item $L_R$: $\dim(L_R)=16n$. Corresponding cycle basis: 
    $\{\bm e_{\{i\}}^S+\bm e_{\{i\}}^M:i\in[n]\}\cup\{\bm e_{\pt(a)}^S+\bm e_{a}^{\mc G}+\bm e_{\pt(\mc G_i(a))}^{M}:i\in[n],a\in{\sf P}^{\{i,i+1\}}\backslash I_n\}$.
\end{itemize}
Note that $\mc Q^{-1}$ is well-defined when acting within the image of $\mc Q$. For $L_R$, instead of writing down a reduced cycle basis, we write down the corresponding cycle basis in $L$ whose connection to experiment design is more direct. Acting $\mc Q^{\T}$ on that basis will give a reduced cycle basis for $L_R$.

\medskip

Using the algorithms introduced in Sec.~\ref{sec:local_learning}, we can construct a reduced cycle basis for the learnable spaces of gate parameters $L_R^G$. In fact, there exists a basis where each basis vector contains parameters of only one gate or two neighboring gates, as follows 
\begin{itemize}
    \item 1-gate (same as Eq.~\eqref{eq:1CZ_gate}): 
    \begin{equation}
    \begin{gathered}
    \left\{\bm\vartheta_{IZ}^{\mc{G}_i},\bm\vartheta_{ZI}^{\mc{G}_i},\bm\vartheta_{ZZ}^{\mc{G}_i},\bm\vartheta_{XX}^{\mc{G}_i},\bm\vartheta_{YY}^{\mc{G}_i},\bm\vartheta_{XY}^{\mc{G}_i},\bm\vartheta_{YX}^{\mc{G}_i},\right.\\\left.
    \bm\vartheta_{XI}^{\mc{G}_i}+\bm\vartheta_{XZ}^{\mc{G}_i},\bm\vartheta_{YI}^{\mc{G}_i}+\bm\vartheta_{YZ}^{\mc{G}_i},\bm\vartheta_{XI}^{\mc{G}_i}+\bm\vartheta_{YZ}^{\mc{G}_i},\bm\vartheta_{IX}^{\mc{G}_i}+\bm\vartheta_{ZX}^{\mc{G}_i},\bm\vartheta_{IY}^{\mc{G}_i}+\bm\vartheta_{ZY}^{\mc{G}_i},\bm\vartheta_{IX}^{\mc{G}_i}+\bm\vartheta_{ZY}^{\mc{G}_i}:i\in[n]\right\}.
\end{gathered}
\end{equation}
    \item 2-gate: 
    \begin{equation}\label{eq:fl_cz_2gate}
        \left\{ \bm\vartheta_{XI}^{\mc{G}_i}+\bm\vartheta_{ZX}^{\mc{G}_{i+1}}:i\in[n]  \right\}.
    \end{equation}
\end{itemize}
We have $\dim(L^G) = 14n$. The 1-gate basis vector can be learned by running the same learning protocol for a single CZ in each consecutive 2-qubit subsystem, shown in Fig.~\ref{fig:all_exp} (a); The 2-gate basis can be learned using the circuits shown in Fig.~\ref{fig:all_exp} (b) on each consecutive 3-qubit subsystem. 
By learning disjoint regions in parallel, a constant number of experiments is sufficient to learn all the gate noise parameters.
We also note that the above basis for $L_R^{G}$ can be augmented to a basis for $L_R$ by adding the following $2n$ reduced cycles,
\begin{equation}
    \left\{\bm\vartheta_j^S+\bm\vartheta_j^M:j\in[n]\right\}\cup\left\{\bm\vartheta_j^S+\bm\vartheta_{XI}^{\mc G_j}+\bm\vartheta_j^M+\bm\vartheta_{j+1}^M:j\in[n]\right\}.
\end{equation}

\subsection{CZ's on ring, nearest-neighbor noise}\label{sec:NN_CZ}

Let $n\ge6$ be an even number. We again consider $n$ qubits placed in an 1D ring, but this time let the gate set be $\mf G = \{\mc G_e,\mc G_o\}$ where
\begin{equation}
    \mc G_e = \mr{CZ}_{1,2}\otimes\mr{CZ}_{3,4}\otimes\cdots\otimes\mr{CZ}_{n-1,n},\quad \mc G_o = \mr{CZ}_{2,3}\otimes\mr{CZ}_{4,5}\otimes\cdots\otimes\mr{CZ}_{n,1},
\end{equation}
and we assume the state preparation, measurement, and gate noise for both $\mc G_e,\mc G_o$ are all $\Omega$-local, with $\Omega$ associated with $\Omega_* = \left\{\{i,i+1\}:i\in[n]\right\}$. The setup is illustrated in Fig.~\ref{fig:NN_CZ}. This setup is essentially the one used in \cite{kim2023evidence} for conducting an error-mitigated quantum simulation task. It is thus of great interest to understand the learnability of this noise model. 

\begin{figure}[!htp]
    \centering
    \includegraphics[width=0.85\linewidth]{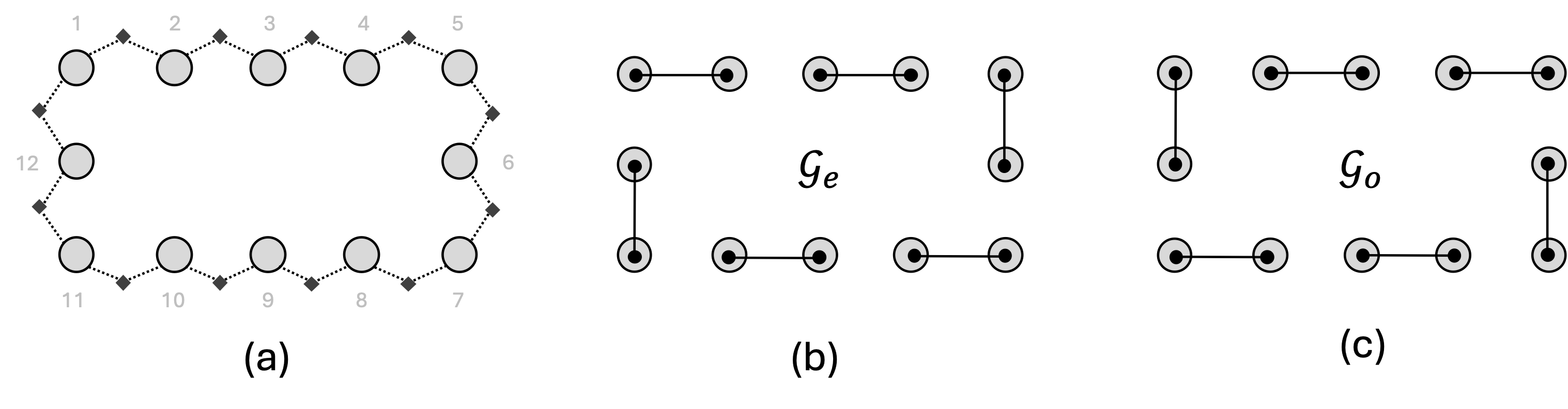}
    \caption{Illustration of the setup in Sec.~\ref{sec:NN_CZ} with $n=12$. (a) Qubit topology and noise model. Each circle represents a qubit with their label annotated. Each diamond connects to a maximal factor from $\Omega_*$, specifying the nearest-neighbor noise model for both SPAM and gate noises. (b), (c) Illustration of $\mc G_e,\mc G_o$, respectively.}
    \label{fig:NN_CZ}
\end{figure}

To begin with, this noise model is not gate-covariant in the sense of Definition~\ref{de:covariant}. 
To see this, consider the Pauli channel $\Lambda'(\rho) = \eta X_2X_3\rho X_2X_3 + (1-\eta)\rho$ with some $\eta<1/2$, which is $\Omega$-local thanks to Lemma~\ref{le:factorize}. We have $\mc G_e\Lambda'\mc G_e^\dagger = \eta Z_1X_2X_3Z_4\rho Z_1X_2X_3Z_4 + (1-\eta)\rho$, which is not $\Omega$-local, again by Lemma~\ref{le:factorize}.
Similarly for $\mc G_o$. This means we cannot directly apply Theorem~\ref{th:quasi_learnability}(c). In fact, we can show that the reduced gauge space is given by $T_R=\mr{span}\{\mc Q^{-1}(\mf d_{\{j\}}):j\in[n]\}$, i.e., only the single-qubit depolarizing gauges. Intuitively, any (linear combination of) two-qubit depolarizing gauge transformation will violate the $\Omega$-local noise assumption of either $\mc G_e$ or $\mc G_o$. A formal proof is given in Appendix~\ref{sec:proof_NN_CZ}. 
Therefore, we obtain the following characterization of the model learnability,
\begin{itemize}
    \item $X_R$: $\dim(X_R)=28n$. Basis: $\{\bm\vartheta_\nu^S,\bm\vartheta_\nu^M,\bm\vartheta_a^{\mc G_e},\bm\vartheta_a^{\mc G_o}:\nu\in\Omega,a\sim\Omega\}$. 
    \item $T_R$: $\dim(T_R)=n$. Basis: $\{\mc Q^{-1}(\mf d_{\{i\}}):i\in[n]\}$. 
    \item $L_R$: $\dim(L_R)=27n$.     
\end{itemize}
Note that the cardinality of $\{a\in{\sf P}^n\backslash I_n: a\sim\Omega\}$ can be computed using the inclusion-exclusion principle: $(4^2-1)n - (4-1)n = 12n$.

\medskip

We have also found a reduced cycle basis for $L_R^{G}$. To make the connection to experiments more straightforward, we will write down a cycle basis that generates the reduced cycle basis by applying $\mc Q^{\T}$ on them (c.f. Definition~\ref{de:reduced_cycle}).
The cycle bases are divided into two parts, depending on whether the basis vector involves one or two gates, as follows,
\begin{itemize}
    \item 1-gate:
    \begin{itemize}
        \item For $k=1,...,\frac n2$, similar as Eq.~\eqref{eq:1CZ_gate},
        \begin{equation}\label{eq:NN_CZ_basis1}
        \begin{gathered}
        \{\bm e^{\mc G_e}_{I_{2k-1}Z_{2k}},\cdots,\bm e^{\mc G_e}_{X_{2k-1}I_{2k}} + \bm e^{\mc G_e}_{X_{2k-1}Z_{2k}},\cdots\}\\
       \cup\{\bm e^{\mc G_o}_{I_{2k}Z_{2k+1}},\cdots,\bm e^{\mc G_o}_{X_{2k}I_{2k+1}} + \bm e^{\mc G_o}_{X_{2k}Z_{2k+1}},\cdots\}
        \end{gathered}
        \end{equation}
        \item For $k=1,...,\frac n2$, 
        \begin{equation}\label{eq:NN_CZ_basis2}
        \begin{gathered}
            \{\bm e_{a}^{\mc G_e} + \bm e_{a'}^{\mc G_e}:a\in{\sf P}^{\{2k,2k+1\}},w(a)=2,a'=\mc G_e(a)\}\\
            \cup \{\bm e_{a}^{\mc G_o} + \bm e_{a'}^{\mc G_o}:a\in{\sf P}^{\{2k+1,2k+2\}},w(a)=2,a'=\mc G_o(a)\}
        \end{gathered}
        \end{equation}
    \end{itemize}
    \item 2-gate: For $k=1,\cdots,\frac n2$,
    \begin{equation}\label{eq:NN_CZ_basis3}
    \begin{gathered}
        \left\{\bm e_{IXIXI_{\{2k-1,\cdots,2k+3\}}}^{\mc G_o} + 
        \bm e_{IXZXZ_{\{2k-1,\cdots,2k+3\}}}^{\mc G_e} + 
        \bm e_{ZXIXZ_{\{2k-1,\cdots,2k+3\}}}^{\mc G_o} + 
        \bm e_{ZXZXI_{\{2k-1,\cdots,2k+3\}}}^{\mc G_e}\right\} \\
        \cup\left\{\bm e_{IXIXI_{\{2k,\cdots,2k+4\}}}^{\mc G_o} + 
        \bm e_{ZXZXI_{\{2k,\cdots,2k+4\}}}^{\mc G_e} + 
        \bm e_{ZXIXZ_{\{2k,\cdots,2k+4\}}}^{\mc G_o} + 
        \bm e_{IXZXZ_{\{2k,\cdots,2k+4\}}}^{\mc G_e}\right\}
    \end{gathered}
    \end{equation}
\end{itemize}
Eqs.~\eqref{eq:NN_CZ_basis1},~\eqref{eq:NN_CZ_basis2},~\eqref{eq:NN_CZ_basis3} contributes ${13}n$, $9n$, $n$ cycles, respectively, and is consistent with the fact that $\dim(L_R^G) = 23n$. Elusive as they might look, each type of cycles corresponds to an explicit family of experiment designs, as shown in Fig.~\ref{fig:all_exp}~(a), (c), (d), respectively.

\medskip

For a proof that the above cycles constitute a cycle basis for $L_R^G$ and a way to supplement it into a complete basis for $L_R$, see Appendix~\ref{sec:basis_NN_CZ}.

\subsection{CZ's on ring, covariant quasi-local noise}\label{sec:cz_covar}

\begin{figure}[!htp]
    \centering
    \includegraphics[width=0.65\linewidth]{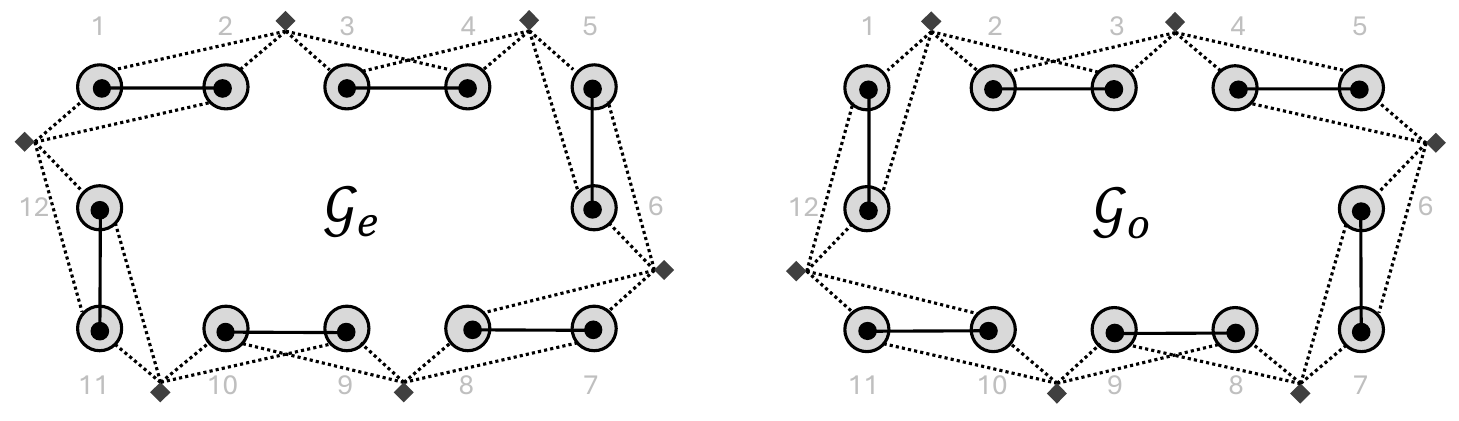}
    \caption{The covariant quasi-local gate noise model from Sec.~\ref{sec:cz_covar}. The left and right shows $\Omega^e$, $\Omega^o$, respectively, with each diamond connecting to a maximal factor from $\Omega^e_*$ (resp. $\Omega^o_*$).}
    \label{fig:covar_CZ}
\end{figure}

Finally, we present a minimal example of quasi-local noise model that is gate-covariant. Let $n\ge6$ be even. Consider the same gate set $\mf G=\{\mc G_e,\mc G_o\}$ as in the last section. Assume the SPAM noise is still $\Omega$-local for $\Omega_* = \{\{i,i+1\}:i\in[n]\}$. Assume the noise for $\mc G_e$, $\mc G_o$ is $\Omega^{\mc G_e}$, $\Omega^{\mc G_o}$-local, respectively, where
\begin{equation}
\begin{aligned}
    \Omega_*^{\mc G_e}&=\{\{2k-1,2k,2k+1,2k+2\}:k\in[n/2]\},\\
    \Omega_*^{\mc G_o}&=\{\{2k,2k+1,2k+2,2k+3\}:k\in[n/2]\}.
\end{aligned}
\end{equation}
The factor graphs for each noise model are given in Fig.~\ref{fig:covar_CZ}. The fact that this model is gate-covariant follows from Lemma~\ref{le:covariant_sufficient}. Indeed, each maximal factor is weight-$4$, containing exactly the support of two neighboring CZ. Consequently, Theorem~\ref{th:quasi_learnability} gives us that
\begin{itemize}
    \item $X_R$: $\dim(X_R)=244n$. Basis: $\{\bm\vartheta_\nu^S,\bm\vartheta_\nu^M,\bm\vartheta_a^{\mc G_e},\bm\vartheta_b^{\mc G_o}:\nu\in\Omega,a\sim\Omega^{\mc G_e},b\sim\Omega^{\mc G_o}\}$. 
    \item $T_R$: $\dim(T_R)=2n$. Basis: $\{\mc Q^{-1}(\mf d_{\nu}):\nu\in\Omega\}$. 
    \item $L_R$: $\dim(L_R)=242n$. Corresponding cycle basis:   $\{\bm e_\nu^S+\bm e_\nu^M:\nu\in\Omega\}\cup
    \{\bm e^S_{\pt(a)} +\bm e^{\mc G}_a +\bm e^M_{\pt(\mc G(a))}:\mc G\in\{\mc G_e,\mc G_o\},a\sim\Omega^{\mc G}\}$.
\end{itemize}
Note that the cardinality of $\{a\in{\sf P}^n\backslash I_n:a\sim\Omega^{\mc G_e}\}$ can again be computed using the inclusion-exclusion principle: $\frac n2\times(4^4-4^2)=120n$. One can see that a linear number of experiments that involves measuring at most weight-$4$ Pauli observables is sufficient to learn the whole noise model. 
It is also not hard to see that, by paralleling many of the experiments, a constant number of experiments would also suffice.

\medskip

We will not pursue a reduced cycle basis for $L_R^G$ in this example due to its complexity, and leave this for future investigation.

\bigskip
\noindent\textbf{Notes added.} We became aware of an independent and concurrent work~\cite{calzona2024multi} that also studied the learnability of the sparse Pauli-Lindblad noise model and proposed novel learning protocols.
We thank Alessio Calzona for communicating with us about their manuscript.

\section*{Acknowledgments} 
We would like to thank Alireza Seif, Yunchao Liu, Edward Chen, Joseph Emerson, Matthew Graydon, Hsin-Yuan Huang, Ming Yuan, Sisi Zhou, and Yunxiao Zhao for helpful discussions. 
S.C. \& L.J. acknowledge support from the ARO(W911NF-23-1-0077), ARO MURI (W911NF-21-1-0325), AFOSR MURI (FA9550-21-1-0209, FA9550-23-1-0338), DARPA (HR0011-24-9-0359, HR0011-24-9-0361), NSF (ERC-1941583, OMA-2137642, OSI-2326767, CCF-2312755, OSI-2426975), and Packard Foundation (2020-71479).

\newpage
\appendix
\renewcommand{\thelemma}{\Alph{section}.\arabic{lemma}}

\section{Proof of Theorem~\ref{th:linearOb}}\label{sec:proof_linearOb}

The proof follows~\cite[App. D]{chen2023learnability} with some modification to fit in our framework.

\begin{proof}[Proof of Theorem~\ref{th:linearOb}(a)]
    For any experiment $F$, we need to find a finite set of rooted cycles (\textit{i.e.}, linear functions corresponding to a rooted cycle) $\{f_j\}_{j=1}^M$ and a function $\widehat F$ such that
    \begin{equation}
        F(\bm x) = \widehat F(f_1(\bm x),\cdots,f_M(\bm x)),\quad\forall\bm x\in X.
    \end{equation}
    By definition, $F(\bm x)$ is a $2^n$-dimensional real vector such that $F_k(\bm x) = \Tr[\widetilde{E}_k \widetilde{\mc C}(\widetilde\rho_0) ]$ where the noisy initial state, circuit, and POVM element are determined by the noise parameter $\bm x$. Expanding the SPAM noise parameters yields
    \begin{equation}
        F_k(\bm x) = \frac{1}{4^n}\sum_{u,u'\in{\{0,1\}}^n}(-1)^{k\cdot u'}\lambda_{u'}^M\lambda_u^{S}\Tr\left[Z^{u'}\widetilde{\mc C}(Z^{u})\right].
    \end{equation}
    On the other hand, any noisy circuit $\mc C$ satisfying our assumptions can be written as
    \begin{equation}
        \widetilde{\mc C} = \mc U\od m \circ \widetilde{\mc G}_{{m}} \circ \cdots \circ \mc U\od 1 \circ \widetilde{\mc G}_{{1}} \circ \mc U\od 0,
    \end{equation}
    where each $\mc U\od{i}=\bigotimes_{l=1}^n \mc U_l\od{i}$ is a layer of parallel single-qubit gate, and each $\widetilde{\mc G}_i$ is the noisy realization of a Clifford gate $\mc G_i$ from our gate set $\mf G$. Note that consecutive layers of single-qubit gates can always be merged into a single layer, and we have assumed single-qubit gate to be noiseless, so the above form of $\widetilde{\mc C}$ is indeed general.
    
    An important property for single-qubit gate layers is that, even if they are non-Clifford, they always preserve the pattern of an input Pauli operator. That is,
    \begin{equation}
        \mc U\od{i}(P_a) = \sum_{b:\pt(b)=\pt(a)} c_{b,a}\od i P_b,\quad\forall P_a\in{\sf P}^n,
    \end{equation}
    for some real coefficients $c_{b,a}\od i$. This comes from the fact that any single-qubit unitary channel is trace-preserving and unital. Thanks to this property, the action of $\widetilde{\mc C}$ on any input Pauli operator can be expressed as
    \begin{equation}
        \begin{aligned}
        \widetilde{\mc C}(P_a) &= ({\mc U}\od m \circ \widetilde{\mc G}_{{m}} \circ \cdots \circ {\mc U}\od 1 \circ \widetilde{\mc G}_{{1}} \circ {\mc U}\od 0)(P_a)\\
        &= ({\mc U}\od m \circ \widetilde{\mc G}_{{m}} \circ \cdots \circ {\mc U}\od 1 \circ \widetilde{\mc G}_{{1}}) \left(\sum_{b_0:\pt(b_0)=\pt(a)} c_{b_0,a}\od{0} P_{b_0} \right)\\
        &= ({\mc U}\od m \circ \widetilde{\mc G}_{{m}} \circ \cdots \circ {\mc U}\od 1) \left(\sum_{b_0:\pt(b_0)=\pt(a)} c_{b_0,a}\od 0\lambda_{b_0}^{\mc G_{1}} P_{\mc G_{1}(b_0)} \right)\\
        &= ({\mc U}\od m \circ \widetilde{\mc G}_{{m}} \circ \cdots \circ {\mc U}\od 2) \left(\sum_{\substack{
        b_0:\pt(b_0)=\pt(a),\\
        b_1:\pt(b_1)=\pt(\mc G_{1}(b_0))
        }} c_{b_1,\mc G_{1}(b_0)}\od{1}c_{b_0,a}\od{0} \lambda_{b_1}^{\mc G_{2}}\lambda_{b_0}^{\mc G_{1}} P_{\mc G_{2}(b_1)} \right)\\
        &= \cdots\\
        &= \sum_{\substack{
        b_0:\pt(b_0)=\pt(a),\\
        b_1:\pt(b_1)=\pt(\mc G_{1}(b_0)),\\
        \dots\\
        b_m:\pt(b_m)=\pt(\mc G_{m}(b_{m-1}))
        }}c_{b_m,\mc G_{m}(b_{m-1})}\od{m}\cdots c_{b_1,\mc G_{1}(b_0)}\od{1} c_{b_0,a}\od{0} \lambda_{b_{m-1}}^{\mc G_{m}}\cdots\lambda_{b_1}^{\mc G_{2}}\lambda_{b_0}^{\mc G_{1}} P_{b_m}.
        \end{aligned}
    \end{equation}
    Now, substitute this back into the expression of $F_k(\bm x)$,
    \begin{equation}
        \begin{aligned}
            &F_k(\bm x)\\ 
            =~&\sum_{\substack{u,u'\in\{0,1\}^n,\\
        b_0:\pt(b_0)=\pt(u),\\
        b_1:\pt(b_1)=\pt(\mc G_{1}(b_0)),\\
        \dots\\
        b_m:\pt(b_m)=\pt(\mc G_{m}(b_{m-1}))
        }}
        \frac{(-1)^{k\cdot u'}}{4^n}~c_{b_m,\mc G_{m}(b_{m-1})}\od{m}\cdots c_{b_1,\mc G_{1}(b_0)}\od{1} c_{b_0,a}\od{0} \lambda_{u'}^M\lambda_{b_{m-1}}^{\mc G_{m}}\cdots\lambda_{b_1}^{\mc G_{2}}\lambda_{b_0}^{\mc G_{1}}\lambda_u^{S} \Tr\left[Z^{u'}P_{b_m}\right]\\
        =~&\sum_{\substack{u\in\{0,1\}^n,\\
        b_0:\pt(b_0)=\pt(u),\\
        b_1:\pt(b_1)=\pt(\mc G_{1}(b_0)),\\
        \dots\\
        b_m:\pt(b_m)=\pt(\mc G_{m}(b_{m-1}))
        }}
        \frac{(-1)^{k\cdot \pt(b_m)}}{2^n}~c_{b_m,\mc G_{m}(b_{m-1})}\od{m}\cdots c_{b_0,a}\od{0} \lambda_{\pt(b_m)}^M\lambda_{b_{m-1}}^{\mc G_{m}}\cdots\lambda_{b_0}^{\mc G_{1}}\lambda_u^{S} ~\mathds 1[Z^{\pt(b_m)}=P_{b_m}].
        \end{aligned}
    \end{equation}
    Let us look at the indexes of the summation. For $u=\bm 0$, we must have $b_m=\cdots=b_0=I$, and the corresponding term in the sum becomes ${1}/{2^n}$ that is a constant; For $u\neq\bm 0$, we must have $b_i\neq I$ for all $i$, and the corresponding term in the sum depends on the noise parameters via
    \begin{equation}
    \begin{aligned}
        \lambda_{\pt(b_m)}^M\lambda_{b_{m-1}}^{\mc G_{m}}\cdots\lambda_{b_1}^{\mc G_{2}}\lambda_{b_0}^{\mc G_{1}}\lambda_u^{S} &= \exp\left(-\left(x_{\pt(b_m)}^M+x_{b_{m-1}}^{\mc G_{m}}+\cdots+x_{b_0}^{\mc G_{1}}+x_u^{S}\right)\right)\eqcol\exp\left(-f_{\bm b,u}(\bm x)\right),
    \end{aligned}
    \end{equation}
    where $f_{\bm b, u}$ is a linear function of $\bm x$. More importantly, thanks to the restriction of $(b_0,\cdots,b_m)$, one can see that $f_{\bm b, u}$ is actually a rooted cycle in the pattern transfer graph. Therefore, $F_k(\bm x)$ (and thus $F(\bm x)$) is indeed determined by $\{f_{\bm b,u}(\bm x)\}_{\bm b,u}$ consisting of finitely many rooted cycles. This completes the proof of Theorem~\ref{th:linearOb}(a).
\end{proof}

The following proof again comes from \cite{chen2023learnability} with minor modification, and is presented here for completeness.

\begin{proof}[Proof of Theorem~\ref{th:linearOb}(b)]
    For any rooted cycles $f$, we need to find an experiment $F$ and a function $\widehat f$ such that
    \begin{equation}
        f(\bm x) = \widehat f(F(\bm x)),\quad\forall \bm x\in X.
    \end{equation}
    A generic rooted cycle $f$ of depth $m$ takes the following form (here, we write down the vertices for clarity, though they are not needed to define a cycle),
    \begin{equation}\label{eq:general_rooted_cycle}
        \begin{gathered}
            (v_R,x^S_{b_0},v_{b_0},x_{P_1}^{\mc G_1},v_{b_1},x_{P_2}^{\mc G_2},\cdots,x_{P_m}^{\mc G_m},v_{b_m},x_{b_m}^M,v_R),\\
            \text{s.t.}\quad b_{i-1}=\pt(P_i),~\pt(\mc G_i(P_i))=b_{i},~\forall i=1,\cdots, m.
        \end{gathered}
    \end{equation}
    Consider the following noisy circuit
    \begin{equation}\label{eq:def_of_tildeC}
    \begin{gathered}
        \widetilde{\mc C} = \mc U\od m \circ \widetilde{\mc G}_{{m}} \circ \cdots \circ \mc U\od 1 \circ \widetilde{\mc G}_{{1}} \circ \mc U\od 0,\\
        \text{where}~ \mc U\od i\equiv \bigotimes_{j=1}^n\mc U_j\od i~\text{is Clifford, and satisfies}\left\{\begin{aligned}
            &\mc U\od 0(Z^{b_0})=P_1;\\
            &\mc U\od i(\mc G_i(P_i))=P_{i+1},~\forall i=1,\cdots,m-1;    \\
            &\mc U\od m(\mc G_m(P_m))=Z^{b_m}.
        \end{aligned}
        \right.
    \end{gathered}
    \end{equation}
    Note that, thanks to the constraints from Eq.~\eqref{eq:general_rooted_cycle}, the above requirements for $\mc U\od i$ are literally mapping one Pauli operator to another \emph{with the same pattern}, thus can indeed be achieved by a layer of single-qubit Clifford gates.
    
    This circuit yields the following experiment (\textit{i.e.}, experimental outcome distribution)
    \begin{equation}
    \begin{aligned}
        F_k(\bm x) = \Tr[\widetilde E_k\widetilde{\mc C}(\widetilde\rho_0)] = \frac{1}{4^n}\sum_{u,u'\in\{0,1\}^n}(-1)^{k\cdot u'}\lambda^M_{u'}\lambda^S_{u}\Tr\left[Z^{u'}\widetilde{\mc C}(Z^u)\right],\quad\forall k\in\{0,1\}^n.
    \end{aligned}
    \end{equation}
    Compute the following function of experimental outcome (which is, intuitively, the expectation value of $Z^{b_m}$ on the equivalent noisy state right before a perfect measurement)
    \begin{equation}
        \begin{aligned}
            \sum_{k\in\{0,1\}^n}(-1)^{k\cdot b_m}F_k(\bm x) &= \frac{1}{4^n}\sum_{k,u,u'}(-1)^{k\cdot(b_m+u')}\lambda^M_{u'}\lambda^S_{u}\Tr\left[Z^{u'}\widetilde{\mc C}(Z^u)\right]\\
            &= \frac{1}{2^n}\sum_u \lambda^M_{b_m}\lambda^S_{u}\Tr\left[Z^{b_m}\widetilde{\mc C}(Z^u)\right]\\
            &= \frac{1}{2^n}\sum_u \lambda^M_{b_m}\lambda^{\mc G_m}_{P_m}\cdots\lambda^{\mc G_1}_{P_1}\lambda^S_{u}\Tr\left[Z^{b_0}Z^u\right]\\
            &= \lambda_{b_m}^M\lambda^{\mc G_m}_{P_m}\cdots\lambda^{\mc G_1}_{P_1}\lambda^S_{b_0}\\
            &= \exp\left(-\left(x_{b_0}^S + x_{P_1}^{\mc G_1}+\cdots+x_{P_m}^{\mc G_m} + x^M_{b_m}\right)\right) \equiv \exp\left(-f(\bm x)\right).
        \end{aligned}
    \end{equation}
    The third line is obtained by acting $\widetilde{\mc C}^\dagger$ on $Z^{b_m}$ inside the trace, and then evolve $Z^{b_m}$ backwards according to the definition of $\widetilde{\mc C}$ from Eq.~\eqref{eq:def_of_tildeC}. As a result, we have
    \begin{equation}
        f(\bm x) = -\log\left(\sum_{k\in\{0,1\}^n}(-1)^{k\cdot b_m}F_k(\bm x)\right).
    \end{equation}
    This completes the proof of Theorem~\ref{th:linearOb}(b). Recall that this means any rooted cycle $f$ can be learned from a single experiment $F$.
\end{proof}

\section{Proof of Theorem~\ref{th:reduced}}\label{sec:proof_reduced}
\begin{proof}
    We first show that $T_R=\mr{Ker}(\mc P\circ\mc Q)$.
    For any experiment $F$, Theorem~\ref{th:linearOb}(b) implies that there exists some function $\widehat F$ such that $F(\mc Q(\bm r)) = \widehat F(\mc P\circ\mc Q(\bm r))$ for all $\bm r\in X_R$. Thus, for any $\bm t\in\mr{Ker}(\mc P\circ\mc Q)$,
    \begin{equation}
        F(\mc Q(\bm r +\eta\bm t)) = \widehat F(\mc P\circ\mc Q(\bm r+\eta\bm t)) = \widehat F(\mc P\circ\mc Q(\bm r)) =F(\mc Q(\bm r)),\quad\forall\bm r\in X_R,~\forall \eta\in\mbb R.
    \end{equation}
    The second equality uses linearity of $\mc P\circ\mc Q$. The above means $\bm t\in T_R$ by definition. Therefore, $\mr{Ker}(\mc P\circ\mc Q)\subseteq T_R$; On the other hand, for any $\bm t\notin\mr{Ker}(\mc P\circ\mc Q)$, $\mc Q(\bm t)\notin\mr{Ker}(\mc P)$, which means there is at least one $f_i$ such that $f_i(\mc Q(\bm t))\neq 0$. Now, Theorem~\ref{th:linearOb}(b) says that there exists an experiment $F_i$ and a function $\widehat{f}_i$ such that $f_i(\mc Q(\bm t))=\widehat f_i(F_i(\mc Q(\bm t)))$ for all $\bm r$. Thus
    \begin{equation}
        \begin{aligned}
            f_i(\mc Q(\bm t))\neq 0&\Rightarrow f_i(\mc Q(\bm r))\neq f_i(\mc Q(\bm r+\bm t)),\quad\forall\bm r\in X_R,\\
            &\Rightarrow \widehat f_i(F(\mc Q(\bm r)))\neq \widehat f_i(F(\mc Q(\bm r+\bm t))),\quad\forall\bm r\in X_R,\\
            &\Rightarrow F(\mc Q(\bm r))\neq F(\mc Q(\bm r+\bm t)),\quad\forall\bm r\in X_R,
        \end{aligned}
    \end{equation}
    where the first line uses linearity of $f_i$ and $\mc Q$.
    This means $\bm t\notin T_R$. Therefore, $T_R\subseteq \mr{Ker}(\mc P\circ\mc Q)$. This proves that $T_R = \mr{Ker}(\mc P\circ\mc Q)$. 

    \medskip

    We then show that $L_R = \mr{Im}(\mc Q^{\T}\circ\mc P^{\T})$.
    For any $f\in\mr{Im}(\mc Q^{\T}\circ\mc P^\T)$, there exists some linear function $h:\mbb R^{|Z|}\mapsto \mbb R$ such that $f(\bm r) = \mc Q^\T\circ\mc P^\T(h)(\bm r)\equiv h(\mc P\circ\mc Q(\bm r)),~\forall\bm r\in X_R$. Theorem~\ref{th:linearOb} implies that each entry of $\mc P(\mc Q(\bm r))$ can be learned from some experiments, thus $f\in L_R$. This means $\mr{Im}(\mc Q^{\T}\circ\mc P^\T)\subseteq L_R$; On the other hand, $L_R\perp T_R$, according to similar arguments as in Lemma~\ref{le:L_T_perp}. We have shown $T_R = \mr{Ker}(\mc P\circ\mc Q)$, thus $L_R\subseteq\mr{Im}(\mc Q^{\T}\circ\mc P^\T)$ thanks to the orthogonality between kernel and adjoint image. This proves that $L_R=\mr{Im}(\mc Q^{\T}\circ\mc P^\T)$. 
\end{proof}

\section{Additional proofs for Sec.~\ref{sec:fully}}
\subsection{Proof of Lemma~\ref{le:sdg_basis}}\label{sec:proof_sdg}
\begin{proof}
Since $\#\{\mf d_s\}_s = \mr{dim}(T) = 2^n-1$ (recall that the PTG is strongly connected thanks to the root vertex), we only need to show the completeness of $\{\mf d_s\}_s$. It is sufficient to show that $\{\mf d_s\}_s$ can express the canonical cut basis $\mf y_z$ generated by cutting off a single vertex ${{\mathsf{V}_0}}=\{v_z\}$ for all non-zero bit string $z$. (Note that the root node gives a redundant degree of freedom and can thus be ignored). By definition, we have
\begin{equation}\label{eq:dep_in_can}
    \mf d_s = -\sum_{z} \mathds 1[z_s\neq 0]\mf y_z.
\end{equation}
For two $n$-bit string $s,t$ we write $t\preceq s$ if $t_i=0$ for any index $i$ such that $s_i = 0$. Let $\bar s$ denote the bitwise flip of $s$ (e.g., $\overline{1101}=0010$). We claim the following relation holds,
\begin{equation}\label{eq:can_in_dep}
    \mf y_z = \sum_{s\neq\bm 0:\bar{z}\preceq s}
    (-1)^{|s|-|\bar z|}\mf d_s,\quad\forall z\neq\bm0. 
\end{equation}
This is basically a consequence of the inclusion-exclusion principle.
To see this, thanks to Eq.~\eqref{eq:dep_in_can},
\begin{equation}
    \begin{aligned}
        R.H.S. &= -\sum_{s\neq\bm 0:\bar{z}\preceq s}(-1)^{|s|-|\bar z|}\sum_{t:t_s\neq 0}\mf y_t\\
        &=-\sum_{t}\sum_{s:\bar z\preceq s,~s\not\preceq \bar{t}}(-1)^{|s|-|\bar{z}|}\mf y_t\\
        &=\sum_t\sum_{s:\bar z\preceq s\preceq\bar t}(-1)^{|s|-|\bar z|}\mf y_t\\
        &= \sum_t \mathds 1[\bar z\preceq\bar t]\sum_{k=0}^{|\bar t|-|\bar z|}(-1)^{k}\binom{|\bar t|-|\bar z|}{k} \mf y_t\\
        &= \sum_{t} \mathds 1[\bar z\preceq\bar t]\mathds 1[|\bar t|-|\bar z|=0]\mf y_t\\
        &= \mf y_z.
    \end{aligned}
\end{equation}
The order change of sum in the second line uses the fact that $t_s\neq 0$ is equivalent to $s\not\preceq \bar t$, and that $\bm 0\preceq \bar t$ for any $t$, so we can drop the constraint $s\neq\bm 0$. 
The third line uses the following
\begin{equation}
    \sum_{s:\bar z\preceq s\preceq t}(-1)^{|s|-|\bar z|} + \sum_{s:\bar z\preceq s,~s\not\preceq t}(-1)^{|s|-|\bar z|} = 
    \sum_{s:\bar z\preceq s}(-1)^{|s|-|\bar z|} = \sum_{k=0}^{n-|\bar z|}\binom{n-|\bar z|}{k}(-1)^k=0.
\end{equation}
Note that $z\neq\bm 0$ implies $n-|\bar z|>0$. The fifth line uses the Binomial theorem.
This proves Eq.~\eqref{eq:can_in_dep}, and thus establishes the fact that the SDG's $\{\mf d_s\}_s$ form a basis for $T$ and $U$.
\end{proof}

\subsection{Proof of converse part of Theorem~\ref{th:fully_learnability}(b)}\label{sec:proof_fully_converse}

\begin{proof}
Let $(X_R,\mc Q)$ satisfies a fully local noise model for gate $\mc G$. Define
\begin{equation}
\begin{aligned}
    T_0 &\coleq \{\bm z\in T:\bm z^{\mc G}\in\imq^{\mc G}\},\\
    \quad D_0 &\coleq \mr{span}\{\mf d_s:s\subset\mr{supp}(\mc G)~\textrm{or}~s\supseteq\mr{supp}(\mc G)~\textrm{or}~s\cap\mr{supp}(\mc G)=\varnothing\}.
\end{aligned}
\end{equation}
It has been proved that $T_0\supseteq D_0$.    
Now, assume that $\hat{\mc G}$ induces a connected pattern transfer subgraph.
Let $k=|\mc G|$ and assume $\mc G$ supports on the first $k$ qubits, without loss of generality. 
By definition, any $\bm z^{\mc G}\in \imq^{\mc G}$ satisfies
    \begin{equation}\label{eq:fully_ortho}
        z^{\mc G}_{b\oplus c_1} = z^{\mc G}_{b\oplus c_2},\quad\forall b\in{\sf P}^{k}\backslash\{I_k\},~\forall c_1,c_2\in{\sf P}^{n-k}.
    \end{equation}
    Here $\oplus$ means concatenation of bit string, i.e., $P_{b\oplus c}=P_b\otimes P_c$.
    Equivalently, $\bm z\cdot (\bm e^{\mc G}_{b\oplus c_1} - \bm e^{\mc G}_{b\oplus c_2})=0$.
    Now, let us pick a set of $k$-qubit Paulis $\{b_j\in{\sf P}^k\backslash\{I\}\}_{j=1}^{2^{k}-2}$ such that 
    the set of edges $\{\bm e^{\mc G}_{b_j\oplus \bm 0}\}_{j=1}^{2^k-2}$, viewed as undirected, form a spanning tree for the subgraph of vertices $\{{t\oplus 0_{n-k}}\in {\mathsf V}:t\in\{0,1\}^k,t\neq 0_{k}\}$.
    Such a choice always exists, as we have assumed that the pattern transfer subgraph of $\widehat{\mc G}$ is strongly connected.
    Next, consider the following vectors
    \begin{equation}
        \bm k_{j,c}\coleq \bm e^{\mc G}_{b_j\oplus I} - \bm e^{\mc G}_{b_j\oplus c},\quad\forall j=1,\cdots, 2^{k}-2,~\forall c\in{\sf P}^{n-k}\backslash\{I\}.
    \end{equation}
    It is not hard to see $\{\bm k_{j,c}\}$ forms a linearly-independent set, as each $ \bm e^{\mc G}_{b_j\oplus c}~(c\neq I)$ appears in exactly one $\bm k_{j,c}$. Let $K\coleq \mr{span}\{\bm k_{j,c}\}$. We thus have 
    \begin{equation}
        \dim(K)=(2^k-2)\times(2^{n-k}-1) = 2^n-2^{k}-2^{n-k+1}+2
    \end{equation}
    Since $K$ supports on a union of spanning trees, the intersection between $K$ and the cycle space $Z$ must be trivial, i.e., $K\cap Z=\{\bm 0\}$ (see e.g.,~\cite{berge2001theory}). 
    Meanwhile, one has $T_0\perp K$ by Eq.~\eqref{eq:fully_ortho} and $T_0\perp Z$ by the orthogonality between cycle space and cut space. Combining all of these, we have
    $X\supseteq T_0\oplus K\oplus Z$, and that
    \begin{equation}
    \begin{aligned}
        \dim(T_0)&\le \dim(X)-\dim(Z)-\dim(K)\\
        &= \dim(T) - \dim(K)\\
        &= 2^k + 2^{n-k+1} - 3.
    \end{aligned}
    \end{equation}
    The second line uses the fact that $T$ is the complement of $Z$ in $X$.
    On the other hand, by directly counting the number of SDG's in $D_0$, we have $\dim(D_0) = (2^k-2) + 2^{n-k} + (2^{n-k}-1) \ge \dim(T_0)$.
    Since we have shown $T_0\supseteq D_0$, we conclude $D_0=T_0$. This completes our proof.
\end{proof}

\section{Relative-precision learning for 2-qubit gates with fully-local noise}\label{sec:local_relative}

We consider the second type of algorithms that learn gate noise parameters to relative precision. The key is to construct a reduced cycle basis for the gate noise parameter space $L_R^{G} = \mc Q^\T(L\cap X^G)$.
Note that, one might think 
\begin{equation}
    \widetilde{L_R^{G}}\coleq L_R\cap X^G_R = \mc Q^\T(L)\cap\mc Q^\T(X^G)
\end{equation}
should be the natural definition of the reduced learnable gate parameter space. However, $L_R^{G}$ and $\widetilde{L_R^{G}}$ might not be equal in general. Theorem~\ref{th:relative} implies the former is the space of functions that can be amplified and learned to relative precision.
Nevertheless, in some cases $L_R^{G}$ and $\widetilde{L_R^{G}}$ can be shown to be equal.

We focus on the case that every $\mc G\in\mf G$ is supported non-trivially on at most $2$ qubits. Since many existing quantum computing platform uses two-qubit entangling gates such as CZ as a primitive, we expect our method to be practically useful.

\medskip

To begin with, we will first describe an algorithm for finding a basis for $\widetilde{L_R^G}$ and then show how to augment it to an algorithm for finding reduced cycle basis for $\widetilde{L_R^G}$.
Since reduced cycle basis belongs to $L_R^G$, this implies that $\widetilde{L_R^G}\subseteq L_R^G$.
Moreover, by definition $L_R^G\subseteq\widetilde{L_R^G}$, so $\widetilde{L_R^G}=L_R^G$ in this case, and the reduced cycle basis we find for $\widetilde{L_R^G}$ is also a reduced cycle basis for $L_R^G$.

\noindent\textbf{Algorithm \algnum} [Finding a basis for $\widetilde{L_R^G}$].%
\begin{enumerate}
\item Set $A=\{\bm\vartheta_a^\mc G:a\in\sf P^2\backslash\{I_2\},\mc G\in\mf G\}$ and $C=\varnothing$.
\item Find a non-zero vector $\bm c\in \widetilde{L_R^G}\cap\mr{span}(A)$.
\item Pick a $\bm\vartheta$ with $\bm\vartheta\cdot \bm c\neq0$, set $A=A\backslash\{\bm\vartheta\}$ and $C=C\cup\{\bm c\}$.
\item Repeat step 2 and 3 until no $\bm c$ can be found, return $C$.
\end{enumerate}

\begin{lemma}\label{le:alg_correct}
Algorithm 3 outputs a basis for $\widetilde{L_R^G}$.
\end{lemma}
\noindent This is a standard linear algebraic procedure. For completeness, a proof is given below.
\begin{proof}
We prove the following loop invariant: $(\widetilde{L_R^G}\cap\mr{span}(A))+\mr{span}(C)=\widetilde{L_R^G}$.

Initially, $\mr{span}(A)=X_R^G\supseteq \widetilde{L_R^G}$ and $\mr{span}(C)=\{\bm0\}$, so $(\widetilde{L_R^G}\cap\mr{span}(A))+\mr{span}(C)=\widetilde{L_R^G}$.
Now suppose the loop invariant is true before the update step 3.
Then $\forall \bm l\in \widetilde{L_R^G}$, $\exists \bm l_0\in \widetilde{L_R^G}\cap\mr{span}(A)$ and $\bm c_0\in\mr{span}(C)$ s.t. $\bm l=\bm l_0+\bm c_0$.
Then $\bm l_0-\frac{\bm\vartheta\cdot \bm l_0}{\bm\vartheta\cdot \bm c}\bm c\in \widetilde{L_R^G}\cap\mr{span}(A)$.
Moreover, $\bm\vartheta\cdot(\bm l_0-\frac{\bm\vartheta\cdot \bm l_0}{\bm\vartheta\cdot \bm c}\bm c)=0$, so $\bm l_0-\frac{\bm\vartheta\cdot \bm l_0}{\bm\vartheta\cdot \bm c}\bm c\in \widetilde{L_R^G}\cap\mr{span}(A\backslash{\bm\vartheta})$.
Since $\frac{\bm\vartheta\cdot \bm l_0}{\bm\vartheta\cdot \bm c}\bm c+\bm c_0\in\mr{span}(C\cup\{\bm c\})$, by $\bm l=(\bm l_0-\frac{\bm\vartheta\cdot \bm l_0}{\bm\vartheta\cdot \bm c}\bm c)+(\frac{\bm\vartheta\cdot \bm l_0}{\bm\vartheta\cdot \bm c}\bm c+\bm c_0)$ we have $(\widetilde{L_R^G}\cap\mr{span}(A\backslash\bm\vartheta))+\mr{span}(C\cup\{\bm c\})=\widetilde{L_R^G}$, proving the loop invariant.
In the end, $\widetilde{L_R^G}\cap\mr{span}(A)=\{\bm0\}$, so the span of the output $C$ is $\widetilde{L_R^G}$.

Furthermore, for any update step 3, the vector $c$ appended to $C$ has support on $\bm\vartheta$, whereas all the vectors appended afterward do not.
Hence $C$ is linearly independent and the output is a basis for $\widetilde{L_R^G}$.
\end{proof}
\medskip

Before we augment the algorithm, we first give a characterization of the space $\widetilde{L_R^G}$.

\begin{definition}
The boundary operator is a linear operator  $\partial:X_R^G\rightarrow X_R^S$ defined as
\[\partial(\bm\vartheta_a^\mc G)=\sum_{j\in\mr{supp}(\mc G)}\left(\mathds1[\pt(a)_j\neq\bm0]-\mathds1[\pt(\hat{\mc G}(a))_j\neq\bm0]\right)\bm\vartheta_j^S.\]
For $\bm r\in X_R^G$, $\partial(\bm r)$ is called its boundary.
\end{definition}

\begin{lemma}\label{le:LR_iff}
$\forall\bm r\in X_R^G$, $\bm r\in \widetilde{L_R^G}$ if and only if $\partial(\bm r)=\bm 0$.
\end{lemma}

\begin{proof}
Note that any $\bm r\in X_R^G$ can be expanded as $\bm r = \sum_{a\in\sf P^2\backslash\{I_2\},\mc G\in\mf G}r_a^\mc G\bm\vartheta_a^\mc G$. We have
\begin{equation}
\begin{aligned}
&\sum_{a\in\sf P^2\backslash\{I_2\},\mc G\in\mf G}r_a^\mc G\bm\vartheta_a^\mc G\in L_R\\
\Leftrightarrow&\sum_{a\in\sf P^2\backslash\{I_2\},\mc G\in\mf G}\sum_{j\in\mr{supp}(\mc G)}r_a^\mc G\left(\mathds1[\pt(a)_j\neq\bm0]\bm\vartheta_j^S+\mathds1[\pt(\hat{\mc G}(a))_j\neq\bm0]\bm\vartheta_j^M\right)\in L_R\\
\Leftrightarrow&\sum_{a\in\sf P^2\backslash\{I_2\},\mc G\in\mf G}\sum_{j\in\mr{supp}(\mc G)}r_a^\mc G\left(\mathds1[\pt(a)_j\neq\bm0]-\mathds1[\pt(\hat{\mc G}(a))_j\neq\bm0]\right)\bm\vartheta_j^S\in L_R\\
\Leftrightarrow&\sum_{a\in\sf P^2\backslash\{I_2\},\mc G\in\mf G}\sum_{j\in\mr{supp}(\mc G)}r_a^\mc G\left(\mathds1[\pt(a)_j\neq\bm0]-\mathds1[\pt(\hat{\mc G}(a))_j\neq\bm0]\right)\bm\vartheta_j^S=\bm0\\
\Leftrightarrow&\quad\partial(\bm r)=\bm 0.
\end{aligned}
\end{equation}
The second and third lines use the fact that ${\mathsf B}_R$ is a basis for $L_R$. 
In the fourth line we use the fact that any nontrivial functions of the reduced state preparation parameters changes under some single-qubit depolarization gauge, and is thus not learnable (More formally, $\forall \bm0\neq\bm f\in X_R^S$, $\exists |s|=1$, s.t., $\bm f\not\perp\mf d_s$).
\end{proof}

Now we augment the algorithm so that it finds a reduced cycle basis.
That is, we need to guarantee that every $\bm c$ appended is a reduced cycle.
To do this, we specify how to find $\bm c$ in step 2 and how to pick $\bm\vartheta$ in step 3.
The augmented algorithm consists of $2$ stages, for which we will describe how to make the choices separately.

In the first stage, we look for $\bm c=\sum_{a\in\sf P^2\backslash\{I_2\},\mc G\in\mf G}r_a^\mc G\bm\vartheta_a^\mc G$ with the property that all non-zero $r_a^\mc G$ share the same $\mr{supp}(\mc G)$.
That is, we group the base vectors $\bm\vartheta_a^\mc G$ by $\mr{supp}(\mc G)$ and consider them separately.
Fix a common support $s$, the task simplifies substantially since Lemma~\ref{le:LR_iff} simplifies into a requirement on $2$ dimensions (the common support, which is of size $2$).
Then we may perform the analysis on the pattern transfer subgraph (the union of the pattern transfer subgraphs of $\hat{\mc G}$ for all $\mc G\in\mf G$ with $\mr{supp}(\mc G)=s$).
In the following we consider the possible shapes of the pattern transfer subgraph.
Since self loops can be indentified and picked as reduced cycle basis easily at this stage, we do not need to explicitly consider them.

\begin{lemma}\label{le:shape}
For two qubit Clifford gates, there are $3$ types of non-trivial pattern transfer subgraphs, the CNOT-like, SWAP-like and iSWAP-like, as shown in Figure~\ref{fig:PTMs}.
\end{lemma}

\begin{figure}[!htp]
    \centering
    \includegraphics[width=0.9\linewidth]{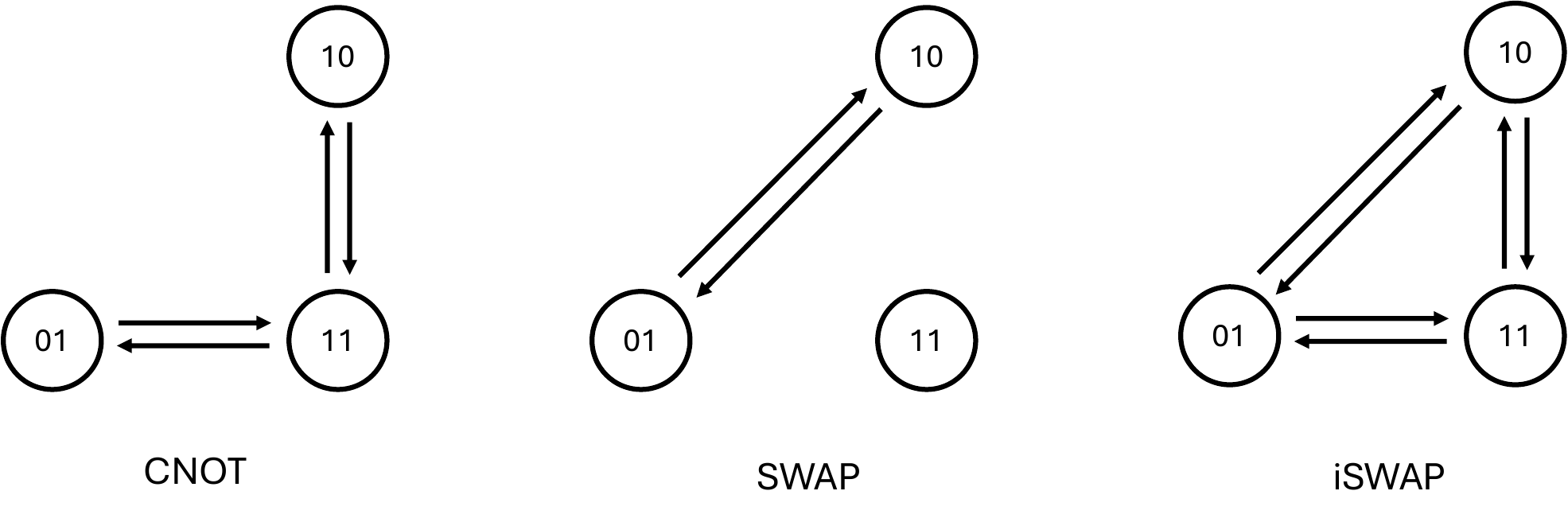}
    \caption{The possible pattern transfer subgraphs of non-trivial $2$-qubit Clifford gates. Self loops are omitted. Multi-edges are shown by single edges.}
    \label{fig:PTMs}
\end{figure}

\begin{proof}
By theorem 26 of~\cite{grier2022classification}, any $2$ qubits Clifford gates is equivalent to at most $3$ layers of gates (each layer may or may not exist): 1. a tensor of single qubit Clifford gates $G\otimes H$ 2. a SWAP gate 3. a generalized CNOT gate $\frac{I\otimes I+P\otimes I+I\otimes Q-P\otimes Q}{2}$.
Since single qubit Clifford gates before or after a $2$ qubit Clifford gate does not change the pattern transfer graph, in the pattern transfer graph sense, any $2$ qubits Clifford gates is equivalent to at most $2$ layers of gates: 1. a SWAP gate 2. a CNOT gate.
When both gates are absent, the pattern transfer graph is trivial (in this case the Clifford gate is a tensor product of single qubit Clifford gates, which by assumption shouldn't be considered as noisy).
If only CNOT is present, the pattern transfer graph is CNOT-like.
If only SWAP is present, the pattern transfer graph is SWAP-like.
If both gate are present, the pattern transfer graph is iSWAP-like.
\end{proof}

From Figure~\ref{fig:PTMs} one can easily see that the union of pattern transfer subgraphs still has only $3$ possible shapes.
The important message of Lemma~\ref{le:shape} is that for any edge in the pattern transfer subgraph, there is another edge of the reverse direction.

Now we resume the description of the first stage.
Obviously the cycles in the pattern transfer subgraphs all corresponds to reduced cycles (by considering Pauli operators with the same support).
One can easily find a cycle basis for the cycle space of the pattern transfer subgraph.
Additionally, the cycle basis can be chosen such that for each edge $\bm\vartheta_i$ in the pattern transfer subgraph, there exists a length $2$ reduced cycle that uses it.
That is, we have picked a reduced cycle $\bm\vartheta_i+\bm\vartheta_i'$ in the basis, where $\bm\vartheta_i'$ is an edge with the reverse direction of $\bm\vartheta_i$.
In the first stage of the algorithm, we just pick $\bm c$ to be such a cycle basis.
They are naturally reduced cycles.
This can always be done due to the vast freedom in the choice of $\bm\vartheta$.
For the pattern transfer subgraph, at each iteration we eliminate an edge from $A$.
In the end the remaining edges form a spanning tree.
Again due to the freedom in the choice of $\bm\vartheta$ we are free to choose which spannings tree should be kept in $A$ and direction of the edges of the remaining spanning tree.
For CNOT-like and SWAP-like pattern transfer subgraphs, there is only one possible shape of the spanning tree.
As for iSWAP-like pattern transfer graph, wlog we pick $(01,11)$ and $(10,11)$ as the spanning tree.
The edge orientation for the spanning trees could be arbitrary.

\begin{figure}[!htp]
\subfloat[There exists two type-1 boundaries with the same support.\label{subfig:2tp1}]{
\begin{minipage}{\textwidth}
\begin{center}
\begin{tabular}{ c c c c c c c c c }
$0$ & $=$ & $r_{a_1}^{\mc G_1}\partial(\bm\vartheta_{a_1}^{\mc G_1})$ & $+$ & $r_{a_2}^{\mc G_2}\partial(\bm\vartheta_{a_2}^{\mc G_2})$ & $+$ & $r_{a_3}^{\mc G_3}\partial(\bm\vartheta_{a_3}^{\mc G_3})$ & $+$ & $r_{a_4}^{\mc G_4}\partial(\bm\vartheta_{a_4}^{\mc G_4})$ \\
$\begin{pmatrix} 0 \\ 0 \\ 0 \end{pmatrix}$ & $=$ & $\begin{pmatrix} 0 \\ \textcolor{red}{2} \\ 0 \end{pmatrix}$ & $+$ & $\begin{pmatrix} 0 \\ \textcolor{red}{-1} \\ 0 \end{pmatrix}$ & $+$ & $\begin{pmatrix} 0 \\ -1 \\ 1 \end{pmatrix}$ & $+$ & $\begin{pmatrix} 0 \\ 0 \\ -1 \end{pmatrix}$ \\
& & & & \huge$\Downarrow$ & & & & \\
$\begin{pmatrix} 0 \\ 0 \\ 0 \end{pmatrix}$ & $=$ & $\begin{pmatrix} 0 \\ \textcolor{red}{1} \\ 0 \end{pmatrix}$ & $+$ & $\begin{pmatrix} 0 \\ \textcolor{red}{-1} \\ 0 \end{pmatrix}$ & $+$ & $\begin{pmatrix} 0 \\ 0 \\ 0 \end{pmatrix}$ & $+$ & $\begin{pmatrix} 0 \\ 0 \\ 0 \end{pmatrix}$ \\
\end{tabular}
\end{center}
\end{minipage}
}\\
\subfloat[Chasing procedure at both end terminates with a type-1 boundary.\label{subfig:tp2chain}]{
\begin{minipage}{\textwidth}
\begin{center}
\begin{NiceTabular}{ c c c c c c c c c }
$0$ & $=$ & $r_{a_{-1}}^{\mc G_{-1}}\partial(\bm\vartheta_{a_{-1}}^{\mc G_{-1}})$ & $+$ & $r_{a_0}^{\mc G_0}\partial(\bm\vartheta_{a_0}^{\mc G_0})$ & $+$ & $r_{a_1}^{\mc G_1}\partial(\bm\vartheta_{a_1}^{\mc G_1})$ & $+$ & $r_{a_2}^{\mc G_2}\partial(\bm\vartheta_{a_2}^{\mc G_2})$ \\
$\begin{pmatrix} 0 \\ 0 \\ 0 \end{pmatrix}$ & $=$ & $\begin{pmatrix} \textcolor{red}{1} \\ 0 \\ 0 \end{pmatrix}$ & $+$ & $\begin{pmatrix} \textcolor{red}{-1} \\ \textcolor{red}{1} \\ 0 \end{pmatrix}$ & $+$ & $\begin{pmatrix} 0 \\ \textcolor{red}{-1} \\ \textcolor{red}{1} \end{pmatrix}$ & $+$ & $\begin{pmatrix} 0 \\ 0 \\ \textcolor{red}{-1} \end{pmatrix}$ \\
& & & & \huge$\Downarrow$ & & & & \\
$\begin{pmatrix} 0 \\ 0 \\ 0 \end{pmatrix}$ & $=$ & $\begin{pmatrix} \textcolor{red}{1} \\ 0 \\ 0 \end{pmatrix}$ & $+$ & $\begin{pmatrix} \textcolor{red}{-1} \\ \textcolor{red}{1} \\ 0 \end{pmatrix}$ & $+$ & $\begin{pmatrix} 0 \\ \textcolor{red}{-1} \\ \textcolor{red}{1} \end{pmatrix}$ & $+$ & $\begin{pmatrix} 0 \\ 0 \\ \textcolor{red}{-1} \end{pmatrix}$ \\
\CodeAfter
\begin{tikzpicture}
\draw[-{.latex}, ultra thick, opacity=0.3, transform canvas={yshift=12pt}] (2-5) to  (2-3);
\draw[-{.latex}, ultra thick, opacity=0.3] (2-5) to  (2-7);
\draw[-{.latex}, ultra thick, opacity=0.3, transform canvas={yshift=-12pt}] (2-7) to  (2-9);
\end{tikzpicture}
\end{NiceTabular}
\end{center}
\end{minipage}
}\\
\subfloat[A chasing procedure ends with a loop.\label{subfig:tp2ring}]{
\begin{minipage}{\textwidth}
\begin{center}
\begin{NiceTabular}{ c c c c c c c c c c c c c }
$0$ & $=$ & $r_{a_{-1}}^{\mc G_{-1}}\partial(\bm\vartheta_{a_{-1}}^{\mc G_{-1}})$ & $+$ & $r_{a_0}^{\mc G_0}\partial(\bm\vartheta_{a_0}^{\mc G_0})$ & $+$ & $r_{a_1}^{\mc G_1}\partial(\bm\vartheta_{a_1}^{\mc G_1})$ & $+$ & $r_{a_2}^{\mc G_2}\partial(\bm\vartheta_{a_2}^{\mc G_2})$ & $+$ & $r_{a_3}^{\mc G_3}\partial(\bm\vartheta_{a_3}^{\mc G_3})$ & $+$ & $r_{a_4}^{\mc G_4}\partial(\bm\vartheta_{a_4}^{\mc G_4})$ \\
$\begin{pmatrix} 0 \\ 0 \\ 0 \\ 0 \end{pmatrix}$ & $=$ & $\begin{pmatrix} 0 \\ 0 \\ 0 \\ 1 \end{pmatrix}$ & $+$ & $\begin{pmatrix} 1 \\ 0 \\ 0 \\ -1 \end{pmatrix}$ & $+$ & $\begin{pmatrix} \textcolor{red}{-2} \\ \textcolor{red}{2} \\ 0 \\ 0 \end{pmatrix}$ & $+$ & $\begin{pmatrix} 0 \\ \textcolor{red}{-1} \\ \textcolor{red}{1} \\ 0 \end{pmatrix}$ & $+$ & $\begin{pmatrix} \textcolor{red}{1} \\ 0 \\ \textcolor{red}{-1} \\ 0 \end{pmatrix}$ & $+$ & $\begin{pmatrix} 0 \\ -1 \\ 0 \\ 0 \end{pmatrix}$ \\
& & & & & & \huge$\Downarrow$ & & & & & & \\
$\begin{pmatrix} 0 \\ 0 \\ 0 \\ 0 \end{pmatrix}$ & $=$ & $\begin{pmatrix} 0 \\ 0 \\ 0 \\ 0 \end{pmatrix}$ & $+$ & $\begin{pmatrix} 0 \\ 0 \\ 0 \\ 0 \end{pmatrix}$ & $+$ & $\begin{pmatrix} \textcolor{red}{-1} \\ \textcolor{red}{1} \\ 0 \\ 0 \end{pmatrix}$ & $+$ & $\begin{pmatrix} 0 \\ \textcolor{red}{-1} \\ \textcolor{red}{1} \\ 0 \end{pmatrix}$ & $+$ & $\begin{pmatrix} \textcolor{red}{1} \\ 0 \\ \textcolor{red}{-1} \\ 0 \end{pmatrix}$ & $+$ & $\begin{pmatrix} 0 \\ 0 \\ 0 \\ 0 \end{pmatrix}$ \\
\CodeAfter
\begin{tikzpicture}
\draw[-{.latex}, ultra thick, opacity=0.3, transform canvas={yshift=-18pt}] (2-5) to  (2-3);
\draw[-{.latex}, ultra thick, opacity=0.3, transform canvas={yshift=18}] (2-5) to  (2-7);
\draw[-{.latex}, ultra thick, opacity=0.3, transform canvas={yshift=6}] (2-7) to  (2-9);
\draw[-{.latex}, ultra thick, opacity=0.3, transform canvas={yshift=-6}] (2-9) to  (2-11);
\draw[-{.latex}, ultra thick, opacity=0.3, transform canvas={yshift=18}, dashed] (2-11) to  (2-7);
\end{tikzpicture}
\end{NiceTabular}
\end{center}
\end{minipage}
}
\caption{Illustration of the second stage of the algorithm. In each subgraph, the upper equation stands for some $\bm c'$ such that $\partial(\bm c')=0$ while the bottom stands for a reduced cycle $\bm c\in L_R^G$ constructed based on $\bm c'$. Concrete examples of $\bm c$ corresponding to each class is given in Fig.~\ref{fig:IlluConcrete}.}
\label{fig:algIllu}
\end{figure}

\begin{figure}[!htp]
    \centering
    \includegraphics[width=0.8\linewidth]{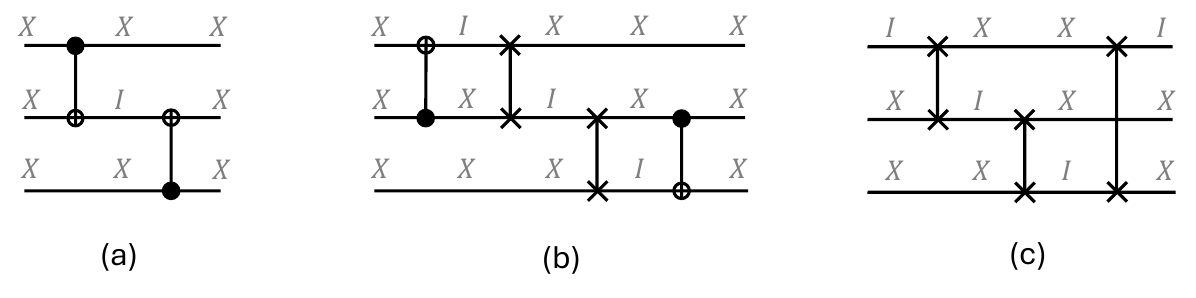}
    \caption{Concrete examples of three possible classes of reduced cycles involving more than one gate, corresponding to Fig.~\ref{fig:algIllu}. The 2-qubit gates appearing here are CNOT and SWAP.}
    \label{fig:IlluConcrete}
\end{figure}

Now the algorithm proceeds to the second stage, where we find reduced cycles across different supports.
This stage is illustrated in Figure~\ref{fig:algIllu}.
According to Lemma~\ref{le:LR_iff} we associate the base vectors in $A$ with their boundaries.
Thanks to Lemma~\ref{le:shape}, there are only two types of boundaries: Type 1 boundary takes $\pm1$ at one entry and $0$ at other places, corresponding to edges (10,11) or (01,11) in Fig.~\ref{fig:PTMs}.
Type $2$ boundary takes $1$ at one entry, $-1$ at another entry, and $0$ at other places, corresponding to edge (10,01) in Fig.~\ref{fig:PTMs}.
Moreover, if we consider the edge of reverse direction, the corresponding boundary is negated.
In the following we will find some candidate $\bm c=\sum\mathrm{sgn}_i\bm\vartheta_i$ and then prove that it is a reduced cycle when all the signs are positive.
When some signs are negative, we revert some edges to make the signs positive.
This is actually a simplification and hides the following fact: though $\bm c$ itself is not a reduced cycle, $\tilde{\bm c}=\bm c+\sum_{\mathrm{sgn}_i=-1}(\bm\vartheta_i+\bm\vartheta_i')$ ($\bm c$ with edge directions reverted) is a reduced cycle.
Moreover, since the $\bm\vartheta_i+\bm\vartheta_i'$s are already in $C$, putting $\bm c$ or $\tilde{\bm c}$ into $C$ is equivalent in terms of making $C$ a basis for $\widetilde{L_R^G}$.

According to the proof of Lemma~\ref{le:alg_correct}, if there are remaining degree of freedom in $\widetilde{L_R^G}$, there should exist a non-zero vector $\bm c'\in \widetilde{L_R^G}\cap\mr{supp}(A)$.
Then $\partial(\bm c')=\bm0$.
We can find such $\bm c'=\sum_ir_{a_i}^{\mc G_i}\bm\vartheta_{a_i}^{\mc G_i}$ (the summation is only over non-zero terms) via simple linear algebra on the boundaries of the edges in $A$.
The $r_{a_i}^{\mc G_i}$ found is rational and thus can be made integral by multiplying them by the least common multiple of the denominators.

{We now explain how to find a reduced cycle $\bm c\in L_R^G$ based on $\bm c'$.}
Expand $\partial(\bm c')=\sum_ir_{a_i}^{\mc G_i}\partial(\bm\vartheta_{a_i}^{\mc G_i})=\bm0$.
Among the summation, if there are two type-1 boundaries with non-zero entry at the same qubit
(reverting the edge directions if needed, wlog assume their sum is $\bm0$: $\partial(\bm\vartheta_{a_i}^{\mc G_i})+\partial(\bm\vartheta_{a_j}^{\mc G_j})=\bm0$), then $\bm\vartheta_{a_i}^{\mc G_i}+\bm\vartheta_{a_j}^{\mc G_j}$ is a reduced cycle that can be chosen as $\bm c$.
To see this, first observe that $\mr{supp}(\mc G_i)\neq\mr{supp}(\mc G_j)$, since otherwise $\bm\vartheta_{a_i}^{\mc G_i}$ and $\bm\vartheta_{a_j}^{\mc G_j}$ are two edges of a spanning tree.
But $\mr{supp}(\mc G_i)$ and $\mr{supp}(\mc G_j)$ intersect at the shared non-zero entry of the two boundaries.
So wlog assume $\mr{supp}(\mc G_i)=\{1,2\}$, $\mr{supp}(\mc G_j)=\{2,3\}$, $\partial(\bm\vartheta_{a_i}^{\mc G_i})=[0,1,0,\cdots]$ and $\partial(\bm\vartheta_{a_j}^{\mc G_j})=[0,-1,0,\cdots]$.
Then $\bm\vartheta_{a_i}^{\mc G_i}+\bm\vartheta_{a_j}^{\mc G_j}=\mc Q^\T(\bm e_{(a_i)_{12}X_3I_4\cdots I_n}^{\mc G_i}+\bm e_{X_1(a_j)_{23}I_4\cdots I_n}^{\mc G_j})$, 
so it is a reduced cycle.

If there are no two type-1 boundaries with the same non-zero entry, since the sum of the boundaries is $\bm0$, there must be some type-2 boundaries in the summation.
Pick a type-2 boundary $r_{a_0}^{\mc G_0}\partial{\bm\vartheta_{a_0}^{\mc G_0}}$.
This vector has a positive entry.
So there must be another boundary $r_{a_1}^{\mc G_1}\partial{\bm\vartheta_{a_1}^{\mc G_1}}$ that is negative on the entry.
Either it is a type-1 boundary and the procedure terminates here, or it is a type-2 boundary and it has another positive entry, and we repeat this chasing procedure.
The procedure cannot last infinitely, so at some point we must terminate at a type-1 boundary, or we have chased back to an entry we have already considered.
We do the same chasing procedure for the negative entry of $r_{a_0}^{\mc G_0}\partial{\bm\vartheta_{a_0}^{\mc G_0}}$.

If both chasing procedure terminates with a type-1 boundary, we get the following vector of boundary $\bm0$: $\sum_{i=-t_1}^{t_2}\mr{sgn}(r_{a_i}^{\mc G_i})\bm\vartheta_{a_i}^{\mc G_i}$, where $\bm\vartheta_{a_{-t_1}}^{\mc G_{-t_1}}$ and $\bm\vartheta_{a_{t_2}}^{\mc G_{t_2}}$ have type-1 boundaries and the others have type-2 boundaries.
Reverting the directions of basis if necessary, we can further assume all the signs are $+1$.
Then we get a reduced cycle that can be used as $\bm c$.
To see this, note that since type-2 boundaries only come from SWAP-like pattern transfer graph, different type-2 boundaries have different supports (type-2 boundaries have the same support as the corresponding $\mc G$).
{Let $m=t_1+t_2-1$. Wlog we assume the support of $\mc G_{-t_1+1},\ldots,\mc G_{t_2-1}$ are 
$\{1,2\},\{2,3\},\{3,4\},\cdots,\{m,m+1\}$, and that $\partial(\bm\vartheta_{a_{-t_1}}^{\mc G_{-t_{1}}})=[1,0,\cdots,0]$, $\partial(\bm\vartheta_{a_{t_2}}^{\mc G_{t_{2}}})=[\cdots,0,-1,0,\cdots]$. Then, 
\begin{equation}\label{eq:type1-type2}
\begin{aligned}
    \sum_{i=-t_1}^{t_2}\bm\vartheta_{a_i}^{\mc G_i}=&\mc Q^\T(\bm e_{P_1Q_2Q_3\cdots Q_n}^{\mc G_{-t_1}}+\bm e_{(a_{-t_1+1})_{12}Q_3\cdots Q_n}^{\mc G_{-t_1+1}}+\bm e_{Q_1(a_{-t_1+2})_{23}Q_3\cdots Q_n}^{\mc G_{-t_1+2}}+\cdots+\\
    &\bm e_{Q_1\cdots Q_{m-1}(a_{t_2-1})_{m,m+1}Q_{m+2}\cdots Q_n}^{\mc G_{t_2-1}} + \bm e_{Q'_1\cdots Q'_mI_{m+1}Q'_{m+2}\cdots Q'_n}^{\mc G_{t_2}} ),
\end{aligned}
\end{equation}
where we use $a_{-t_1} = P_1Q_k$ and $a_{t_2} = Q'_{k'}I_{m+1}$, for some indexes $k,k'$ and non-trivial single-qubit Pauli operators $P,Q,Q'$. Note that the patterns of $a_{-t_1+1},\cdots,a_{t_2-1}$ are all $01$, according to their corresponding boundaries. It is not hard to see the above equation represents the action of a sequence of SWAP-like gates, thus yields a reduced cycle which can be chosen as $\bm c$. An example is given in Fig.~\ref{fig:algIllu}b.}

Finally, consider the case where some chasing procedure ends with a loop.
That is, we have found a cycle of type-2 boundaries $\sum_{i=1}^m\mr{sgn}(r_{a_i}^{\mc G_i})\bm\vartheta_{a_i}^{\mc G_i}$.
Revert directions if necessary, wlog we can assume that all the signs are positive.
Wlog assume the support of $\mc G_1,\mc G_2,\cdots,\mc G_m$ are $\{1,2\},\{2,3\},\cdots,\{m,1\}$, and that $\partial(\bm\vartheta_{a_1}^{\mc G_1})$ = $[-1,1,0,\cdots]$.
Then,
\begin{equation}
    \sum_i\bm\vartheta_{a_i}^{\mc G_i}=\mc Q^\T(\bm e_{(a_1)_{12}X_3X_4\cdots X_n}^{\mc G_1}+\bm e_{X_1(a_2)_{23}X_4\cdots X_n}^{\mc G_2}+\cdots + \bm e_{X_2\cdots X_{m-1}(a_m)_{m,1}X_{m+1}\cdots X_n}^{\mc G_m})
\end{equation}
so it is a reduced cycle.
Note that the patterns of $a_{1},\cdots,a_{m}$ are all $01$ according to their corresponding boundaries.
This case is illustrated in Figure~\ref{subfig:tp2ring}.
Thus we see that as long as there are remaining degree of freedom in $\widetilde{L_R^G}$, we can find a reduced cycle, so algorithm 3 can terminate correctly. Since all the basis vectors are also elements of $L_R^G$, we have proved in this case $L_R^G = \widetilde{L_R^G}$.

\section{About quasi-local Pauli channels}\label{app:factorize}

Throughout this section, when we talk about a Pauli channel, we always assume all of its eigenvalues to be strictly positive, so that $x_a = -\log\lambda_a$ is well-defined.

\begin{lemma}\label{le:quasi_local_complete}
    Let $\Omega = 2^{[n]}\backslash\varnothing$. Then any $n$-qubit Pauli channel is $\Omega$-local.
\end{lemma}
\begin{proof}
    For any set of parameters $\{x_a:a\in{\sf P}^n\}$, we claim there always exist $\{r_b:b\in{\sf P}^n\}$ such that,
    \begin{align}
        x_a &= \sum_{b\triangleleft a} r_b,\quad&&\forall a\in{\sf P}^n.\label{eq:in_out_1} \\
        r_b &= \sum_{a\triangleleft b}(-1)^{w(a)-w(b)}x_a,\quad&&\forall b\in{\sf P}^n.\label{eq:in_out_2}
    \end{align}
    Indeed, substituting the second expression into R.H.S. of the first yields
    \begin{equation}
        \begin{aligned}
            \sum_{b\triangleleft a}\sum_{c\triangleleft b}(-1)^{w(c)-w(b)}x_c &= \sum_{c\triangleleft a}x_c\sum_{b: c\triangleleft b\triangleleft a}(-1)^{w(c)-w(b)}\\
            &= \sum_{c\triangleleft a}x_c\sum_{k=0}^{|w(a)-w(c)|}\binom{w(a)-w(c)}{k}(-1)^{k}\\
            &= \sum_{c\triangleleft a}x_c\mathds 1[w(a)=w(c)]\\
            &= x_a.
        \end{aligned}
    \end{equation}
    The first line changes the order of summation. The third line uses the binomial theorem. Further note that $x_{I_n} = 0\Leftrightarrow r_{I_n}=0$. Since any Pauli channel is trace-preserving (i.e., $x_{I_n}=0$), it is $2^{[n]}\backslash\varnothing$-local by choosing $r_b$ according to Eq.~\eqref{eq:in_out_2}.
\end{proof}

\medskip

The following Lemma explains the equivalence between the quasi-local Pauli channels we consider and the sparse Pauli-Lindblad noise model.
We note that basically the same results are given by \cite[ Theorem SIV.1]{van2023probabilistic}.
\begin{lemma}\label{le:factorize}
    Given an $n$-qubit Pauli channel $\Lambda$ and a factor set $\Omega$, if $\Lambda$ can be written in the following concatenating form
    \begin{equation}\label{eq:eta_form}
        \Lambda(\cdot) = \mathop{\bigcirc}_{b\sim \Omega}\left(\eta_b P_b(\cdot) P_b + (1-\eta_b)(\cdot)\right),
    \end{equation}
    with real parameters $\{\eta_b\}_{b\sim \Omega}$ such that $\eta_b<1/2,~\forall b\neq I_n$,
    then $\Lambda$ is $\Omega$-local. Conversely, any $\Omega$-local Pauli channel can be written in the above form.
\end{lemma}

\begin{proof}

    Suppose Eq.~\eqref{eq:eta_form} holds. By definition of the Pauli eigenvalues,
    \begin{equation}
        \begin{aligned}
            \lambda_a = \Tr[P_a\Lambda(P_a)]/2^n=\prod_{b\sim \Omega}(1-2\eta_b)^{\expval{a,b}}.
        \end{aligned}
    \end{equation}
    Taking a log on both sides yields (recall that we require $\eta_b<1/2$)
    \begin{equation}\label{eq:x_factor_wh}
        x_a\coleq -\log(\lambda_a) = -\sum_{b\sim \Omega} \expval{a,b}\log(1-2\eta_b) \eqcol \sum_{b\sim \Omega}\expval{a,b}\tau_b,
    \end{equation}
    where we define $\tau_b\coleq -\log(1-2\eta_b)$. On the other hand, recall $\{r_b\}$ as defined in Eq.~\eqref{eq:in_out_2}.
    Substituting Eq.~\eqref{eq:x_factor_wh} into the expression of $r_b$ yields,
    \begin{equation}
    \begin{aligned}
        r_b &= \sum_{c\sim \Omega}\sum_{a\triangleleft b}(-1)^{w(a)-w(b)}\expval{a,c}\tau_c.
    \end{aligned}
    \end{equation}
    Now we claim that $r_b = 0$ unless $b\sim \Omega$. If the condition does not hold, for any $c\sim \Omega$, there is at least one $i\in[n]$ such that $c_i=I$ but $b_i\neq I$. For any $a\triangleleft b$ such that $a_i=I$, let $\tilde a$ be obtained by replacing the $i$th entry of $a$ by $b_i$. Then we obviously have $$(-1)^{w(a)-w(b)}\expval{a,c}+(-1)^{w(\tilde a)-w(b)}\expval{\tilde a,c}=0.$$
    Further note that $\{a\in{\sf P}^n:a\triangleleft b\}$ can be decomposed into disjoint pairs of the above form ($a$, $\tilde a$), we thus conclude that $r_b$ must be $0$, justifying our claim. We can thus rewrite the expression for $\bm x$ as
    \begin{equation}\label{eq:x_in_r_Omega}
        x_a = \sum_{b\triangleleft a,b\sim \Omega}r_b,    
    \end{equation}
    depending on a set of real parameters $\{r_b\}_{b\sim \Omega}$.
    Exponentiating both sides yields,
    \begin{equation}
        \lambda_a = \prod_{b\triangleleft a,b\sim \Omega}\exp(-r_b).
    \end{equation}
    This completes the proof that $\Lambda$ is $\Omega$-local.

    \medskip

    Conversely, let us assume that $\Lambda$ is $\Omega$-local, i.e., Eq.~\eqref{eq:x_in_r_Omega} holds. 
    Define $\tau_b$ according to
    \begin{equation}\label{eq:tau_in_x}
        \tau_b \coleq \mathds 1[b\neq I_n]\cdot\frac{-2}{4^n}\sum_{a}(-1)^{\expval{a,b}}x_a,\quad\forall b\in{\sf P}^n,
    \end{equation}
    which gives us $x_a = \sum_{b}\expval{a,b}\tau_b$. Indeed,
    \begin{equation}
    \begin{aligned}
        \sum_b\expval{a,b}\tau_b &= -\frac{2}{4^n}\sum_{b\neq I_n}\expval{a,b}\sum_c (-1)^{\expval{c,b}}x_c\\  
        &=-\frac{2}{4^n}\sum_{b\neq I_n}\frac{1-(-1)^{\expval{a,b}}}{2}\sum_c (-1)^{\expval{c,b}}x_c\\
        &= \frac{1}{4^n}\sum_c x_c \sum_{b\neq I_n}\left(-(-1)^{\expval{c,b}}+(-1)^{\expval{a+c,b}}\right)\\
        &= \sum_c x_c \left(-\mathds 1[c=I_n]+\mathds 1[c= a]\right)      \\
        &= x_a - x_{I_n},
    \end{aligned}
    \end{equation}
    which is just $x_a$ since $x_{I_n}=0$ by the trace preserving condition. Now substituting Eq.~\eqref{eq:x_in_r_Omega} into Eq.~\eqref{eq:tau_in_x}, considering only $b\neq I_n$,
    \begin{equation}
    \begin{aligned}
        \tau_b &= -\frac2{4^n}\sum_a(-1)^{\expval{a,b}}\sum_{c\triangleleft a,c\sim\Omega}r_c\\
        &= -\frac2{4^n}\sum_{c\sim\Omega}r_c\sum_{a:c\triangleleft a}(-1)^{\expval{a,b}}.
    \end{aligned}
    \end{equation}
    We claim that $\tau_b= 0$ unless $b\sim\Omega$. Similar to the previous case, if $b\sim\Omega$ does not hold, for any $c\sim\Omega$ there exists an index $i$ such that $c_i=I$ but $b_i\neq I$. Consequently, exactly half of the Pauli operators from $\{a:c\triangleleft a\}$ commute with $b$ while the other half anti-commute with $b$, depending on whether $a_i$ commutes with $b_i$. Thus, the above sum always yields zero, proving our claim. By setting $\eta_b = (1-\exp(-\tau_b))/2$ for all $b\sim\Omega$, we see Eq.~\eqref{eq:eta_form} is a valid representation for $\Lambda$. This completes our proof.             
\end{proof}

\medskip

We are now ready to give a proof for Lemma~\ref{le:covariant_sufficient}.
\begin{proof}[Proof for Lemma~\ref{le:covariant_sufficient}]
    Thanks to Lemma~\ref{le:factorize}, any $\Omega$-local Pauli channel $\Lambda$ can be written in the form of Eq.~\eqref{eq:eta_form}. Conjugate $\Lambda$ with $\mc G$ yields,
    \begin{equation}
    \begin{aligned}
        \mc G\Lambda\mc G^\dagger &= \mathop{\bigcirc}_{b\sim\Omega}\left(\eta_b \mc G\left(P_b(\mc G^\dagger(\cdot))P_b\right) + (1-\eta_b)(\cdot)\right)\\
        &= \mathop{\bigcirc}_{b\sim\Omega}\left(\eta_b P_{\mc G(b)}(\cdot)P_{\mc G(b)} + (1-\eta_b)(\cdot)\right).
    \end{aligned}
    \end{equation}
    If one can show that $\mc G$ is a permutation acting on $\{b:b\sim\Omega\}$, by simple relabeling one can see $\mc G\Lambda\mc G^\dagger$ also has the form of Eq.~\eqref{eq:eta_form} and is thus $\Omega$-local.
    For any $b\sim\Omega$, there must be a $\nu\in\Omega$ such that $\supp(b)=\nu$.    
    it is not hard to see that
    \begin{equation}
        \supp(\mc G(b)) \subseteq \Xi_{\mc G}(\nu).
    \end{equation}
    By assumption, $\Xi_{\mc G}(\nu)\in\Omega$, thus $\supp(\mc G(b))\in\Omega$ by the property of $\Omega$, which means $\mc G(b)\sim\Omega$. Since $\mc G$ is invertible, this completes the proof that $\mc G$ acts as a permutation over $\{b:b\sim\Omega\}$, and thus proving $\mc G\Lambda\mc G^\dagger$ is $\Omega$-local. Since the same argument holds by replacing $\mc G$ with $\mc G^\dagger$, this completes the proof of Lemma~\ref{le:covariant_sufficient}.
\end{proof}

\section{Additional proofs for Sec.~\ref{sec:NN_CZ}}
\subsection{Gauge in the nearest-neighbor \texorpdfstring{$CZ$}{CZ} model}\label{sec:proof_NN_CZ}

\begin{proposition}
    For the gate set and noise model described in Sec.~\ref{sec:NN_CZ}, the embedded gauge space is given by $\mc Q(T_R)=\mr{span}\{\mf d_{\{i\}}:i\in[n]\}$.
\end{proposition}

\begin{proof}
    Since both state preparation and measurement noise are $\Omega$-local, Theorem~\ref{th:quasi_learnability}(a) says,
    \begin{equation}
        \{\bm z\in T:{\pi}^S(\bm z)\in\imq^S,~{\pi}^M(\bm z)\in\imq^M\} =
        \mr{span}\{\mf d_\nu:\nu\in\Omega\}=\mr{span}\{\mf d_{\{i\}},~\mf d_{\{i,i+1\}}:i\in[n]\}.
    \end{equation}
    Then, by Lemma~\ref{le:layerwise_indep},
    \begin{equation}
    \begin{aligned}
        \mc Q(T_R) &= \mr{span}\{\mf d_{\{i\}},~\mf d_{\{i,i+1\}}:i\in[n]\}\cap\{\bm z\in T:{\pi}^{\mc G_e}(\bm z)\in\imq^{\mc G_e},~{\pi}^{\mc G_o}(\bm z)\in\imq^{\mc G_o}\}\\
        &= \left\{\bm z\in \mr{span}\{\mf d_{\{i\}},~\mf d_{\{i,i+1\}}:i\in[n]\}:{\pi}^{\mc G_e}(\bm z)\in\imq^{\mc G_e},~{\pi}^{\mc G_o}(\bm z)\in\imq^{\mc G_o}\right\}.
    \end{aligned}
    \end{equation}
    We are going to show that
    \begin{align}
        \left\{\bm z\in \mr{span}\{\mf d_{\{i\}},~\mf d_{\{i,i+1\}}:i\in[n]\}:{\pi}^{\mc G_e}(\bm z)\in\imq^{\mc G_e}\right\} &= \mr{span}\left\{\{\mf d_{\{i\}}:i\in[n]\}\cup\{\mf d_{\{2k-1,2k\}}:k\in[n/2]\}\right\},\label{eq:AppE_1}\\
        \left\{\bm z\in \mr{span}\{\mf d_{\{i\}},~\mf d_{\{i,i+1\}}:i\in[n]\}:{\pi}^{\mc G_o}(\bm z)\in\imq^{\mc G_o}\right\} &= \mr{span}\left\{\{\mf d_{\{i\}}:i\in[n]\}\cup\{\mf d_{\{2k,2k+1\}}:k\in[n/2]\}\right\},\label{eq:AppE_2}
    \end{align}
    which would then implies $\mc Q(T_R)=\mr{span}\{\mf d_{\{i\}}:i\in[n]\}$. Let us look at Eq.~\eqref{eq:AppE_1} first. To see $\mr{L.H.S.}\supseteq\mr{R.H.S.}$, note that $\mc G_e^\dagger\mc D_{\{i\}}\mc G_e\mc D_{\{i\}}^{-1}$ and $\mc G_e^\dagger\mc D_{\{2k-1,2k\}}\mc G_e\mc D_{\{2k-1,2k\}}^{-1}$ are both clearly $\Omega$-local for all $i\in[n],k\in[n/2]$. Here $\mc D_\nu$ is the subsystem depolarizing channel on $\nu$ defined in Eq.~\eqref{eq:sd_channel}. 
    Using the same argument as in the proof of Theorem~\ref{th:quasi_learnability}(b), we can conclude those SDGs are indeed allowed.
    To see $\mr{L.H.S.}\subseteq\mr{R.H.S}$, it remains to show that
    $\mr{span}\{\mf d_{\{2k,2k+1\}}:k\in[n/2]\}\cap \mr{L.H.S.} = \{\bm 0\}$. 
    Consider any vector of the form $\bm z = \sum_{k\in[n/2]}\alpha_k \mf d_{\{2k,2k+1\}}$. By definition of the SDG,
    \begin{equation}\label{eq:AppE_main}
        z^{\mc G_e}_a \coleq \bm e_a^{\mc G_e}\cdot\bm z = \sum_{k\in[n/2]}\alpha_k\left(\mathds 1[\pt(\mc G(a))_{\{2k,2k+1\}}\neq 00]-\mathds 1[\pt(a)_{\{2k,2k+1\}}\neq 00]\right).
    \end{equation}
    Suppose ${\pi}^{\mc G_e}(\bm z)\in\imq^{\mc G_e}$, then $z_a^{\mc G_e} = \sum_{b\triangleleft a,b\sim\Omega}t_b$ for some parameters $t_b$. Specifically, for any $k\in[n/2]$, if we choose $a=X_{2k-1}X_{2k+3}$, since $\{2k-1,2k+3\}\notin\Omega$ given that $n\ge6$, we have
    \begin{equation}
        z_{X_{2k-1}X_{2k+3}}^{\mc G_e} = z^{\mc G_e}_{X_{2k-1}} + z^{\mc G_e}_{X_{2k+3}}.
    \end{equation}
    On the other hand, Eq.~\eqref{eq:AppE_main} gives (see Fig.~\ref{fig:AppE} for how these are computed)
    \begin{equation}
    \begin{aligned}
        z_{X_{2k-1}X_{2k+3}}^{\mc G_e} &= \alpha_{k-1}(1-1) + \alpha_{k}(1-0) + \alpha_{k+1}(1-1) = \alpha_k.\\
        z^{\mc G_e}_{X_{2k-1}}&= \alpha_{k-1}(1-1) + \alpha_{k}(1-0) = \alpha_k.\\
        z^{\mc G_e}_{X_{2k+3}}&= \alpha_{k+1}(1-1) + \alpha_{k}(1-0) =\alpha_k.
    \end{aligned}
    \end{equation}
    \begin{figure}[!htp]
        \centering
        \includegraphics[width=0.7\linewidth]{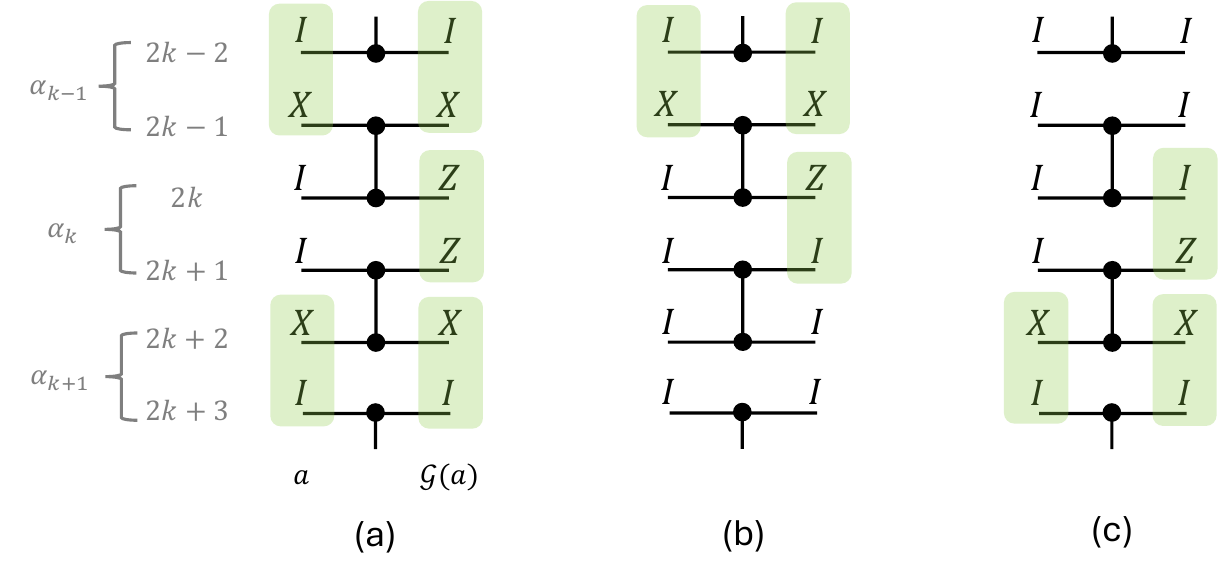}
        \caption{Illustration for calculations of $z_a^{\mc G_e}$ where (a) $a=X_{2k-1}X_{2k+3}$, (b) $a=X_{2k-1}$, and (c) $a=X_{2k+3}$, respectively. The labels of qubits and the corresponding coefficient $\alpha_k$ is shown in the leftmost side.   
        The shaded region represents a non-zero pattern in each $\{2k,2k+1\}$ subsystem, for both $a$ and $\mc G_e(a)$.}
        \label{fig:AppE}
    \end{figure}    
    \noindent But this means $\alpha_k$ must be $0$, which finishes the proof that $\mr{span}\{\mf d_{\{2k,2k+1\}}:k\in[n/2]\}\cap \mr{L.H.S.} = \{\bm 0\}$. Eq.~\eqref{eq:AppE_2} can be proved similarly by cyclically shifting every index by $1$. Combining everything, we complete the proof that $\mc Q(T_R) = \mr{span}\{\mf d_{\{i\}}:i\in[n]\}$.
\end{proof}

\subsection{Cycle basis for the nearest-neighbor \texorpdfstring{$CZ$}{CZ} model}
\label{sec:basis_NN_CZ}

To prove that Eqs.~\eqref{eq:NN_CZ_basis1},~\eqref{eq:NN_CZ_basis2},~\eqref{eq:NN_CZ_basis3} constitute a basis for $L_R^G$, it suffices to prove that they form a complete basis for $\widetilde{L_R^G}$ since $L_R^G\subseteq\widetilde{L_R^G}$.
We may further perform an invertible linear operation on them first. 

Take $\mc Q^\T\left(\bm e_{I_1X_2I_3X_4I_5}^{\mc G_o} + \bm e_{I_1X_2Z_3X_4Z_5}^{\mc G_e} + \bm e_{Z_1X_2I_3X_4Z_5}^{\mc G_o} + \bm e_{Z_1X_2Z_3X_4I_5}^{\mc G_e}\right)$ from Eq.~\eqref{eq:NN_CZ_basis3} as an example.
We change it into
\begin{align}
&\frac{1}{2}\mc Q^\T\left(\bm e_{I_1X_2I_3X_4I_5}^{\mc G_o} + \bm e_{I_1X_2Z_3X_4Z_5}^{\mc G_e} + \bm e_{Z_1X_2I_3X_4Z_5}^{\mc G_o} + \bm e_{Z_1X_2Z_3X_4I_5}^{\mc G_e}\right)-\frac{1}{4}\mc Q^\T\left(\bm e_{Z_1X_2I_3}^{\mc G_o}+\bm e_{Z_1X_2Z_3}^{\mc G_o}\right)\nonumber\\
&-\frac{1}{2}\mc Q^\T\left(\bm e_{I_1X_2Z_3}^{\mc G_e}+\bm e_{Z_1X_2Z_3}^{\mc G_e}\right)-\frac{1}{4}\mc Q^\T\left(\bm e_{I_3X_4Z_5}^{\mc G_e}+\bm e_{Z_3X_4Z_5}^{\mc G_e}\right)+\frac{1}{4}\mc Q^\T\left(\bm e_{X_2Z_3}^{\mc G_o}+\bm e_{X_2I_3}^{\mc G_o}\right)\nonumber\\
&-\frac{1}{2}\mc Q^\T\left(\bm e_{X_4Z_5}^{\mc G_o}+\bm e_{X_4I_5}^{\mc G_o}\right)-\frac{3}{4}\mc Q^\T\left(\bm e_{Z_3X_4}^{\mc G_e}+\bm e_{I_3X_4}^{\mc G_e}\right)+\mc Q^\T\left(\bm e_{Z_3I_4}^{\mc G_e}\right)\nonumber\\
=~&\frac{1}{2}\left(\bm\vartheta_{Z_1}^{\mc G_o}+2\bm\vartheta_{X_2}^{\mc G_o}+2\bm\vartheta_{X_4}^{\mc G_o}+\bm\vartheta_{Z_5}^{\mc G_o}+\bm\vartheta_{Z_1X_2}^{\mc G_o}+\bm\vartheta_{X_4Z_5}^{\mc G_o}+\bm\vartheta_{Z_1}^{\mc G_e}+2\bm\vartheta_{X_2}^{\mc G_e}+2\bm\vartheta_{Z_3}^{\mc G_e}+2\bm\vartheta_{X_4}^{\mc G_e}+\bm\vartheta_{Z_5}^{\mc G_e}\right.\nonumber\\
&+\left.\bm\vartheta_{Z_1X_2}^{\mc G_e}+2\bm\vartheta_{X_2Z_3}^{\mc G_e}+2\bm\vartheta_{Z_3X_4}^{\mc G_e}+\bm\vartheta_{X_4Z_5}^{\mc G_e}\right)-\frac{1}{4}\left(2\bm\vartheta_{Z_1}^{\mc G_o}+2\bm\vartheta_{X_2}^{\mc G_o}+\bm\vartheta_{Z_3}^{\mc G_o}+2\bm\vartheta_{Z_1X_2}^{\mc G_o}+\bm\vartheta_{X_2Z_3}^{\mc G_o}\right)\nonumber\\
&-\frac{1}{2}\left(\bm\vartheta_{Z_1}^{\mc G_e}+2\bm\vartheta_{X_2}^{\mc G_e}+2\bm\vartheta_{Z_3}^{\mc G_e}+\bm\vartheta_{Z_1X_2}^{\mc G_e}+2\bm\vartheta_{X_2Z_3}^{\mc G_e}\right)-\frac{1}{4}\left(\bm\vartheta_{Z_3}^{\mc G_e}+2\bm\vartheta_{X_4}^{\mc G_e}+2\bm\vartheta_{Z_5}^{\mc G_e}+\bm\vartheta_{Z_3X_4}^{\mc G_e}+2\bm\vartheta_{X_4Z_5}^{\mc G_e}\right)\nonumber\\
&+\frac{1}{4}\left(2\bm\vartheta_{X_2}^{\mc G_o}+\bm\vartheta_{Z_3}^{\mc G_o}+\bm\vartheta_{X_2Z_3}^{\mc G_o}\right)-\frac{1}{2}\left(2\bm\vartheta_{X_4}^{\mc G_o}+\bm\vartheta_{Z_5}^{\mc G_o}+\bm\vartheta_{X_4Z_5}^{\mc G_o}\right)-\frac{3}{4}\left(\bm\vartheta_{Z_3}^{\mc G_e}+2\bm\vartheta_{X_4}^{\mc G_e}+\bm\vartheta_{Z_3X_4}^{\mc G_e}\right)+\left(\bm\vartheta_{Z_3}^{\mc G_e}\right)\nonumber\\
=~&\bm\vartheta_{X_2}^{\mc G_o}-\bm\vartheta_{X_4}^{\mc G_e},
\end{align}
where among the added terms, the first $3$ terms are from Eq.~\eqref{eq:NN_CZ_basis2} and the rest are from Eq.~\eqref{eq:NN_CZ_basis1}.
In this way, we can change the basis in Eq.~\eqref{eq:NN_CZ_basis3} into for $k=1,...,\frac n2$, 
\begin{equation}\label{eq:NN_CZ_basis4}
\left\{-\bm\vartheta_{X_{2k-1}}^{\mc G_e}+\bm\vartheta_{X_{2k+1}}^{\mc G_o},\bm\vartheta_{X_{2k}}^{\mc G_o}-\bm\vartheta_{X_{2k+2}}^{\mc G_e}\right\}.
\end{equation}
The change from Eqs.~\eqref{eq:NN_CZ_basis1},~\eqref{eq:NN_CZ_basis2},~\eqref{eq:NN_CZ_basis3} to Eqs.~\eqref{eq:NN_CZ_basis1},~\eqref{eq:NN_CZ_basis2},~\eqref{eq:NN_CZ_basis4} is obviously invertible, so now we prove that the latter is a basis for $\widetilde{L_R^G}$.

In fact, Eqs.~\eqref{eq:NN_CZ_basis1},~\eqref{eq:NN_CZ_basis2},~\eqref{eq:NN_CZ_basis4} is a valid output of Algorithm 3.
Initially, $A=\{\bm\vartheta_a^{\mc G_e},\bm\vartheta_a^{\mc G_o}:a\sim\Omega\}$.
We then put cycles from Eq.~\eqref{eq:NN_CZ_basis2} into $C$ one by one (the order does not matter).
For example, when we pick $\bm c=\mc Q^\T\left(\bm e_a^{\mc G_e}+\bm e_{a'}^{\mc G_e}\right)$ $(a\in{\sf P}^{2k,2k+1},w(a)=2,a'=\mc G_e(a))$ into $C$, we choose $\bm\vartheta_a^{\mc G_e}$ as $\bm\vartheta$ and remove it from $A$.
The case when $\mc G_e$ is replaced by $\mc G_o$ is similar.
These $\bm\vartheta$ removed does not appear in other cycles in Eq.~\eqref{eq:NN_CZ_basis2}, so the $\bm c$ picked each time indeed belongs to $\widetilde{L_R^G}\cap\mr{span}(A)$ at the time.

Now $C$ contains elements from Eq.~\eqref{eq:NN_CZ_basis2}, and $A$ can be partitioned into disjoint sets: $A=\bigcup_{k=1,\ldots,\frac{n}{2}}\{\bm\vartheta_a^{\mc G_e}:a\in{\sf P}^{2k-1,2k}\}\cup\bigcup_{k=1,\ldots,\frac{n}{2}}\{\bm\vartheta_a^{\mc G_e}:a\in{\sf P}^{2k-1,2k}\}$.
This partition corresponds to the partition in Eq.~\eqref{eq:NN_CZ_basis1}, and we put them into $C$ one by one.
For $k=1,\ldots,\frac{n}{2}$, first pick
\begin{equation}
\bm c=\mc Q^\T\left(\bm e_{X_{2k-1}X_{2k}}^{\mc G_e}\right),\mc Q^\T\left(\bm e_{Y_{2k-1}Y_{2k}}^{\mc G_e}\right),\mc Q^\T\left(\bm e_{X_{2k-1}Y_{2k}}^{\mc G_e}\right),\mc Q^\T\left(\bm e_{Y_{2k-1}X_{2k}}^{\mc G_e}\right),\mc Q^\T\left(\bm e_{Z_{2k-1}Z_{2k}}^{\mc G_e}\right),
\end{equation}
which corresponds to picking
\begin{equation}
\bm\vartheta=\bm\vartheta_{X_{2k-1}X_{2k}}^{\mc G_e},\bm\vartheta_{Y_{2k-1}Y_{2k}}^{\mc G_e},\bm\vartheta_{X_{2k-1}Y_{2k}}^{\mc G_e},\bm\vartheta_{Y_{2k-1}X_{2k}}^{\mc G_e},\bm\vartheta_{Z_{2k-1}Z_{2k}}^{\mc G_e}.
\end{equation}
Then pick in the order of
\begin{align}
\bm c=&\mc Q^\T\left(\bm e_{Y_{2k-1}Z_{2k}}^{\mc G_e}+\bm e_{Y_{2k-1}I_{2k}}^{\mc G_e}\right),\mc Q^\T\left(\bm e_{Y_{2k-1}I_{2k}}^{\mc G_e}+\bm e_{X_{2k-1}Z_{2k}}^{\mc G_e}\right),\mc Q^\T\left(\bm e_{X_{2k-1}Z_{2k}}^{\mc G_e}+\bm e_{X_{2k-1}I_{2k}}^{\mc G_e}\right),\nonumber\\
&\mc Q^\T\left(\bm e_{Z_{2k-1}Y_{2k}}^{\mc G_e}+\bm e_{I_{2k-1}Y_{2k}}^{\mc G_e}\right),\mc Q^\T\left(\bm e_{I_{2k-1}Y_{2k}}^{\mc G_e}+\bm e_{Z_{2k-1}X_{2k}}^{\mc G_e}\right),\mc Q^\T\left(\bm e_{Z_{2k-1}X_{2k}}^{\mc G_e}+\bm e_{I_{2k-1}X_{2k}}^{\mc G_e}\right),
\end{align}
which corresponds to picking
\begin{equation}
\bm\vartheta=\bm\vartheta_{Y_{2k-1}Z_{2k}}^{\mc G_e},\bm\vartheta_{Y_{2k-1}I_{2k}}^{\mc G_e},\bm\vartheta_{X_{2k-1}Z_{2k}}^{\mc G_e},\bm\vartheta_{Z_{2k-1}Y_{2k}}^{\mc G_e},\bm\vartheta_{I_{2k-1}Y_{2k}}^{\mc G_e},\bm\vartheta_{Z_{2k-1}X_{2k}}^{\mc G_e}.
\end{equation}
Finally pick
\begin{equation}
\bm c=\mc Q^\T\left(\bm e_{Z_{2k-1}I_{2k}}^{\mc G_e}\right),\mc Q^\T\left(\bm e_{I_{2k-1}Z_{2k}}^{\mc G_e}\right),
\end{equation}
which corresponds to picking
\begin{equation}
\bm\vartheta=\bm\vartheta_{Z_{2k-1}I_{2k}}^{\mc G_e},\bm\vartheta_{I_{2k-1}Z_{2k}}^{\mc G_e}.
\end{equation}
And similarly for replacing $\mc G_e$ by $\mc G_o$.
One can verify that in this order each time $\bm c\in\widetilde{L_R^G}\cap A$, and after picking cycles in Eq.~\eqref{eq:NN_CZ_basis1} in this order, $A=\{\bm\vartheta_{X_k}^{\mc G_e},\bm\vartheta_{X_k}^{\mc G_o}:k\in[n]\}$.

Then put elements in Eq.~\eqref{eq:NN_CZ_basis4} into $C$ one by one (order does not matter)
When $\bm c=-\bm\vartheta_{X_{2k-1}}^{\mc G_e}+\bm\vartheta_{X_{2k+1}}^{\mc G_o}$, pick $\bm\vartheta=\bm\vartheta_{X_{2k+1}}^{\mc G_o}$, and when $\bm c=\bm\vartheta_{X_{2k}}^{\mc G_o}-\bm\vartheta_{X_{2k+2}}^{\mc G_e}$ is picked, pick $\bm\vartheta=\bm\vartheta_{X_{2k+2}}^{\mc G_o}$.

At this stage $C$ contains all elements from Eqs.~\eqref{eq:NN_CZ_basis1},~\eqref{eq:NN_CZ_basis2},~\eqref{eq:NN_CZ_basis4} and $A=\{\bm\vartheta_{X_k}^{\mc G_e}:k\in[n]\}$.
It remains to prove that $\widetilde{L_R^G}\cap\mr{span}\left(\{\bm\vartheta_{X_k}^{\mc G_e}:k\in[n]\}\right)=\varnothing$.
We have
\begin{equation}
\sum_{k=1}^{\frac{n}{2}}a_{2k-1}\bm\vartheta_{X_{2k-1}}^{\mc G_e}+a_{2k}\bm\vartheta_{X_{2k}}^{\mc G_e}\in L_R\Leftrightarrow\sum_{k=1}^{\frac{n}{2}}-a_{2k-1}\bm\vartheta_{\{2k\}}^M-a_{2k}\bm\vartheta_{\{2k-1\}}^M\in L_R\Leftrightarrow\forall k,~a_k=0
\end{equation}
Where the first equivalence is because
\begin{align}
&\mc Q^\T\left(\bm e_{\{2k-1\}}^S+\bm e_{X_{2k-1}}^{\mc G_e}+\bm e_{\{2k-1,2k\}}^M\right)-\mc Q^\T\left(\bm e_{\{2k-1\}}^S+\bm e_{\{2k-1\}}^M\right)+\mc Q^\T\left(\bm e_{\{2k-2\}}^S+\bm e_{X_{2k-2}}^{\mc G_o}+\bm e_{\{2k-2,2k-1\}}^M\right)\nonumber\\
&-\mc Q^\T\left(\bm e_{\{2k-2,2k+1\}}^S+\bm e_{X_{2k-2}X_{2k+1}}^{\mc G_o}+\bm e_{\{2k-2,2k-1,2k,2k+1\}}^M\right)+\mc Q^\T\left(\bm e_{\{2k+1\}}^S+\bm e_{X_{2k+1}}^{\mc G_o}+\bm e_{\{2k,2k+1\}}^M\right)\nonumber\\
=&\bm\vartheta_{X_{2k-1}}^{\mc G_e}+\bm\vartheta_{\{2k\}}^M\in L_R
\end{align}
and similarly $\bm\vartheta_{X_{2k}}^{\mc G_e}+\bm\vartheta_{\{2k-1\}}^M\in L_R$.
The second equivalence is because $\bm\vartheta_{\{k\}}^M$ changes independently under the single-qubit depolarizing gauges.
Hence now $\widetilde{L_R^G}\cap\mr{span}(A)=\varnothing$ and Algorithm $3$ terminates.
By the correctness of Algorithm $3$, $C$, which contains Eqs.~\eqref{eq:NN_CZ_basis1},~\eqref{eq:NN_CZ_basis2},~\eqref{eq:NN_CZ_basis4} is a basis for $\widetilde{L_R^G}$ and thus a basis for $L_R^G$.
Thus Eqs.~\eqref{eq:NN_CZ_basis1},~\eqref{eq:NN_CZ_basis2},~\eqref{eq:NN_CZ_basis3} is a basis for $L_R^G$.

Next, we give a way to supplement Eqs.~\eqref{eq:NN_CZ_basis1},~\eqref{eq:NN_CZ_basis2},~\eqref{eq:NN_CZ_basis3} into a reduced cycle basis for $L_R$.
We claim the following reduced cycles work:
\begin{equation}\label{eq:NN_CZ_basis5}
\begin{gathered}
\left\{\mc Q^\T\left(\bm e_{\{2k-1\}}^S+\bm e_{X_{2k-1}}^{\mc G_e}+\bm e_{\{2k-1,2k\}}^M\right),\mc Q^\T\left(\bm e_{\{2k\}}^S+\bm e_{X_{2k}}^{\mc G_e}+\bm e_{\{2k-1,2k\}}^M\right):k=1,\ldots,\frac{n}{2}\right\}\\
\backslash\left\{\mc Q^\T\left(\bm e_{\{1\}}^S+\bm e_{X_{1}}^{\mc G_e}+\bm e_{\{1,2\}}^M\right)\right\}
\end{gathered}
\end{equation}
\begin{equation}\label{eq:NN_CZ_basis6}
\left\{\mc Q^\T\left(\bm e_{\{2k\}}^S+\bm e_{X_{2k}}^{\mc G_o}+\bm e_{\{2k,2k+1\}}^M\right),\mc Q^\T\left(\bm e_{\{2k+1\}}^S+\bm e_{X_{2k+1}}^{\mc G_o}+\bm e_{\{2k,2k+1\}}^M\right):k=1,\ldots,\frac{n}{2}\right\}
\end{equation}
\begin{equation}\label{eq:NN_CZ_basis7}
\left\{\mc Q^\T\left(\bm e_{\{k,k+1\}}^S+\bm e_{\{k,k+1\}}^M\right):k\in[n]\right\}
\end{equation}
\begin{equation}\label{eq:NN_CZ_basis8}
\left\{\mc Q^\T\left(\bm e_{\{k\}}^S+e_{\{k\}}^M\right):k\in[n]\right\}
\end{equation}
\begin{equation}\label{eq:NN_CZ_basis9}
\left\{\mc Q^\T\left(\bm e_{\{n,3\}}^S+\bm e_{X_nX_3}^{\mc G_o}+\bm e_{\{n,1,2,3\}}^M\right)\right\}.
\end{equation}
As a sanity check, these give $(n-1)+n+n+n+1=4n$ dimensions, which together with the $23n$ dimensions from $L_R^G$ gives the $27n$ dimensions for $L_R$.

To prove that Eqs.~\eqref{eq:NN_CZ_basis1},~\eqref{eq:NN_CZ_basis2},~\eqref{eq:NN_CZ_basis3},~\eqref{eq:NN_CZ_basis5},~\eqref{eq:NN_CZ_basis6},~\eqref{eq:NN_CZ_basis7},~\eqref{eq:NN_CZ_basis8},~\eqref{eq:NN_CZ_basis9} form a basis of $L_R$, again we may first perform an invertible linear transformation on them.
We change Eq.~\eqref{eq:NN_CZ_basis9} into
\begin{equation}\label{eq:NN_CZ_basis10}
\mc Q^\T\left(\bm e_{\{n,3\}}^S+\bm e_{X_nX_3}^{\mc G_o}+\bm e_{\{n,1,2,3\}}^M\right)-\mc Q^\T\left(\bm e_{\{n\}}^S+\bm e_{X_n}^{\mc G_o}+\bm e_{\{n,1\}}^M\right)-\mc Q^\T\left(\bm e_{\{3\}}^S+\bm e_{X_3}^{\mc G_o}+\bm e_{\{2,3\}}^M\right)=\bm\vartheta_{\{1,2\}}^M.
\end{equation}
We then change Eq.~\eqref{eq:NN_CZ_basis3} into Eq.~\eqref{eq:NN_CZ_basis4} and then change for example $\bm\vartheta_{X_{2k}}^{\mc G_o}-\bm\vartheta_{X_{2k+2}}^{\mc G_e}$ into
\begin{align}
&\bm\vartheta_{X_{2k}}^{\mc G_o}-\bm\vartheta_{X_{2k+2}}^{\mc G_e}-\mc Q^\T\left(\bm e_{\{2k\}}^S+\bm e_{X_{2k}}^{\mc G_o}+\bm e_{\{2k,2k+1\}}^M\right)+\mc Q^\T\left(\bm e_{\{2k+2\}}^S+\bm e_{X_{2k+2}}^{\mc G_e}+\bm e_{\{2k+1,2k+2\}}^M\right)\nonumber\\
&+\mc Q^\T\left(\bm e_{\{2k\}}^S+\bm e_{\{2k\}}^M\right)-\mc Q^\T\left(\bm e_{\{2k+2\}}^S+\bm e_{\{2k+2\}}^M\right)\nonumber\\
=&\bm\vartheta_{\{2k+1,2k+2\}}^M-\bm\vartheta_{\{2k,2k+1\}}^M,
\end{align}
and similarly for the $-\bm\vartheta_{X_{2k-1}}^{\mc G_e}+\bm\vartheta_{X_{2k+1}}^{\mc G_o}$s, except that $-\bm\vartheta_{X_1}^{\mc G_e}+\bm\vartheta_{X_3}^{\mc G_o}$ is changed into
\begin{equation}\label{eq:NN_CZ_basis11}
-\bm\vartheta_{X_1}^{\mc G_e}+\bm\vartheta_{X_3}^{\mc G_o}-\mc Q^\T\left(\bm e_{\{3\}}^S+\bm e_{X_3}^{\mc G_o}+\bm e_{\{2,3\}}^M\right)+\mc Q^\T\left(\bm e_{\{3\}}^S+\bm e_{\{3\}}^M\right)=-\bm\vartheta_{X_1}^{\mc G_e}-\bm\vartheta_{\{2\}}^M-\bm\vartheta_{\{2,3\}}^M.
\end{equation}
That is, Eq.~\eqref{eq:NN_CZ_basis3} is changed into
\begin{equation}\label{eq:NN_CZ_basis12}
\left\{\bm\vartheta_{\{k,k+1\}}^M-\bm\vartheta_{\{k+1,k+2\}}^M:k\in[n],k\neq1\right\}
\end{equation}
together with Eq.~\eqref{eq:NN_CZ_basis11}.

Hence now it suffices to prove that Eqs.~\eqref{eq:NN_CZ_basis2},~\eqref{eq:NN_CZ_basis1},~\eqref{eq:NN_CZ_basis5},~\eqref{eq:NN_CZ_basis6},~\eqref{eq:NN_CZ_basis11},~\eqref{eq:NN_CZ_basis7},~\eqref{eq:NN_CZ_basis8},~\eqref{eq:NN_CZ_basis12},~\eqref{eq:NN_CZ_basis10} form a basis for $L_R$.
Again we resort to Algorithm $3$.
This time the algorithm tries to find a basis for $L_R$, and initially $A=\{\bm\vartheta_\nu^S,\bm\vartheta_\nu^M,\bm\vartheta_a^{\mc G_e},\bm\vartheta_a^{\mc G_o}:\nu\in\Omega,a\sim\Omega\}$.
We first pick elements in Eqs.~\eqref{eq:NN_CZ_basis2},~\eqref{eq:NN_CZ_basis1} into $C$, using the same procedure described previously.
Now $A=\{\bm\vartheta_\nu^S,\bm\vartheta_\nu^M,\bm\vartheta_{X_k}^{\mc G_e},\bm\vartheta_{X_k}^{\mc G_o}:\nu\in\Omega,k\in[n]\}$.

Then put elements in Eqs.~\eqref{eq:NN_CZ_basis5},~\eqref{eq:NN_CZ_basis6} into $C$.
For example, when $\mc Q^\T\left(\bm e_{\{2k-1\}}^S+\bm e_{X_{2k-1}}^{\mc G_e}+\bm e_{\{2k-1,2k\}}^M\right)$ is picked, choose $\bm \vartheta=\bm\vartheta_{X_{2k-1}}^{\mc G_e}$.
Afterward, $A=\{\bm\vartheta_\nu^S,\bm\vartheta_\nu^M:\nu\in\Omega\}\cup\{\bm\vartheta_{X_1}^{\mc G_e}\}$.

Next, put Eq.~\eqref{eq:NN_CZ_basis11} into $C$, picking $\bm\vartheta=\bm\vartheta_{X_1}^{\mc G_e}$.
Now $A=\{\bm\vartheta_\nu^S,\bm\vartheta_\nu^M:\nu\in\Omega\}$.

Then put elements in Eq.~\eqref{eq:NN_CZ_basis7} into $C$, and when $\mc Q^\T\left(\bm e_{\{k,k+1\}}^S+\bm e_{\{k,k+1\}}^M\right)$ is picked, set $\bm\vartheta$ to be $\bm\vartheta_{\{k,k+1\}}^S$.
What are left in $A$ at this time are $A=\{\bm\vartheta_{\{k\}}^S:k\in[n]\}\cup\{\bm\vartheta_\nu^M:\nu\sim\Omega\}$.

Now put elements in Eq.~\eqref{eq:NN_CZ_basis8} into $C$, this time when $\mc Q^\T\left(\bm e_{\{k\}}^S+e_{\{k\}}^M\right)$ is picked, let $\bm\vartheta=\bm\vartheta_{\{k\}}^S$.
At this time $A=\{\bm\vartheta_\nu^M:\nu\sim\Omega\}$.

Then put elements in Eq.~\eqref{eq:NN_CZ_basis12} as $\bm c$ and put into $C$, in the increasing order of $k$.
When $\bm\vartheta_{\{k,k+1\}}^M-\bm\vartheta_{\{k+1,k+2\}}^M$ is picked, use $\bm\vartheta=\bm\vartheta_{\{k,k+1\}}^M$.
Now $A$ is left with $\{\bm\vartheta_{\{k\}}^M:k\in[n]\}\cup\{\bm\vartheta_{\{1,2\}}^M\}$.

Finally put Eq.~\eqref{eq:NN_CZ_basis10} into $C$ and let $\bm\vartheta=\bm\vartheta_{\{1,2\}}^M$.
$A$ is now $\{\bm\vartheta_{\{k\}}^M:k\in[n]\}$.
Since each $\bm\vartheta_{\{k\}}^M$ changes independently under single-qubit depolarization gauges, we conclude that $\mr{span}(A)\cap L_R=\varnothing$ and Algorithm $3$ terminates.
By the correctness of the Algorithm $3$, by this time $C$, which contains Eqs.~\eqref{eq:NN_CZ_basis2},~\eqref{eq:NN_CZ_basis1},~\eqref{eq:NN_CZ_basis5},~\eqref{eq:NN_CZ_basis6},~\eqref{eq:NN_CZ_basis11},~\eqref{eq:NN_CZ_basis7},~\eqref{eq:NN_CZ_basis8},~\eqref{eq:NN_CZ_basis12},~\eqref{eq:NN_CZ_basis10}, form a basis for $L_R$.
Thus Eqs.~\eqref{eq:NN_CZ_basis1},~\eqref{eq:NN_CZ_basis2},\\\eqref{eq:NN_CZ_basis3},~\eqref{eq:NN_CZ_basis5},~\eqref{eq:NN_CZ_basis6},~\eqref{eq:NN_CZ_basis7},~\eqref{eq:NN_CZ_basis8},~\eqref{eq:NN_CZ_basis9} form a basis for $L_R$.

\bibliography{refs.bib}
\clearpage

\end{document}